\documentclass{article}
 \usepackage{amssymb,nicefrac}
\usepackage{amsfonts}
\usepackage{graphicx}
\usepackage[fleqn]{amsmath}
\usepackage[varg]{txfonts}
\usepackage{stmaryrd}

\usepackage{framed}
\usepackage{color}
\usepackage[normalem]{ulem}
\usepackage{tikzfig}
\usepackage{placeins}

\usepackage{blkarray}

\usepackage{mathtools} 
\DeclarePairedDelimiter\bra{\langle}{\rvert}
\DeclarePairedDelimiter\ket{\lvert}{\rangle}

\newenvironment{proof}{\textbf{Proof:}}{\hfill$\Box$\newline}

\tikzstyle{env}=[copoint,regular polygon rotate=0,minimum width=0.2cm, fill=black]

\tikzstyle{every picture}=[baseline=-0.25em]
\tikzstyle{dotpic}=[scale=0.5]
\tikzstyle{diredges}=[every to/.style={diredge}]
\tikzstyle{dot graph}=[shorten <=-0.1mm,shorten >=-0.1mm,scale=0.6]
\tikzstyle{plot point}=[circle,fill=black,minimum width=2mm,inner sep=0]

\tikzstyle{braceedge}=[decorate,decoration={brace,amplitude=2mm,raise=-1mm}]
\tikzstyle{small braceedge}=[decorate,decoration={brace,amplitude=1mm,raise=-1mm}]
\tikzstyle{left hook arrow}=[left hook-latex]
\tikzstyle{right hook arrow}=[right hook-latex]

\tikzstyle{dtriangle}=[fill=yellow,draw=black,shape=isosceles triangle,shape border rotate=-90,isosceles triangle stretches=true,inner sep=0.8pt,minimum width=0.25cm,minimum height=2mm]
\tikzstyle{vtriang}=[fill=yellow,draw=black,shape=isosceles triangle,shape border rotate=180,isosceles triangle stretches=true,inner sep=0.8pt,minimum width=0.25cm,minimum height=2mm]
\tikzstyle{trigmc}=[fill=green,draw=black,shape=isosceles triangle,shape border rotate=90,isosceles triangle stretches=true,inner sep=0.8pt,minimum width=0.3cm,minimum height=2mm]
\tikzstyle{vrt}=[fill=yellow,draw=black,shape=isosceles triangle,shape border rotate=0,isosceles triangle stretches=true,inner sep=0.8pt,minimum width=0.25cm,minimum height=2mm]
\tikzstyle{H box}=[rectangle,fill=yellow,draw=black,xscale=0.8,yscale=0.8, inner sep=0.6pt]
\tikzstyle{gbox}=[rectangle,fill=green,draw=black,xscale=1.0,yscale=1.0, inner sep=1.pt]
\tikzstyle{rbox}=[rectangle,fill=red,draw=black,xscale=1.0,yscale=1.0, inner sep=1.pt]
\tikzstyle{zhbx}=[rectangle,fill=white,draw=black,xscale=1.0,yscale=1.0, inner sep=1.6pt]
\tikzstyle{newh}=[rectangle,fill=yellow,draw=black,xscale=2.0,yscale=2.0, inner sep=1.6pt]
\tikzstyle{triangle}=[fill=yellow,draw=black,shape=isosceles triangle,shape border rotate=90,isosceles triangle stretches=true,inner sep=0.8pt,minimum width=0.25cm,minimum height=2mm]

\tikzstyle{bn}=[circle,fill=black,draw=black,scale=.4]
\tikzstyle{wn}=[circle,fill=white,draw=black,scale=.6]
\tikzstyle{dn}=[circle,fill=none,draw=gray]
\tikzstyle{bspider}=[fill=black,draw=black,scale=1,shape=isosceles triangle,shape border rotate=-90,isosceles triangle stretches=true,inner sep=1pt,minimum width=0.4cm,minimum height=3mm]
\tikzstyle{dbspider}=[fill=black,draw=black,scale=1,shape=isosceles triangle,shape border rotate=90,isosceles triangle stretches=true,inner sep=1pt,minimum width=0.4cm,minimum height=3mm]
\tikzstyle{L}=[rectangle,shape=rectangle,fill=green,draw=black]

\tikzstyle{Z dot}=[inner sep=0mm, minimum size=2mm, shape=circle, draw=black, fill={rgb,255: red,221; green,255; blue,221}]
\tikzstyle{Z phase dot}=[minimum size=5mm, font={\footnotesize\boldmath}, shape=rectangle, rounded corners=2mm, inner sep=0.2mm, outer sep=-2mm, scale=0.8, draw=black, fill={rgb,255: red,221; green,255; blue,221}]
\tikzstyle{X dot}=[Z dot, shape=circle, draw=black, fill={rgb,255: red,255; green,136; blue,136}]
\tikzstyle{X phase dot}=[Z phase dot, fill={rgb,255: red,255; green,136; blue,136}, font={\footnotesize\boldmath}]
\tikzstyle{hadamard edge}=[-, dashed, dash pattern=on 2pt off 0.5pt, thick, draw={rgb,255: red,68; green,136; blue,255}]

\tikzstyle{black dot}=[inner sep=0.7mm,minimum width=0pt,minimum height=0pt,fill=black,draw=black,shape=circle]

\tikzstyle{dot}=[black dot]
\tikzstyle{smalldot}=[inner sep=0.4mm,minimum width=0pt,minimum height=0pt,fill=black,draw=black,shape=circle]
\tikzstyle{white dot}=[dot,fill=white]
\tikzstyle{antipode}=[white dot,inner sep=0.3mm,font=\footnotesize]
\tikzstyle{smallwhitedot}=[smalldot,fill=white]
\tikzstyle{alt white dot}=[white dot,label={[xshift=3.07mm,yshift=-0.05mm,font=\footnotesize]left:$*$}]
\tikzstyle{gray dot}=[dot,fill=gray!40!white]
\tikzstyle{smallgraydot}=[smalldot,fill=gray!40!white]
\tikzstyle{box vertex}=[draw=black,rectangle]
\tikzstyle{small box}=[box vertex,fill=white]
\tikzstyle{whitebg}=[fill=white,inner sep=2pt]
\tikzstyle{graph state vertex}=[sg vertex,fill=black]

\tikzstyle{wide copoint}=[fill=white,draw=black,shape=isosceles triangle,shape border rotate=90,isosceles triangle stretches=true,inner sep=1pt,minimum width=1.5cm,minimum height=5mm]
\tikzstyle{wide point}=[fill=white,draw=black,shape=isosceles triangle,shape border rotate=-90,isosceles triangle stretches=true,inner sep=1pt,minimum width=1.5cm,minimum height=4mm]
\tikzstyle{very wide copoint}=[fill=white,draw=black,shape=isosceles triangle,shape border rotate=-90,isosceles triangle stretches=true,inner sep=1pt,minimum width=2.5cm,minimum height=4mm]
\tikzstyle{very wide empty copoint}=[draw=black,shape=isosceles triangle,shape border rotate=-90,isosceles triangle stretches=true,inner sep=1pt,minimum width=2.5cm,minimum height=4mm]
\tikzstyle{symm}=[ultra thick,shorten <=-1mm,shorten >=-1mm]

\tikzstyle{square box}=[rectangle,fill=white,draw=black,minimum height=5mm,minimum width=5mm,font=\small]
\tikzstyle{square gray box}=[rectangle,fill=gray!30,draw=black,minimum height=6mm,minimum width=6mm]
\tikzstyle{copoint}=[regular polygon,regular polygon sides=3,draw=black,scale=0.75,inner sep=-0.5pt,minimum width=7mm,fill=white]
\tikzstyle{point}=[regular polygon,regular polygon sides=3,draw=black,scale=0.75,inner sep=-0.5pt,minimum width=7mm,fill=white,regular polygon rotate=180]
\tikzstyle{gray point}=[point,fill=gray!40!white]
\tikzstyle{gray copoint}=[copoint,fill=gray!40!white]

\newcommand{\edgearrow}{{\arrow[black]{>}}}
\newcommand{\edgetick}{{\arrow[black,scale=0.7,very thick]{|}}}

\tikzstyle{diredge}=[->]
\tikzstyle{rdiredge}=[<-]
\tikzstyle{medium diredge}=[->]

\tikzstyle{short diredge}=[->]
\tikzstyle{halfedge}=[-)]
\tikzstyle{other halfedge}=[(-]
\tikzstyle{freeedge}=[(-)]
\tikzstyle{white edge}=[line width=5pt,white]
\tikzstyle{tick}=[postaction=decorate,decoration={markings, mark=at position 0.5 with \edgetick}]
\tikzstyle{small map edge}=[|-latex, gray!60!blue, shorten <=0.9mm, shorten >=0.5mm]
\tikzstyle{thick dashed edge}=[very thick,dashed,gray!40]
\tikzstyle{map edge}=[|-latex,very thick, gray!40, shorten <=1mm, shorten >=0.5mm]
\tikzstyle{tickedge}=[postaction=decorate,
  decoration={markings, mark=at position 0.5 with \edgetick}]
\tikzstyle{dirtickedge}=[postaction=decorate,
  decoration={markings, mark=at position 0.5 with \edgetick},
  decoration={markings, mark=at position 0.85 with \edgearrow}]
\tikzstyle{dirdoubletickedge}=[postaction=decorate,
  decoration={markings, mark=at position 0.4 with \edgetick},
  decoration={markings, mark=at position 0.6 with \edgetick},
  decoration={markings, mark=at position 0.85 with \edgearrow}]

\makeatletter
\newcommand{\boxshape}[3]{%
\pgfdeclareshape{#1}{
\inheritsavedanchors[from=rectangle] 
\inheritanchorborder[from=rectangle]
\inheritanchor[from=rectangle]{center}
\inheritanchor[from=rectangle]{north}
\inheritanchor[from=rectangle]{south}
\inheritanchor[from=rectangle]{west}
\inheritanchor[from=rectangle]{east}
\backgroundpath{
\southwest \pgf@xa=\pgf@x \pgf@ya=\pgf@y
\northeast \pgf@xb=\pgf@x \pgf@yb=\pgf@y

\@tempdima=#2
\@tempdimb=#3

\pgfpathmoveto{\pgfpoint{\pgf@xa - 5pt + \@tempdima}{\pgf@ya}}
\pgfpathlineto{\pgfpoint{\pgf@xa - 5pt - \@tempdima}{\pgf@yb}}
\pgfpathlineto{\pgfpoint{\pgf@xb + 5pt + \@tempdimb}{\pgf@yb}}
\pgfpathlineto{\pgfpoint{\pgf@xb + 5pt - \@tempdimb}{\pgf@ya}}
\pgfpathlineto{\pgfpoint{\pgf@xa - 5pt + \@tempdima}{\pgf@ya}}
\pgfpathclose
}
}}

\boxshape{NEbox}{0pt}{8pt}
\boxshape{SEbox}{0pt}{-8pt}
\boxshape{NWbox}{8pt}{0pt}
\boxshape{SWbox}{-8pt}{0pt}
\makeatother

\tikzstyle{map}=[draw,shape=NEbox,inner sep=7pt]
\tikzstyle{mapdag}=[draw,shape=SEbox,inner sep=7pt]
\tikzstyle{maptrans}=[draw,shape=SWbox,inner sep=7pt]
\tikzstyle{mapconj}=[draw,shape=NWbox,inner sep=7pt]

\tikzstyle{probs}=[shape=semicircle,fill=gray!40!white,draw=black,shape border rotate=180,minimum width=1.2cm]

\tikzstyle{arrs}=[-latex,font=\small,auto]
\tikzstyle{arrow plain}=[arrs]
\tikzstyle{arrow dashed}=[dashed,arrs]
\tikzstyle{arrow bold}=[very thick,arrs]
\tikzstyle{arrow hide}=[draw=white!0,-]
\tikzstyle{arrow reverse}=[latex-]
\tikzstyle{cdnode}=[]

\tikzstyle{gn}=[dot,fill=green,minimum width=0.25cm,inner sep=0pt]
\tikzstyle{rno}=[dot,fill=red,inner sep=0pt,minimum width=0.25cm]
\tikzstyle{rn}=[dot,fill=pink,inner sep=0pt,minimum width=0.25cm]

\tikzstyle{rc}=[dot,thick,fill=white,draw = red,minimum width=0.3cm,inner sep=0pt]
\tikzstyle{gc}=[dot,thick,fill=white,draw= green,inner sep=0pt,minimum width=0.3cm]
\tikzstyle{bc}=[dot,thick,fill=white,draw= blue,minimum width=0.3cm]

\tikzstyle{label}=[circle,fill=white,minimum width=0.3cm]

\tikzstyle{clocklabel}=[dot,fill=yellow,draw=black,font=\tiny,inner sep=0.75pt]

\tikzstyle{rsn}=[circle split,draw,fill=red,font=\tiny,inner sep=0.75pt]
\tikzstyle{gsn}=[circle split,draw,fill=green,font=\tiny,inner sep=0.75pt]
\tikzstyle{bsn}=[circle split,draw,fill=blue,font=\tiny,inner sep=0.75pt]

\tikzstyle{rsc}=[circle split,thick,draw= red,draw,fill=white,font=\tiny,inner sep=0.75pt]
\tikzstyle{gsc}=[circle split,thick,draw= green,draw,fill=white,font=\tiny,inner sep=0.75pt]
\tikzstyle{bsc}=[circle split,thick,draw= blue,draw,fill=white,font=\tiny,inner sep=0.75pt]

\tikzstyle{cnot}=[fill=white,shape=circle,inner sep=-1.4pt]
\tikzstyle{wire label}=[font=\tiny, auto]

\tikzstyle{cdiag}=[matrix of math nodes, row sep=3em, column sep=3em, text height=1.5ex, text depth=0.25ex,inner sep=0.5em]
\tikzstyle{arrow above}=[transform canvas={yshift=0.5ex}]
\tikzstyle{arrow below}=[transform canvas={yshift=-0.5ex}]

\newtheorem{Th}{Theorem}[section]
\newtheorem{theorem}[Th]{Theorem}
\newtheorem{proposition}[Th]{Proposition} 
\newtheorem{lemma}[Th]{Lemma}
\newtheorem{corollary}[Th]{Corollary}

\newtheorem{remark}[Th]{Remark}

\makeatletter
\newcommand{\vast}{\bBigg@{6.5}}
\makeatother

\newcommand{\vertrule}[1][1ex]{\rule{.4pt}{#1}}

\title{ Algebraic complete axiomatisation of ZX-calculus with a normal form via elementary matrix operations}
\author{Quanlong Wang\\
Cambridge Quantum Computing Ltd.}
\begin{document}
\date{}\maketitle
\begin{abstract}
In this paper we give a complete axiomatisation of qubit ZX-calculus via elementary transformations which are basic operations in linear algebra. This formalism has two main advantages. First, all the operations of the phases are algebraic ones without trigonometry functions involved, thus paved the way for generalising complete axiomatisation of qubit ZX-calculus to qudit ZX-calculus and ZX-calculus over commutative semirings. Second, we characterise elementary transformations in terms of ZX diagrams, so a lot of linear algebra stuff 
can be done purely diagrammatically.  
\end{abstract}


  \section{ Introduction}

 
ZX-calculus was introduced by Coecke and Duncan \cite{CoeckeDuncan} as a graphical language for quantum computing. It is quite intuitive but still mathematically restrict as formalised in the framework of compact closed categories. The diagrams of ZX-calculus are mainly generated by so-called Z-spiders and X-spiders (based on Z and X quantum observables with a spider-like shape).  Any operation performed on a ZX diagram is just a replacement of a part of the diagram with another diagram according to a diagrammatic equality. Such  an operation is called a rewriting which is local in the sense that other parts of the operated diagram won't be affected by the rewriting. The property of being local can be seen as an advantage of ZX-calculus, since one has no need to remember any previous steps when rewriting a diagram.
 
 Since its invention, ZX-calculus has been focused on unitaries which are the core of quantum computing \cite{Nielsen}. Actually, the first proof of universality of ZX-calculus (which means each matrix of size $2^m \times 2^n$ can be represented by a ZX diagram) was based on representation of unitaries in terms of Z and X phases \cite{CoeckeDuncan}.  Moreover, ZX-calculus has exhibited its power in the application field of quantum circuit optimisation \cite{debeaudrapbianwang, KissingerG20, Cowtan_2020}.  On the other hand, ZX-calculus has been first proved to be complete for overall qubit quantum computing in \cite{ngwang} and incorporated in \cite{amarngwang}, which means quantum computation done by matrices can purely be done in ZX-calculus. Afterwards, there came a few different complete axiomatisations of ZX-calculus  \cite{jpvbeyondlics, jpvnormfmlics, Renaudprulelics}.  However, all of these  universal axiomatisations are non-algebraic in the sense that there are always some rules with trigonometry functions involved, which makes the application of non-algebraic rules quite difficult. This inconvenience was rescued in \cite{wangalg2020} by introducing an algebraic complete axiomatisation of ZX-calculus. Nevertheless, 
  except for \cite{jpvnormfmlics} whose completeness was established internally via a normal form of ZX diagrams, the other complete axiomatisations obtained their completeness externally, either deployed the completeness of ZW-calculus \cite{amarngwang}, or utilised pre-existing completeness of ZX-calculus. Could there be a universally complete axiomatisation of ZX-calculus which is both algebraic and internally established?   
 
 In this paper, we give a definite answer to that question by providing an algebraic complete axiomatisation of ZX-calculus with a normal form via elementary matrix operations. The idea of using elementary matrices comes from a simple observation: all of the elementary matrices of size $2\times 2$ already existed in the previous complete  axiomatisations of ZX-calculus \cite{amarngwang} and \cite{wangalg2020}. So it is natural to generalise these expressions for  arbitrary elementary matrix of size $2^n\times 2^n$. To focus on the completeness proof, the representation of arbitrary elementary row switching operations
 in ZX-calculus as presented in \cite{qwangslides} is not given in this paper,  which will be included in the further work \cite{wanglinearstm}. Based on the representation of elementary matrices in ZX, we obtain a normal form for any vectors, then any matrix can be represented by a diagram due to the map-state duality. This normal form leads to a proof of completeness, though full of nontrivial techniques. In addition, we gain an advantage that no scalable techniques are needed, while seriously considered in  \cite{tdsscalarble}.

  \section{ Algebraic ZX-calculus}
  In this section, we give generators and rules for algebraic ZX-calculus which is shown to be complete via a normal form. We assume that we are working in a compact closed category $\mathfrak{C}$.
  \begin{table}[!h]
\begin{center} 
\begin{tabular}{|r@{~}r@{~}c@{~}c|r@{~}r@{~}c@{~}c|}
\hline
$R_{Z,a}^{(n,m)}$&$:$&$n\to m$ & %
	\beginpgfgraphicnamed{TikZit//generalgreenspider}
	\InputIfFileExists{TikZit//generalgreenspider.tikz}{}{\input{./figures/TikZit//generalgreenspider.tikz}}%
	\endpgfgraphicnamed
  &  $\mathbb I$&$:$&$1\to 1$&%
	\beginpgfgraphicnamed{TikZit//Id}
	\InputIfFileExists{TikZit//Id.tikz}{}{\input{./figures/TikZit//Id.tikz}}%
	\endpgfgraphicnamed
 \\\hline
$H$&$:$&$1\to 1$ &%
	\beginpgfgraphicnamed{TikZit//newhadamard}
	\InputIfFileExists{TikZit//newhadamard.tikz}{}{\input{./figures/TikZit//newhadamard.tikz}}%
	\endpgfgraphicnamed
&  $\sigma$&$:$&$ 2\to 2$& %
	\beginpgfgraphicnamed{TikZit//swap}
	\InputIfFileExists{TikZit//swap.tikz}{}{\input{./figures/TikZit//swap.tikz}}%
	\endpgfgraphicnamed
\\\hline
   $C_a$&$:$&$ 0\to 2$& %
	\beginpgfgraphicnamed{TikZit//cap}
	\InputIfFileExists{TikZit//cap.tikz}{}{\input{./figures/TikZit//cap.tikz}}%
	\endpgfgraphicnamed
 &$ C_u$&$:$&$ 2\to 0$&%
	\beginpgfgraphicnamed{TikZit//cup}
	\InputIfFileExists{TikZit//cup.tikz}{}{\input{./figures/TikZit//cup.tikz}}%
	\endpgfgraphicnamed
 \\\hline
  $T$&$:$&$1\to 1$&%
	\beginpgfgraphicnamed{TikZit//triangle}
	\InputIfFileExists{TikZit//triangle.tikz}{}{\input{./figures/TikZit//triangle.tikz}}%
	\endpgfgraphicnamed
  & $T^{-1}$&$:$&$1\to 1$&%
	\beginpgfgraphicnamed{TikZit//triangleinv}
	\InputIfFileExists{TikZit//triangleinv.tikz}{}{\input{./figures/TikZit//triangleinv.tikz}}%
	\endpgfgraphicnamed
 \\\hline
\end{tabular}\caption{Generators of algebraic ZX-calculus, where $m,n\in \mathbb N$, $ a  \in \mathbb C$.} \label{qbzxgenerator} 
\end{center}
\end{table}
\FloatBarrier
\begin{remark}
In Table \ref{qbzxgenerator}, $C_a$ and $ C_u$ compose the compact structure of $\mathfrak{C}$, while $\sigma$ comprises  the symmetric structure of $\mathfrak{C}$.
  \end{remark}

 For simplicity, we make the following conventions: 
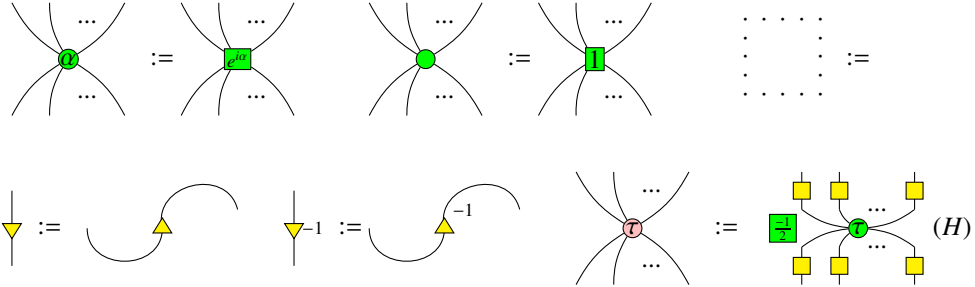
\[
	\beginpgfgraphicnamed{TikZit//spider0denote3}
	\InputIfFileExists{TikZit//spider0denote3.tikz}{}{\input{./figures/TikZit//spider0denote3.tikz}}%
	\endpgfgraphicnamed
 
\]
 where $\alpha \in \mathbb R, \tau \in \{ 0, \pi \}$.
\begin{remark}\label{finitegeneratorsrk}
The generator $R_{Z,a}^{(n,m)}$ can be equivalently  expressed in terms of other generators and the following three diagrams:
 \begin{equation}\label{finitegenerators}
	\beginpgfgraphicnamed{TikZit//greencopy}
	\begin{tikzpicture}
	\begin{pgfonlayer}{nodelayer}
		\node [style=gn] (0) at (0, 0) {};
		\node [style=none] (1) at (0, 0.25) {};
		\node [style=none] (2) at (-0.25, -0.25) {};
		\node [style=none] (3) at (0.25, -0.25) {};
	\end{pgfonlayer}
	\begin{pgfonlayer}{edgelayer}
		\draw (0) to (1.center);
		\draw (0) to (2.center);
		\draw (0) to (3.center);
	\end{pgfonlayer}
\end{tikzpicture}}%
	\endpgfgraphicnamed
 \quad  %
	\beginpgfgraphicnamed{TikZit//greencodot}
	\begin{tikzpicture}
	\begin{pgfonlayer}{nodelayer}
		\node [style=none] (0) at (0, 0.25) {};
		\node [style=gn] (1) at (0, 0) {};
	\end{pgfonlayer}
	\begin{pgfonlayer}{edgelayer}
		\draw (1) to (0.center);
	\end{pgfonlayer}
\end{tikzpicture}}%
	\endpgfgraphicnamed
 \quad  %
	\beginpgfgraphicnamed{TikZit//greenbxa}
	\begin{tikzpicture}
	\begin{pgfonlayer}{nodelayer}
		\node [style=none] (0) at (0, 0.25) {};
		\node [style=gbox] (1) at (0, 0) {$a$};
		\node [style=none] (2) at (0, -0.25) {};
	\end{pgfonlayer}
	\begin{pgfonlayer}{edgelayer}
		\draw (0.center) to (2.center);
	\end{pgfonlayer}
\end{tikzpicture}}%
	\endpgfgraphicnamed
 
 \end{equation}
\end{remark}

\begin{remark}
Comparing to the original generators of ZX-calculus as in \cite{CoeckeDuncan}, the generators $R_{Z,a}^{(n,m)}$, $T$ and  $T^{-1}$ are new ones, but they have been firstly introduced in  \cite{ngwang2}, and  essentially shown to be representable in terms of original generators  in  \cite{ngwang} and  \cite{amarngwang}. The $H$ node has a slight difference with the normal Hadamard node in a global scalar, which leads to the X spider defined in (H) distinguished from the normal X spider by a global scalar, in addition that only X-phase angles of $0$ and $\pi$ are defined now.  
  \end{remark}

There is a standard interpretation $\left\llbracket \cdot \right\rrbracket$ for the ZX diagrams:
\[
\left\llbracket %
	\beginpgfgraphicnamed{TikZit//generalgreenspider}
	\InputIfFileExists{TikZit//generalgreenspider.tikz}{}{\input{./figures/TikZit//generalgreenspider.tikz}}%
	\endpgfgraphicnamed
 \right\rrbracket=\ket{0}^{\otimes m}\bra{0}^{\otimes n}+a\ket{1}^{\otimes m}\bra{1}^{\otimes n},
\left\llbracket %
	\beginpgfgraphicnamed{TikZit//redspider0p}
	\InputIfFileExists{TikZit//redspider0p.tikz}{}{\input{./figures/TikZit//redspider0p.tikz}}%
	\endpgfgraphicnamed
 \right\rrbracket=\sum_{\substack{0\leq i_1, \cdots, i_m,  j_1, \cdots, j_n\leq 1\\ i_1+\cdots+ i_m\equiv  j_1+\cdots +j_n(mod~ 2)}}\ket{i_1, \cdots, i_m}\bra{j_1, \cdots, j_n},
\]
\[
\left\llbracket%
	\beginpgfgraphicnamed{TikZit//newhadamard}
	\begin{tikzpicture}
	\begin{pgfonlayer}{nodelayer}
		\node [style=newh] (0) at (0, 0) {};
		\node [style=none] (1) at (0, 0.5) {};
		\node [style=none] (2) at (0, -0.5) {};
	\end{pgfonlayer}
	\begin{pgfonlayer}{edgelayer}
		\draw (1.center) to (2.center);
	\end{pgfonlayer}
\end{tikzpicture}}%
	\endpgfgraphicnamed
\right\rrbracket=\begin{pmatrix}
        1 & 1 \\
        1 & -1
 \end{pmatrix}, \quad \left\llbracket%
	\beginpgfgraphicnamed{TikZit//triangle}
	\begin{tikzpicture}
	\begin{pgfonlayer}{nodelayer}
		\node [style=none] (0) at (0, 0.5) {};
		\node [style=triangle] (1) at (0, 0) {};
		\node [style=none] (2) at (0, -0.5) {};
	\end{pgfonlayer}
	\begin{pgfonlayer}{edgelayer}
		\draw (0.center) to (2.center);
	\end{pgfonlayer}
\end{tikzpicture}}%
	\endpgfgraphicnamed
\right\rrbracket=\begin{pmatrix}
        1 & 1 \\
        0 & 1
 \end{pmatrix}, \quad \quad
  \left\llbracket%
	\beginpgfgraphicnamed{TikZit//triangleinv}
	\begin{tikzpicture}
	\begin{pgfonlayer}{nodelayer}
		\node [style=none] (0) at (0.25, 0.25) {-{\scriptsize1}};
		\node [style=triangle] (1) at (0, 0) {};
		\node [style=none] (2) at (0, -0.5) {};
		\node [style=none] (3) at (0, 0.5) {};
	\end{pgfonlayer}
	\begin{pgfonlayer}{edgelayer}
		\draw (3.center) to (2.center);
	\end{pgfonlayer}
\end{tikzpicture}}%
	\endpgfgraphicnamed
\right\rrbracket=\begin{pmatrix}
        1 & -1 \\
        0 & 1
 \end{pmatrix}, \quad
\left\llbracket%
	\beginpgfgraphicnamed{TikZit//singleredpi}
	\begin{tikzpicture}
	\begin{pgfonlayer}{nodelayer}
		\node [style=none] (0) at (0, 0.5) {};
		\node [style=rn] (1) at (0, 0) {$\pi$};
		\node [style=none] (2) at (0, -0.5) {};
	\end{pgfonlayer}
	\begin{pgfonlayer}{edgelayer}
		\draw (0.center) to (2.center);
	\end{pgfonlayer}
\end{tikzpicture}}%
	\endpgfgraphicnamed
\right\rrbracket=\begin{pmatrix}
        0 & 1 \\
        1 & 0
 \end{pmatrix}, \quad
\left\llbracket%
	\beginpgfgraphicnamed{TikZit//Id}
	\begin{tikzpicture}
	\begin{pgfonlayer}{nodelayer}
		\node [style=none] (1) at (0.5, 0.3) {};
		\node [style=none] (2) at (0.5, -0.3) {};
		\node [style=none] (3) at (0.5, -0.5) {};
		\node [style=none] (4) at (0.5, 0.5) {};
	\end{pgfonlayer}
	\begin{pgfonlayer}{edgelayer}
		\draw (1.center) to (2.center);
	\end{pgfonlayer}
\end{tikzpicture}}%
	\endpgfgraphicnamed
\right\rrbracket=\begin{pmatrix}
        1 & 0 \\
        0 & 1
 \end{pmatrix}, 
  \]

\[
 \left\llbracket%
	\beginpgfgraphicnamed{TikZit//swap}
	\InputIfFileExists{TikZit//swap.tikz}{}{\input{./figures/TikZit//swap.tikz}}%
	\endpgfgraphicnamed
\right\rrbracket=\begin{pmatrix}
        1 & 0 & 0 & 0 \\
        0 & 0 & 1 & 0 \\
        0 & 1 & 0 & 0 \\
        0 & 0 & 0 & 1 
 \end{pmatrix}, \quad
  \left\llbracket%
	\beginpgfgraphicnamed{TikZit//cap}
	\begin{tikzpicture}
	\begin{pgfonlayer}{nodelayer}
		\node [style=none] (0) at (0, -0) {};
		\node [style=none] (1) at (1, -0) {};
	\end{pgfonlayer}
	\begin{pgfonlayer}{edgelayer}
		\draw [bend left=90, looseness=1.50] (0.center) to (1.center);
	\end{pgfonlayer}
\end{tikzpicture}}%
	\endpgfgraphicnamed
\right\rrbracket=\begin{pmatrix}
        1  \\
        0  \\
        0  \\
        1  \\
 \end{pmatrix}, \quad
   \left\llbracket%
	\beginpgfgraphicnamed{TikZit//cup}
	\begin{tikzpicture}
	\begin{pgfonlayer}{nodelayer}
		\node [style=none] (0) at (0, 0.5) {};
		\node [style=none] (1) at (1, 0.5) {};
	\end{pgfonlayer}
	\begin{pgfonlayer}{edgelayer}
		\draw [bend right=90, looseness=1.50] (0.center) to (1.center);
	\end{pgfonlayer}
\end{tikzpicture}}%
	\endpgfgraphicnamed
\right\rrbracket=\begin{pmatrix}
        1 & 0 & 0 & 1 
         \end{pmatrix}, 
 \quad
  \left\llbracket%
	\beginpgfgraphicnamed{TikZit//emptysquare}
	\InputIfFileExists{TikZit//emptysquare.tikz}{}{\input{./figures/TikZit//emptysquare.tikz}}%
	\endpgfgraphicnamed
\right\rrbracket=1,  
   \]

\[  \llbracket D_1\otimes D_2  \rrbracket =  \llbracket D_1  \rrbracket \otimes  \llbracket  D_2  \rrbracket, \quad 
 \llbracket D_1\circ D_2  \rrbracket =  \llbracket D_1  \rrbracket \circ  \llbracket  D_2  \rrbracket,
  \]
where 
$$ a  \in \mathbb C, \quad \ket{0}= \begin{pmatrix}
        1  \\
        0  \\
 \end{pmatrix}, \quad 
 \bra{0}=\begin{pmatrix}
        1 & 0 
         \end{pmatrix},
 \quad  \ket{1}= \begin{pmatrix}
        0  \\
        1  \\
 \end{pmatrix}, \quad 
  \bra{1}=\begin{pmatrix}
     0 & 1 
         \end{pmatrix}.
 $$

 \begin{figure}[!h]
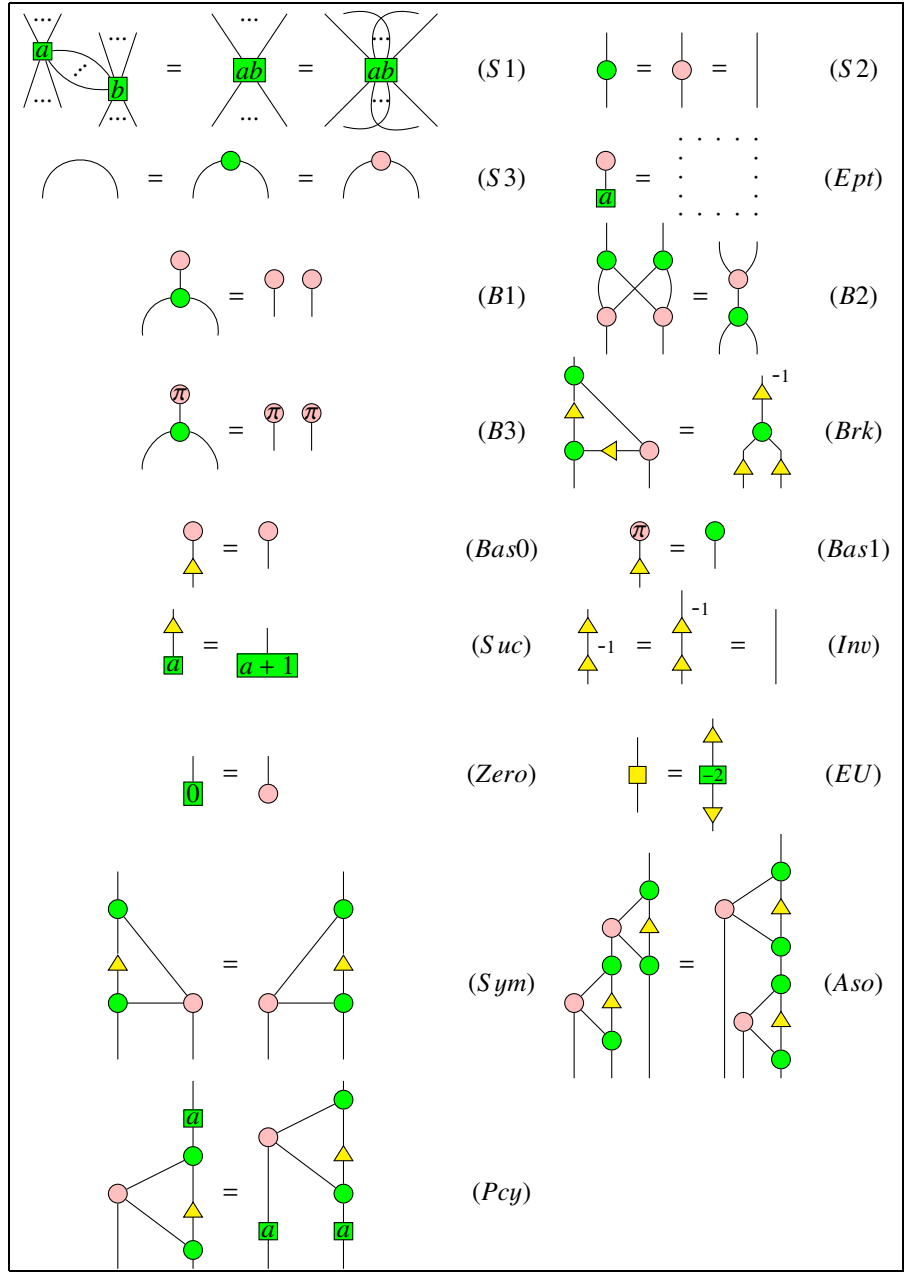

\begin{center}
\[
\quad \qquad\begin{array}{|cccc|}
\hline
	\beginpgfgraphicnamed{TikZit//generalgreenspiderfusesym}
	\InputIfFileExists{TikZit//generalgreenspiderfusesym.tikz}{}{\input{./figures/TikZit//generalgreenspiderfusesym.tikz}}%
	\endpgfgraphicnamed
&(S1) &%
	\beginpgfgraphicnamed{TikZit//s2new2}
	\InputIfFileExists{TikZit//s2new2.tikz}{}{\input{./figures/TikZit//s2new2.tikz}}%
	\endpgfgraphicnamed
 &(S2)\\
	\beginpgfgraphicnamed{TikZit//induced_compact_structure}
	\InputIfFileExists{TikZit//induced_compact_structure.tikz}{}{\input{./figures/TikZit//induced_compact_structure.tikz}}%
	\endpgfgraphicnamed
&(S3) & %
	\beginpgfgraphicnamed{TikZit//rdotaempty}
	\InputIfFileExists{TikZit//rdotaempty.tikz}{}{\input{./figures/TikZit//rdotaempty.tikz}}%
	\endpgfgraphicnamed
  &(Ept) \\
	\beginpgfgraphicnamed{TikZit//b1ring}
	\InputIfFileExists{TikZit//b1ring.tikz}{}{\input{./figures/TikZit//b1ring.tikz}}%
	\endpgfgraphicnamed
&(B1)  & %
	\beginpgfgraphicnamed{TikZit//b2ring}
	\InputIfFileExists{TikZit//b2ring.tikz}{}{\input{./figures/TikZit//b2ring.tikz}}%
	\endpgfgraphicnamed
&(B2)\\ 
  %
	\beginpgfgraphicnamed{TikZit//rpicopyns}
	\InputIfFileExists{TikZit//rpicopyns.tikz}{}{\input{./figures/TikZit//rpicopyns.tikz}}%
	\endpgfgraphicnamed
 &(B3)& %
	\beginpgfgraphicnamed{TikZit//anddflipns}
	\InputIfFileExists{TikZit//anddflipns.tikz}{}{\input{./figures/TikZit//anddflipns.tikz}}%
	\endpgfgraphicnamed
&(Brk) \\
 & &&\\
	\beginpgfgraphicnamed{TikZit//triangleocopy}
	\InputIfFileExists{TikZit//triangleocopy.tikz}{}{\input{./figures/TikZit//triangleocopy.tikz}}%
	\endpgfgraphicnamed
 &(Bas0) &%
	\beginpgfgraphicnamed{TikZit//trianglepicopyns}
	\InputIfFileExists{TikZit//trianglepicopyns.tikz}{}{\input{./figures/TikZit//trianglepicopyns.tikz}}%
	\endpgfgraphicnamed
&(Bas1)\\
	\beginpgfgraphicnamed{TikZit//plus1}
	\InputIfFileExists{TikZit//plus1.tikz}{}{\input{./figures/TikZit//plus1.tikz}}%
	\endpgfgraphicnamed
&(Suc)& %
	\beginpgfgraphicnamed{TikZit//triangleinvers}
	\InputIfFileExists{TikZit//triangleinvers.tikz}{}{\input{./figures/TikZit//triangleinvers.tikz}}%
	\endpgfgraphicnamed
  & (Inv) \\
   & &&\\
	\beginpgfgraphicnamed{TikZit//zerotoredns}
	\InputIfFileExists{TikZit//zerotoredns.tikz}{}{\input{./figures/TikZit//zerotoredns.tikz}}%
	\endpgfgraphicnamed
&(Zero)& %
	\beginpgfgraphicnamed{TikZit//eunoscalar}
	\InputIfFileExists{TikZit//eunoscalar.tikz}{}{\input{./figures/TikZit//eunoscalar.tikz}}%
	\endpgfgraphicnamed
&(EU) \\
	\beginpgfgraphicnamed{TikZit//lemma4}
	\InputIfFileExists{TikZit//lemma4.tikz}{}{\input{./figures/TikZit//lemma4.tikz}}%
	\endpgfgraphicnamed
&(Sym) &  %
	\beginpgfgraphicnamed{TikZit//associate}
	\InputIfFileExists{TikZit//associate.tikz}{}{\input{./figures/TikZit//associate.tikz}}%
	\endpgfgraphicnamed
 &(Aso)\\ 
	\beginpgfgraphicnamed{TikZit//TR1314combine2}
	\InputIfFileExists{TikZit//TR1314combine2.tikz}{}{\input{./figures/TikZit//TR1314combine2.tikz}}%
	\endpgfgraphicnamed
&(Pcy) &&\\ 
  		  		\hline
  		\end{array}\]  
  	\end{center}
  	\caption{Algebraic rules, $a, b \in \mathbb C.$ The upside-down flipped versions of the rules are assumed to hold as well. }\label{figurealgebra}
  \end{figure}
 \FloatBarrier
  \begin{remark}
As pointed out in \cite{wangalg2020}, the last three rules (Sym), (Aso), and (Pcy) are all about the properties of the W state (in a map form via map-state duality) %
	\beginpgfgraphicnamed{TikZit//wstatezx}
	\InputIfFileExists{TikZit//wstatezx.tikz}{}{\input{./figures/TikZit//wstatezx.tikz}}%
	\endpgfgraphicnamed
: (Sym) means the W state is symmetric, (Aso) means the W state is associative, and (Pcy) means any phase can be copied by the W state.
 \end{remark}

It can be verified that the rues in Figure \ref{figurealgebra} still hold under the standard interpretation $ \llbracket \cdot  \rrbracket $, which means the algebraic ZX-calculus is sound.

 \section{Derivable equalities mainly from previous reuslts}
 In this section, we derive equalities from algebraic rules in Figure \ref{figurealgebra}.  Some equalities have been essentially derived (up to scalars) in previous papers, but we prove them again due to generators and rules changed in the current formalism.
 
  For simplicity, we give two denotations as follows:
 \begin{equation}    \label{andshortnotationeq}
	\beginpgfgraphicnamed{TikZit//andshortnote}
	\InputIfFileExists{TikZit//andshortnote.tikz}{}{\input{./figures/TikZit//andshortnote.tikz}}%
	\endpgfgraphicnamed
 
      \end{equation}  

  Clearly, they have the  following relation
  \begin{equation}    \label{andshortnoterelat}
	\beginpgfgraphicnamed{TikZit//andshortnoterelation}
	\InputIfFileExists{TikZit//andshortnoterelation.tikz}{}{\input{./figures/TikZit//andshortnoterelation.tikz}}%
	\endpgfgraphicnamed
 
      \end{equation}  

  \begin{lemma}\cite{wangalg2020}
	\beginpgfgraphicnamed{TikZit//scalartimes2}
	\begin{tikzpicture}
	\begin{pgfonlayer}{nodelayer}
		\node [style=none] (0) at (0, 0) {$=$};
		\node [style=rn] (1) at (-0.5, 0.25) {$\pi$};
		\node [style=gbox] (2) at (0.5, 0) {${\scriptstyle a-1}$};
		\node [style=gbox] (3) at (-0.5, -0.25) {$a$};
	\end{pgfonlayer}
	\begin{pgfonlayer}{edgelayer}
		\draw (3) to (1);
	\end{pgfonlayer}
\end{tikzpicture}}%
	\endpgfgraphicnamed
 (Sca) 
      \end{lemma}    
   \begin{proof}
	\beginpgfgraphicnamed{TikZit//scalartimesprf2}
	\InputIfFileExists{TikZit//scalartimesprf2.tikz}{}{\input{./figures/TikZit//scalartimesprf2.tikz}}%
	\endpgfgraphicnamed
  
   \end{proof}

 \begin{corollary}\cite{wangalg2020}\label{zeroiscalarempty}
	\beginpgfgraphicnamed{TikZit//zeroscalarempty2}
	\InputIfFileExists{TikZit//zeroscalarempty2.tikz}{}{\input{./figures/TikZit//zeroscalarempty2.tikz}}%
	\endpgfgraphicnamed
  (Zos)
   \end{corollary} 

 \begin{proof}
	\beginpgfgraphicnamed{TikZit//zeroscalarempty2prf}
	\InputIfFileExists{TikZit//zeroscalarempty2prf.tikz}{}{\input{./figures/TikZit//zeroscalarempty2prf.tikz}}%
	\endpgfgraphicnamed
  
     \end{proof}   

   \begin{lemma}\cite{wangalg2020}
	\beginpgfgraphicnamed{TikZit//scalartimesgeneral}
	\begin{tikzpicture}
	\begin{pgfonlayer}{nodelayer}
		\node [style=gbox] (0) at (-1, 0) {$a$};
		\node [style=gbox] (1) at (-0.5, 0) {$b$};
		\node [style=gbox] (2) at (1, 0) {${\scriptstyle (a+1)(b+1)-1}$};
		\node [style=none] (3) at (0, 0) {$=$};
	\end{pgfonlayer}
\end{tikzpicture}}%
	\endpgfgraphicnamed
 (Sml)  
      \end{lemma}   
      
 \begin{proof}
	\beginpgfgraphicnamed{TikZit//scalartimesgeneralprf2}
	\InputIfFileExists{TikZit//scalartimesgeneralprf2.tikz}{}{\input{./figures/TikZit//scalartimesgeneralprf2.tikz}}%
	\endpgfgraphicnamed
  
   \end{proof}   
   
    \begin{corollary}\label{halfinverse}
	\beginpgfgraphicnamed{TikZit//halfinversedm}
	\InputIfFileExists{TikZit//halfinversedm.tikz}{}{\input{./figures/TikZit//halfinversedm.tikz}}%
	\endpgfgraphicnamed
  (Siv)
   \end{corollary} 
 \begin{proof}
	\beginpgfgraphicnamed{TikZit//halfinversedmprf}
	\InputIfFileExists{TikZit//halfinversedmprf.tikz}{}{\input{./figures/TikZit//halfinversedmprf.tikz}}%
	\endpgfgraphicnamed
  
   \end{proof}  
   
     \begin{lemma}
	\beginpgfgraphicnamed{TikZit//nhsquare}
	\InputIfFileExists{TikZit//nhsquare.tikz}{}{\input{./figures/TikZit//nhsquare.tikz}}%
	\endpgfgraphicnamed
 (H2)  
      \end{lemma} 
    \begin{proof}
	\beginpgfgraphicnamed{TikZit//nhsquareprf}
	\InputIfFileExists{TikZit//nhsquareprf.tikz}{}{\input{./figures/TikZit//nhsquareprf.tikz}}%
	\endpgfgraphicnamed
  
   \end{proof}  
   
      \begin{lemma} Suppose $\tau \in \{ 0, \pi \}$. Then
	\beginpgfgraphicnamed{TikZit//colorchanges}
	\InputIfFileExists{TikZit//colorchanges.tikz}{}{\input{./figures/TikZit//colorchanges.tikz}}%
	\endpgfgraphicnamed
 (H)  
      \end{lemma} 
   \begin{proof}
The proof directly follows from the definition of pink nodes (H), (H2), (Siv) and (Sca) . We also call this derived equality (H), since together with the definition of pink nodes (H) they compose the colour-change rule. 
   \end{proof}    
   
   \begin{lemma} Suppose $\tau, \sigma \in \{ 0, \pi \}$. Then
	\beginpgfgraphicnamed{TikZit//redspider0pifusion}
	\InputIfFileExists{TikZit//redspider0pifusion.tikz}{}{\input{./figures/TikZit//redspider0pifusion.tikz}}%
	\endpgfgraphicnamed
 (S1)  
      \end{lemma} 
   \begin{proof}
The proof directly follows from (S1), (H), (H2) and (Siv). We also call this derived equality (S1), since together with the green version of  (S1) they compose the spider fusion rule. 
   \end{proof}    
   
    \begin{lemma}\cite{msw2017}\label{hopfnslm}
	\beginpgfgraphicnamed{TikZit//hopfns}
	\InputIfFileExists{TikZit//hopfns.tikz}{}{\input{./figures/TikZit//hopfns.tikz}}%
	\endpgfgraphicnamed
 (Hopf)
\end{lemma}
   \begin{proof}
	\beginpgfgraphicnamed{TikZit//hopfnsprf}
	\InputIfFileExists{TikZit//hopfnsprf.tikz}{}{\input{./figures/TikZit//hopfnsprf.tikz}}%
	\endpgfgraphicnamed
  
   \end{proof} 
   
     \begin{lemma}\cite{wangalg2020}
	\beginpgfgraphicnamed{TikZit//redpitogreen2}
	\InputIfFileExists{TikZit//redpitogreen2.tikz}{}{\input{./figures/TikZit//redpitogreen2.tikz}}%
	\endpgfgraphicnamed
 (Bas1')
  \end{lemma} 
   \begin{proof}
	\beginpgfgraphicnamed{TikZit//redpitogreen2prf}
	\InputIfFileExists{TikZit//redpitogreen2prf.tikz}{}{\input{./figures/TikZit//redpitogreen2prf.tikz}}%
	\endpgfgraphicnamed
  
   \end{proof} 
   
    \begin{lemma}\cite{wangalg2020}\label{2eprf}
	\beginpgfgraphicnamed{TikZit//zx2e}
	\InputIfFileExists{TikZit//zx2e.tikz}{}{\input{./figures/TikZit//zx2e.tikz}}%
	\endpgfgraphicnamed
 
    \end{lemma}   
  \begin{proof}
$$ %
	\beginpgfgraphicnamed{TikZit//tr8primeprf32}
	\InputIfFileExists{TikZit//tr8primeprf32.tikz}{}{\input{./figures/TikZit//tr8primeprf32.tikz}}%
	\endpgfgraphicnamed
  $$
    \end{proof}     
  
   \begin{lemma}
\[ %
	\beginpgfgraphicnamed{TikZit//brkvariant}
	\InputIfFileExists{TikZit//brkvariant.tikz}{}{\input{./figures/TikZit//brkvariant.tikz}}%
	\endpgfgraphicnamed
   (Brk) \]
    \end{lemma}
     \begin{proof}
\[  %
	\beginpgfgraphicnamed{TikZit//brkvariantprf}
	\InputIfFileExists{TikZit//brkvariantprf.tikz}{}{\input{./figures/TikZit//brkvariantprf.tikz}}%
	\endpgfgraphicnamed
  \]
  The second equality can be obtained via the symmetry of green and pink spiders. We  call this derived equality (Brk) as well, since it is a variant of the  (Brk).
   \end{proof}

  \begin{lemma}\cite{wangalg2020}\label{trianglehopflm}
\[ %
	\beginpgfgraphicnamed{TikZit//trianglehopfns}
	\InputIfFileExists{TikZit//trianglehopfns.tikz}{}{\input{./figures/TikZit//trianglehopfns.tikz}}%
	\endpgfgraphicnamed
  \]
    \end{lemma}
  \begin{proof}
\[ %
	\beginpgfgraphicnamed{TikZit//Hopftrprf2}
	\InputIfFileExists{TikZit//Hopftrprf2.tikz}{}{\input{./figures/TikZit//Hopftrprf2.tikz}}%
	\endpgfgraphicnamed
  \]
 The second equality can be obtained via the symmetry of green and pink spiders. 
   \end{proof}
 
  \begin{lemma}\cite{wangalg2020}\label{2mprf}
	\beginpgfgraphicnamed{TikZit//2triangledeloopnopi2}
	\InputIfFileExists{TikZit//2triangledeloopnopi2.tikz}{}{\input{./figures/TikZit//2triangledeloopnopi2.tikz}}%
	\endpgfgraphicnamed
 
 \end{lemma}  
 
  \begin{proof}
 $$%
	\beginpgfgraphicnamed{TikZit//brk1prf2}
	\InputIfFileExists{TikZit//brk1prf2.tikz}{}{\input{./figures/TikZit//brk1prf2.tikz}}%
	\endpgfgraphicnamed
  $$
  The second equality can be obtained via the symmetry of green and pink spiders. 
   \end{proof}

    \begin{lemma}\cite{wangalg2020}\label{triangleonreddotlm}
	\beginpgfgraphicnamed{TikZit//triangleonreddot}
	\InputIfFileExists{TikZit//triangleonreddot.tikz}{}{\input{./figures/TikZit//triangleonreddot.tikz}}%
	\endpgfgraphicnamed
 
 \end{lemma}  
    \begin{proof}
	\beginpgfgraphicnamed{TikZit//triangleonreddotprf}
	\InputIfFileExists{TikZit//triangleonreddotprf.tikz}{}{\input{./figures/TikZit//triangleonreddotprf.tikz}}%
	\endpgfgraphicnamed
  
   \end{proof} 
   
   \begin{lemma}\label{2trianglebw2gnlm} 
\[ %
	\beginpgfgraphicnamed{TikZit//2trianglebw2gnlmdm}
	\InputIfFileExists{TikZit//2trianglebw2gnlmdm.tikz}{}{\input{./figures/TikZit//2trianglebw2gnlmdm.tikz}}%
	\endpgfgraphicnamed
  \]
 \end{lemma}
 \begin{proof}
 \[ %
	\beginpgfgraphicnamed{TikZit//2trianglebw2gnlmdmprf2}
	\InputIfFileExists{TikZit//2trianglebw2gnlmdmprf2.tikz}{}{\input{./figures/TikZit//2trianglebw2gnlmdmprf2.tikz}}%
	\endpgfgraphicnamed
  \]
 \end{proof}
 
 \begin{corollary} \label{andcopy}
	\beginpgfgraphicnamed{TikZit//andcopymet}
	\InputIfFileExists{TikZit//andcopymet.tikz}{}{\input{./figures/TikZit//andcopymet.tikz}}%
	\endpgfgraphicnamed

  \end{corollary}

   \begin{lemma}\cite{wangalg2020}\label{TR4g}
  \begin{equation}\label{TR4geq}
	\beginpgfgraphicnamed{TikZit//tr4g2}
	\InputIfFileExists{TikZit//tr4g2.tikz}{}{\input{./figures/TikZit//tr4g2.tikz}}%
	\endpgfgraphicnamed
 
   \end{equation}
    \end{lemma}
 \begin{proof}
 $$%
	\beginpgfgraphicnamed{TikZit//tr4gprf2}
	\InputIfFileExists{TikZit//tr4gprf2.tikz}{}{\input{./figures/TikZit//tr4gprf2.tikz}}%
	\endpgfgraphicnamed
$$  
 The other part can be obtained by the symmetry of green spider.
  \end{proof}  
  
    \begin{lemma}\cite{wangalg2020}\label{Hopfgtr}
  \begin{equation}\label{Hopfgtreq}
	\beginpgfgraphicnamed{TikZit//trianglehopfgreen2}
	\InputIfFileExists{TikZit//trianglehopfgreen2.tikz}{}{\input{./figures/TikZit//trianglehopfgreen2.tikz}}%
	\endpgfgraphicnamed

  \end{equation}
    \end{lemma}
 \begin{proof}
$$ %
	\beginpgfgraphicnamed{TikZit//trianglehopfgreenprf2}
	\InputIfFileExists{TikZit//trianglehopfgreenprf2.tikz}{}{\input{./figures/TikZit//trianglehopfgreenprf2.tikz}}%
	\endpgfgraphicnamed
  $$
 The second equality  can be obtained by the symmetry of green spider.
  \end{proof}

  \begin{lemma}\cite{msw2017}\label{gpiinhadalm}
	\beginpgfgraphicnamed{TikZit//gpiinhada}
	\InputIfFileExists{TikZit//gpiinhada.tikz}{}{\input{./figures/TikZit//gpiinhada.tikz}}%
	\endpgfgraphicnamed

\end{lemma}

    \begin{proof}
	\beginpgfgraphicnamed{TikZit//gpiinhadaprf}
	\InputIfFileExists{TikZit//gpiinhadaprf.tikz}{}{\input{./figures/TikZit//gpiinhadaprf.tikz}}%
	\endpgfgraphicnamed
  
   \end{proof}  
   
\begin{lemma}\cite{msw2017}\label{pimultiplecplm}
Suppose $m \geq 0$. Then
	\beginpgfgraphicnamed{TikZit//pigrcopy}
	\InputIfFileExists{TikZit//pigrcopy.tikz}{}{\input{./figures/TikZit//pigrcopy.tikz}}%
	\endpgfgraphicnamed
 (Pic)
\end{lemma}
   \begin{proof}
\[  %
	\beginpgfgraphicnamed{TikZit//pigrcopyprf}
	\InputIfFileExists{TikZit//pigrcopyprf.tikz}{}{\input{./figures/TikZit//pigrcopyprf.tikz}}%
	\endpgfgraphicnamed
  \]
The general case follows directly from the above two special cases.
   \end{proof}  

 \begin{corollary}
	\beginpgfgraphicnamed{TikZit//pimultiplecp}
	\InputIfFileExists{TikZit//pimultiplecp.tikz}{}{\input{./figures/TikZit//pimultiplecp.tikz}}%
	\endpgfgraphicnamed
  (Pic)
   \end{corollary} 
 \begin{proof}
It follows directly from Lemma \ref{pimultiplecplm} and the colour-change rule (H). We call this equality  (Pic) as well.
   \end{proof}  
   
       \begin{corollary}
  \begin{equation}
	\beginpgfgraphicnamed{TikZit//2triangledeloopnopiflipns}
	\InputIfFileExists{TikZit//2triangledeloopnopiflipns.tikz}{}{\input{./figures/TikZit//2triangledeloopnopiflipns.tikz}}%
	\endpgfgraphicnamed
 \quad (Brk1') 
   \end{equation}  
    \end{corollary}  
     \begin{proof}
 $$%
	\beginpgfgraphicnamed{TikZit//brk1primeprf2}
	\InputIfFileExists{TikZit//brk1primeprf2.tikz}{}{\input{./figures/TikZit//brk1primeprf2.tikz}}%
	\endpgfgraphicnamed
  $$
 The second equality  can be proved by symmetry of green spider. 
   \end{proof}

    \begin{lemma}\label{trianglerpidotlm}
	\beginpgfgraphicnamed{TikZit//trianglerpidot}
	\InputIfFileExists{TikZit//trianglerpidot.tikz}{}{\input{./figures/TikZit//trianglerpidot.tikz}}%
	\endpgfgraphicnamed
 
    \end{lemma}     
     \begin{proof}
	\beginpgfgraphicnamed{TikZit//trianglerpidotprf}
	\InputIfFileExists{TikZit//trianglerpidotprf.tikz}{}{\input{./figures/TikZit//trianglerpidotprf.tikz}}%
	\endpgfgraphicnamed
  
    \end{proof}   

 \begin{lemma}\cite{wangalg2020}
	\beginpgfgraphicnamed{TikZit//zerodecom2}
	\InputIfFileExists{TikZit//zerodecom2.tikz}{}{\input{./figures/TikZit//zerodecom2.tikz}}%
	\endpgfgraphicnamed
 (Zero') 
  \end{lemma}
 \begin{proof}
$$ %
	\beginpgfgraphicnamed{TikZit//zerodecomprf2ns}
	\InputIfFileExists{TikZit//zerodecomprf2ns.tikz}{}{\input{./figures/TikZit//zerodecomprf2ns.tikz}}%
	\endpgfgraphicnamed
  $$
     \end{proof}

\begin{lemma}\cite{wangalg2020}\label{tr5primelm}
\[%
	\beginpgfgraphicnamed{TikZit//tr5prime2}
	\InputIfFileExists{TikZit//tr5prime2.tikz}{}{\input{./figures/TikZit//tr5prime2.tikz}}%
	\endpgfgraphicnamed
\]
\end{lemma}
  \begin{proof}
$$ %
	\beginpgfgraphicnamed{TikZit//tr5primeprf2}
	\InputIfFileExists{TikZit//tr5primeprf2.tikz}{}{\input{./figures/TikZit//tr5primeprf2.tikz}}%
	\endpgfgraphicnamed
  $$
    \end{proof}

\begin{lemma}\label{1triangle1pibw2gnlm} 
\[ %
	\beginpgfgraphicnamed{TikZit//1triangle1pibw2gndm}
	\InputIfFileExists{TikZit//1triangle1pibw2gndm.tikz}{}{\input{./figures/TikZit//1triangle1pibw2gndm.tikz}}%
	\endpgfgraphicnamed
  \]
 \end{lemma}
 \begin{proof}
 \[ %
	\beginpgfgraphicnamed{TikZit//1triangle1pibw2gndmprf}
	\InputIfFileExists{TikZit//1triangle1pibw2gndmprf.tikz}{}{\input{./figures/TikZit//1triangle1pibw2gndmprf.tikz}}%
	\endpgfgraphicnamed
  \]
 \end{proof}

 \begin{lemma}\label{1tricpto2redlm} 
\[ %
	\beginpgfgraphicnamed{TikZit//1tricpto2red}
	\InputIfFileExists{TikZit//1tricpto2red.tikz}{}{\input{./figures/TikZit//1tricpto2red.tikz}}%
	\endpgfgraphicnamed
  \]
 \end{lemma}
 \begin{proof}
 \[ %
	\beginpgfgraphicnamed{TikZit//1tricpto2redprf}
	\InputIfFileExists{TikZit//1tricpto2redprf.tikz}{}{\input{./figures/TikZit//1tricpto2redprf.tikz}}%
	\endpgfgraphicnamed
  \]
 \end{proof}
 
   \begin{lemma}\cite{wangalg2020}\label{trianglecopylrlm}
\[  %
	\beginpgfgraphicnamed{TikZit//trianglecopylr}
	\InputIfFileExists{TikZit//trianglecopylr.tikz}{}{\input{./figures/TikZit//trianglecopylr.tikz}}%
	\endpgfgraphicnamed
 \]
 \end{lemma}    
   \begin{proof}
\[ %
	\beginpgfgraphicnamed{TikZit//trianglecopylrprf2}
	\InputIfFileExists{TikZit//trianglecopylrprf2.tikz}{}{\input{./figures/TikZit//trianglecopylrprf2.tikz}}%
	\endpgfgraphicnamed
  \]
       \end{proof}

  \begin{lemma}\cite{wangalg2020}
   \begin{equation*}
	\beginpgfgraphicnamed{TikZit//equivalentaddrulens}
	\InputIfFileExists{TikZit//equivalentaddrulens.tikz}{}{\input{./figures/TikZit//equivalentaddrulens.tikz}}%
	\endpgfgraphicnamed
 \quad (AD') 
   \end{equation*}    
     \end{lemma}
    \begin{proof}
    If  $b\neq 0$, then 
$$%
	\beginpgfgraphicnamed{TikZit//equivalentaddruleprf2}
	\InputIfFileExists{TikZit//equivalentaddruleprf2.tikz}{}{\input{./figures/TikZit//equivalentaddruleprf2.tikz}}%
	\endpgfgraphicnamed
  $$
 If $b= 0$, then
$$%
	\beginpgfgraphicnamed{TikZit//equivalentaddruleprf0ns}
	\InputIfFileExists{TikZit//equivalentaddruleprf0ns.tikz}{}{\input{./figures/TikZit//equivalentaddruleprf0ns.tikz}}%
	\endpgfgraphicnamed
  $$
   \end{proof}   
 
 \begin{lemma}\cite{wangalg2020}\label{additiongbxlm}
	\beginpgfgraphicnamed{TikZit//additiongbx}
	\InputIfFileExists{TikZit//additiongbx.tikz}{}{\input{./figures/TikZit//additiongbx.tikz}}%
	\endpgfgraphicnamed
 
  \end{lemma}
  \begin{proof}
$$ %
	\beginpgfgraphicnamed{TikZit//addprf}
	\InputIfFileExists{TikZit//addprf.tikz}{}{\input{./figures/TikZit//addprf.tikz}}%
	\endpgfgraphicnamed
  $$
   \end{proof}

 \begin{lemma}\cite{wangalg2020}
	\beginpgfgraphicnamed{TikZit//definitionTriangleInverse2}
	\InputIfFileExists{TikZit//definitionTriangleInverse2.tikz}{}{\input{./figures/TikZit//definitionTriangleInverse2.tikz}}%
	\endpgfgraphicnamed
 (Ivt)
   \end{lemma}
  \begin{proof}
$$ %
	\beginpgfgraphicnamed{TikZit//invtriprf}
	\InputIfFileExists{TikZit//invtriprf.tikz}{}{\input{./figures/TikZit//invtriprf.tikz}}%
	\endpgfgraphicnamed
  $$
    \end{proof}   

\begin{lemma}\cite{wangalg2020}\label{1iprf}
 %
	\beginpgfgraphicnamed{TikZit//picommutationdm}
	\InputIfFileExists{TikZit//picommutationdm.tikz}{}{\input{./figures/TikZit//picommutationdm.tikz}}%
	\endpgfgraphicnamed
 
   \end{lemma} 
 \begin{proof}
\[ %
	\beginpgfgraphicnamed{TikZit//k2prf}
	\InputIfFileExists{TikZit//k2prf.tikz}{}{\input{./figures/TikZit//k2prf.tikz}}%
	\endpgfgraphicnamed
  \]
The proof completes when setting $a=e^{i\alpha}$.
 \end{proof}
 
  \begin{lemma}\cite{wangalg2020}\label{gpiintriangleslm}
	\beginpgfgraphicnamed{TikZit//gpiintriangles}
	\InputIfFileExists{TikZit//gpiintriangles.tikz}{}{\input{./figures/TikZit//gpiintriangles.tikz}}%
	\endpgfgraphicnamed
 
     \end{lemma}
  \begin{proof}
$$ %
	\beginpgfgraphicnamed{TikZit//gpiintrianglesprf}
	\InputIfFileExists{TikZit//gpiintrianglesprf.tikz}{}{\input{./figures/TikZit//gpiintrianglesprf.tikz}}%
	\endpgfgraphicnamed
  $$
    \end{proof}

 \begin{corollary}\cite{wangalg2020}\label{pitinvcomut}
   \begin{equation}\label{pitinvcomuteq}
	\beginpgfgraphicnamed{TikZit/rpitrinverse}
	\InputIfFileExists{TikZit/rpitrinverse.tikz}{}{\input{./figures/TikZit/rpitrinverse.tikz}}%
	\endpgfgraphicnamed
 
     \end{equation}
 \end{corollary}
  \begin{proof}
$$ %
	\beginpgfgraphicnamed{TikZit//rpitrinverseprf}
	\InputIfFileExists{TikZit//rpitrinverseprf.tikz}{}{\input{./figures/TikZit//rpitrinverseprf.tikz}}%
	\endpgfgraphicnamed
  $$
    \end{proof} 

  \begin{lemma}\cite{wangalg2020}\label{andgate2v}
	\beginpgfgraphicnamed{TikZit//andgate2vs}
	\InputIfFileExists{TikZit//andgate2vs.tikz}{}{\input{./figures/TikZit//andgate2vs.tikz}}%
	\endpgfgraphicnamed
 
 \end{lemma}
   \begin{proof}
   $$  %
	\beginpgfgraphicnamed{TikZit//andgate2vsprf}
	\InputIfFileExists{TikZit//andgate2vsprf.tikz}{}{\input{./figures/TikZit//andgate2vsprf.tikz}}%
	\endpgfgraphicnamed
 $$
   The second equality can be obtained via symmetry.
  \end{proof}

  \begin{lemma}\cite{wangalg2020}\label{trianglehopflip}
 \begin{equation}\label{trianglehopflipeq}
	\beginpgfgraphicnamed{TikZit//trianglehopfflip}
	\InputIfFileExists{TikZit//trianglehopfflip.tikz}{}{\input{./figures/TikZit//trianglehopfflip.tikz}}%
	\endpgfgraphicnamed
 
   \end{equation} 
    \end{lemma}
  \begin{proof}
 $$%
	\beginpgfgraphicnamed{TikZit//trianglehopfflipprf}
	\InputIfFileExists{TikZit//trianglehopfflipprf.tikz}{}{\input{./figures/TikZit//trianglehopfflipprf.tikz}}%
	\endpgfgraphicnamed
  $$
   \end{proof}   

 \begin{lemma}\cite{wangalg2020}\label{2kprf}
	\beginpgfgraphicnamed{TikZit//2triangleup}
	\InputIfFileExists{TikZit//2triangleup.tikz}{}{\input{./figures/TikZit//2triangleup.tikz}}%
	\endpgfgraphicnamed
 
   \end{lemma}  
   
 \begin{proof}
$$  %
	\beginpgfgraphicnamed{TikZit//2triangleupprf}
	\InputIfFileExists{TikZit//2triangleupprf.tikz}{}{\input{./figures/TikZit//2triangleupprf.tikz}}%
	\endpgfgraphicnamed
$$
  \end{proof}

      \begin{lemma}\cite{wangalg2020}
	\beginpgfgraphicnamed{TikZit//anddflipwitha2}
	\InputIfFileExists{TikZit//anddflipwitha2.tikz}{}{\input{./figures/TikZit//anddflipwitha2.tikz}}%
	\endpgfgraphicnamed
 (Brkp)      
   \end{lemma} 
   \begin{proof}
   \[  %
	\beginpgfgraphicnamed{TikZit//anddflipwitha2prf}
	\InputIfFileExists{TikZit//anddflipwitha2prf.tikz}{}{\input{./figures/TikZit//anddflipwitha2prf.tikz}}%
	\endpgfgraphicnamed
 \]
   Therefore,
      \[  %
	\beginpgfgraphicnamed{TikZit//anddflipwitha2prf2}
	\InputIfFileExists{TikZit//anddflipwitha2prf2.tikz}{}{\input{./figures/TikZit//anddflipwitha2prf2.tikz}}%
	\endpgfgraphicnamed
 \]
     \end{proof}   

  \begin{lemma}\cite{wangalg2020}\label{andbial}
	\beginpgfgraphicnamed{TikZit//ruleA3}
	\InputIfFileExists{TikZit//ruleA3.tikz}{}{\input{./figures/TikZit//ruleA3.tikz}}%
	\endpgfgraphicnamed
 (BiA)
 \end{lemma}
  
   \begin{proof}
 $$  %
	\beginpgfgraphicnamed{TikZit//andbialgebraprf}
	\InputIfFileExists{TikZit//andbialgebraprf.tikz}{}{\input{./figures/TikZit//andbialgebraprf.tikz}}%
	\endpgfgraphicnamed
 $$
  \end{proof}
  
    \begin{corollary}\cite{wangalg2020}\label{generalbialgebra}
	\beginpgfgraphicnamed{TikZit//generalBiA}
	\InputIfFileExists{TikZit//generalBiA.tikz}{}{\input{./figures/TikZit//generalBiA.tikz}}%
	\endpgfgraphicnamed
 
 \end{corollary}
 
 \begin{corollary}[\cite{bobanthonywang}]\label{generalbialgebra}
   For any $k\geq 0$, we have   
$$%
	\beginpgfgraphicnamed{TikZit//generalBiAvariant}
	\InputIfFileExists{TikZit//generalBiAvariant.tikz}{}{\input{./figures/TikZit//generalBiAvariant.tikz}}%
	\endpgfgraphicnamed
 $$
 or equivalently,
$$%
	\beginpgfgraphicnamed{TikZit//appendixL32eqv}
	\InputIfFileExists{TikZit//appendixL32eqv.tikz}{}{\input{./figures/TikZit//appendixL32eqv.tikz}}%
	\endpgfgraphicnamed
 $$
where
$$%
	\beginpgfgraphicnamed{TikZit//appendixL32a}
	\InputIfFileExists{TikZit//appendixL32a.tikz}{}{\input{./figures/TikZit//appendixL32a.tikz}}%
	\endpgfgraphicnamed
, \quad\quad\quad %
	\beginpgfgraphicnamed{TikZit//appendixL32b}
	\InputIfFileExists{TikZit//appendixL32b.tikz}{}{\input{./figures/TikZit//appendixL32b.tikz}}%
	\endpgfgraphicnamed
 $$
 \end{corollary}

\begin{corollary}  \label{andadditionco}
$$%
	\beginpgfgraphicnamed{TikZit//andadditioncor}
	\InputIfFileExists{TikZit//andadditioncor.tikz}{}{\input{./figures/TikZit//andadditioncor.tikz}}%
	\endpgfgraphicnamed
 $$ 
 \end{corollary}


   \begin{lemma}\cite{wangalg2020}\label{andgatehadam}
	\beginpgfgraphicnamed{TikZit//andgatehadama}
	\InputIfFileExists{TikZit//andgatehadama.tikz}{}{\input{./figures/TikZit//andgatehadama.tikz}}%
	\endpgfgraphicnamed
 
 \end{lemma}
  \begin{proof}
 $$  %
	\beginpgfgraphicnamed{TikZit//andgatehadamaprf}
	\InputIfFileExists{TikZit//andgatehadamaprf.tikz}{}{\input{./figures/TikZit//andgatehadamaprf.tikz}}%
	\endpgfgraphicnamed
 $$
 The second equality can be obtained by symmetry. 
  \end{proof}
  
 \begin{lemma}\cite{wangalg2020}\label{distribute}
	\beginpgfgraphicnamed{TikZit//ruleA1_Draft}
	\InputIfFileExists{TikZit//ruleA1_Draft.tikz}{}{\input{./figures/TikZit//ruleA1_Draft.tikz}}%
	\endpgfgraphicnamed
 (Dis)
 \end{lemma}
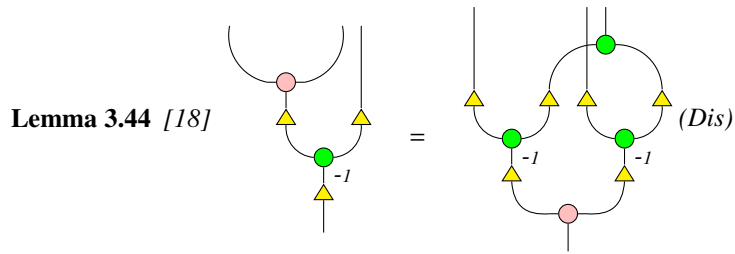
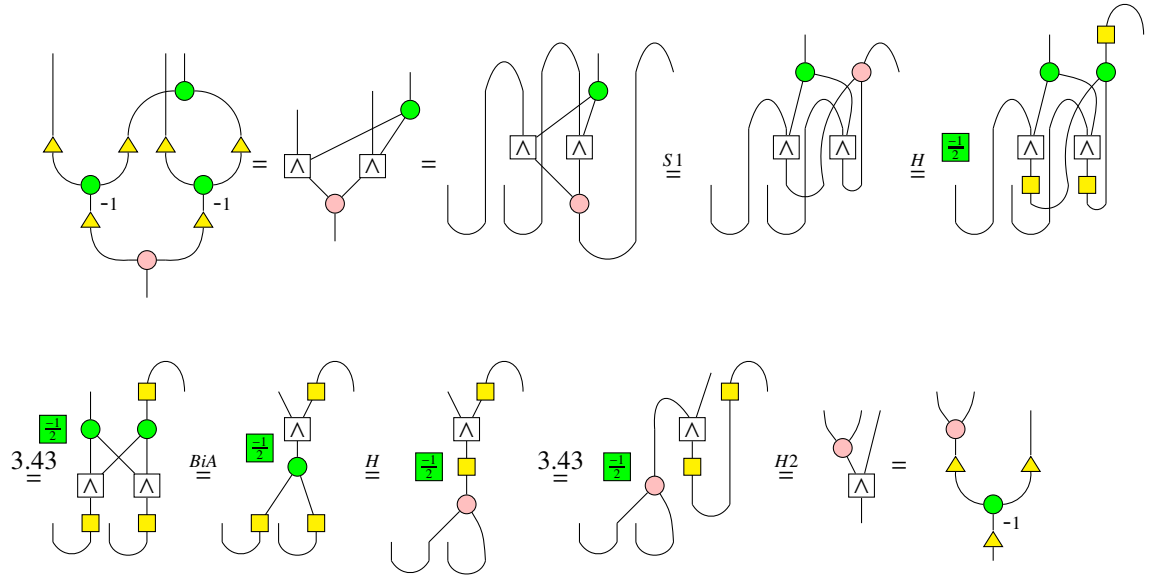
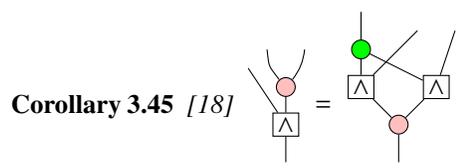
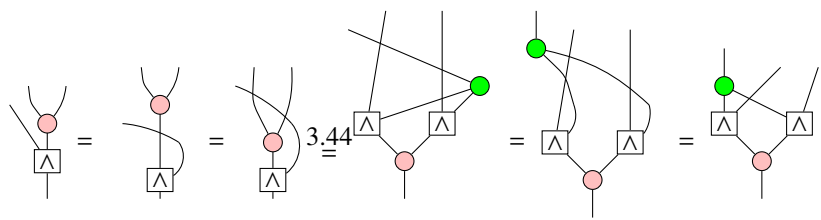
 \begin{proof}
 $$  %
	\beginpgfgraphicnamed{TikZit//distributionprf}
	\InputIfFileExists{TikZit//distributionprf.tikz}{}{\input{./figures/TikZit//distributionprf.tikz}}%
	\endpgfgraphicnamed
 $$
  \end{proof}
  
   \begin{corollary}\cite{wangalg2020}\label{distribute2}
	\beginpgfgraphicnamed{TikZit//distribute2}
	\InputIfFileExists{TikZit//distribute2.tikz}{}{\input{./figures/TikZit//distribute2.tikz}}%
	\endpgfgraphicnamed
 
 \end{corollary}
  \begin{proof}
 $$  %
	\beginpgfgraphicnamed{TikZit//distribute2prf}
	\InputIfFileExists{TikZit//distribute2prf.tikz}{}{\input{./figures/TikZit//distribute2prf.tikz}}%
	\endpgfgraphicnamed
 $$
  \end{proof}

\section{Newly derivable equalities}
In this section, we derive all the equalities needed for the proof of completeness while relatively new to previous results. 

  \begin{proposition}\label{picntcommut}
  Let $i, j_1, \cdots, j_t, \cdots,  j_s \in \{1, \cdots, m\}, i \notin \{j_1, \cdots, j_t, \cdots,  j_s\}$. Then we have
 $$%
	\beginpgfgraphicnamed{TikZit//picntcommute}
	\InputIfFileExists{TikZit//picntcommute.tikz}{}{\input{./figures/TikZit//picntcommute.tikz}}%
	\endpgfgraphicnamed
$$
 where the  node $a$ is connected to $ j_1, \cdots, j_t, \cdots,  j_s $ via pink nodes. 
 \end{proposition}   
   \begin{proof}
  $$  %
	\beginpgfgraphicnamed{TikZit//picntcommuteprf}
	\InputIfFileExists{TikZit//picntcommuteprf.tikz}{}{\input{./figures/TikZit//picntcommuteprf.tikz}}%
	\endpgfgraphicnamed
 $$
  \end{proof}
  
   \begin{corollary}\label{picntcommutcro}
 $$ %
	\beginpgfgraphicnamed{TikZit//picntcommutecro}
	\InputIfFileExists{TikZit//picntcommutecro.tikz}{}{\input{./figures/TikZit//picntcommutecro.tikz}}%
	\endpgfgraphicnamed
$$
 \end{corollary}

   \begin{proposition}\label{picntcommutesam}
  Let $i, j_1, \cdots, j_t, \cdots,  j_s \in \{1, \cdots, m\}, i \notin \{j_1, \cdots, j_t, \cdots,  j_s\}$. Then we have
 $$%
	\beginpgfgraphicnamed{TikZit//picntcommutesame}
	\InputIfFileExists{TikZit//picntcommutesame.tikz}{}{\input{./figures/TikZit//picntcommutesame.tikz}}%
	\endpgfgraphicnamed
$$
 where the  node $a$ is connected to $ j_1, \cdots, j_t, \cdots,  j_s $ via pink nodes. 
 \end{proposition}   
  \begin{proof}
  $$  %
	\beginpgfgraphicnamed{TikZit//picntcommutesameprf}
	\InputIfFileExists{TikZit//picntcommutesameprf.tikz}{}{\input{./figures/TikZit//picntcommutesameprf.tikz}}%
	\endpgfgraphicnamed
 $$
  \end{proof}
Similarly, we have 
    \begin{proposition}\label{picntcommutesamgrn}
  Let $i, k, j_1, \cdots,  j_s \in \{1, \cdots, m\}, i, k \notin \{j_1, \cdots,   j_s\}$. Then we have
 $$%
	\beginpgfgraphicnamed{TikZit//picntcommutesamegr}
	\InputIfFileExists{TikZit//picntcommutesamegr.tikz}{}{\input{./figures/TikZit//picntcommutesamegr.tikz}}%
	\endpgfgraphicnamed
$$
 where the  node $a$ is connected to $ j_1, \cdots,  j_s $ via pink nodes. 
 \end{proposition}  
 
   
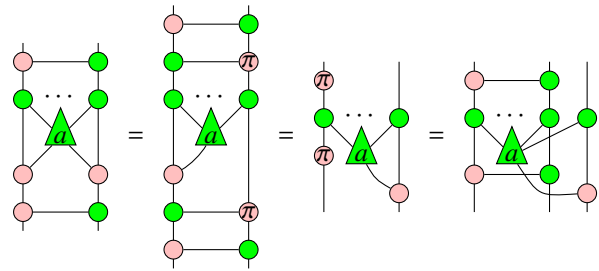
\begin{corollary}\label{picntcommutcro2}
 $$ %
	\beginpgfgraphicnamed{TikZit//picntcommutecro2}
	\InputIfFileExists{TikZit//picntcommutecro2.tikz}{}{\input{./figures/TikZit//picntcommutecro2.tikz}}%
	\endpgfgraphicnamed
$$
 \end{corollary}
   
In the same way, we have 
    \begin{proposition}\label{picntcommuteand}
     $$%
	\beginpgfgraphicnamed{TikZit//picntcommuteandgt}
	\InputIfFileExists{TikZit//picntcommuteandgt.tikz}{}{\input{./figures/TikZit//picntcommuteandgt.tikz}}%
	\endpgfgraphicnamed
$$
   \end{proposition}  
  
    \begin{corollary}\label{picntcommuteandcr1}
     $$%
	\beginpgfgraphicnamed{TikZit//picntcommuteandgtcr1}
	\InputIfFileExists{TikZit//picntcommuteandgtcr1.tikz}{}{\input{./figures/TikZit//picntcommuteandgtcr1.tikz}}%
	\endpgfgraphicnamed
$$
   \end{corollary}

    \begin{lemma}\label{hopfvar2}
$$ %
	\beginpgfgraphicnamed{TikZit//hopfvariant2}
	\InputIfFileExists{TikZit//hopfvariant2.tikz}{}{\input{./figures/TikZit//hopfvariant2.tikz}}%
	\endpgfgraphicnamed
$$
 \end{lemma}
   \begin{proof}
   $$  %
	\beginpgfgraphicnamed{TikZit//hopfvariant2prf}
	\InputIfFileExists{TikZit//hopfvariant2prf.tikz}{}{\input{./figures/TikZit//hopfvariant2prf.tikz}}%
	\endpgfgraphicnamed
 $$
  \end{proof}

  \begin{lemma}\label{ruletensorad}
$$%
	\beginpgfgraphicnamed{TikZit//ruletensoradd}
	\InputIfFileExists{TikZit//ruletensoradd.tikz}{}{\input{./figures/TikZit//ruletensoradd.tikz}}%
	\endpgfgraphicnamed
 $$
 \end{lemma}
   \begin{proof}
   $$  %
	\beginpgfgraphicnamed{TikZit//ruletensoraddprf}
	\InputIfFileExists{TikZit//ruletensoraddprf.tikz}{}{\input{./figures/TikZit//ruletensoraddprf.tikz}}%
	\endpgfgraphicnamed
 $$
  \end{proof}

\begin{proposition}\label{prop1}
 For any $k\geq 1$, we have
    \begin{equation}\label{prop1eq}
	\beginpgfgraphicnamed{TikZit//propo1}
	\InputIfFileExists{TikZit//propo1.tikz}{}{\input{./figures/TikZit//propo1.tikz}}%
	\endpgfgraphicnamed

 \end{equation}
 
 \end{proposition}
 \begin{proof}
 $$ %
	\beginpgfgraphicnamed{TikZit//propo1prf}
	\InputIfFileExists{TikZit//propo1prf.tikz}{}{\input{./figures/TikZit//propo1prf.tikz}}%
	\endpgfgraphicnamed
$$
 \end{proof}

 \begin{corollary}\label{propo1cro1}
 $$ %
	\beginpgfgraphicnamed{TikZit//propo1cr1}
	\InputIfFileExists{TikZit//propo1cr1.tikz}{}{\input{./figures/TikZit//propo1cr1.tikz}}%
	\endpgfgraphicnamed
$$
 \end{corollary}
This can be immediately obtained by plugging pink $\pi$ phase gates from the top and the bottom of the left-most line of diagrams on both sides of  (\ref{prop1eq}).

 \begin{corollary}\label{propo1cro2}
 $$ %
	\beginpgfgraphicnamed{TikZit//propo1cr2}
	\InputIfFileExists{TikZit//propo1cr2.tikz}{}{\input{./figures/TikZit//propo1cr2.tikz}}%
	\endpgfgraphicnamed
$$
 \end{corollary}
 This can be obtained by swapping the $1$-th and the $j$-th lines of diagrams on both side of  (\ref{prop1eq}).
 
 \begin{corollary}
 $$ %
	\beginpgfgraphicnamed{TikZit//propo1cr3}
	\InputIfFileExists{TikZit//propo1cr3.tikz}{}{\input{./figures/TikZit//propo1cr3.tikz}}%
	\endpgfgraphicnamed
$$
 \end{corollary} 
  This follows directly from Proposition \ref{prop1} and Corollary \ref{picntcommutcro}.
  
   \begin{corollary}\label{nlinestensornormalform}
 \begin{equation}\label{nlinestensornormalformeq}
	\beginpgfgraphicnamed{TikZit//nlinetensornormalform}
	\InputIfFileExists{TikZit//nlinetensornormalform.tikz}{}{\input{./figures/TikZit//nlinetensornormalform.tikz}}%
	\endpgfgraphicnamed

  \end{equation}
 where on the RHD of (\ref{nlinestensornormalformeq}), there are $2^n$ green triangles labeled by $a$ on the right-most $m$ wires, each green triangle  is connected to the left-most $n$ wires via green dots which are surrounded by $k$ pairs of red $\pi$s with $0 \leq k  \leq n$, and different green triangles have different  distribution of pairs of red $\pi$s, that's why there are $\binom{n}{0}+\binom{n}{1} +\cdots + \binom{n}{n}=2^n$ green triangles labeled by $a$.
 \end{corollary} 

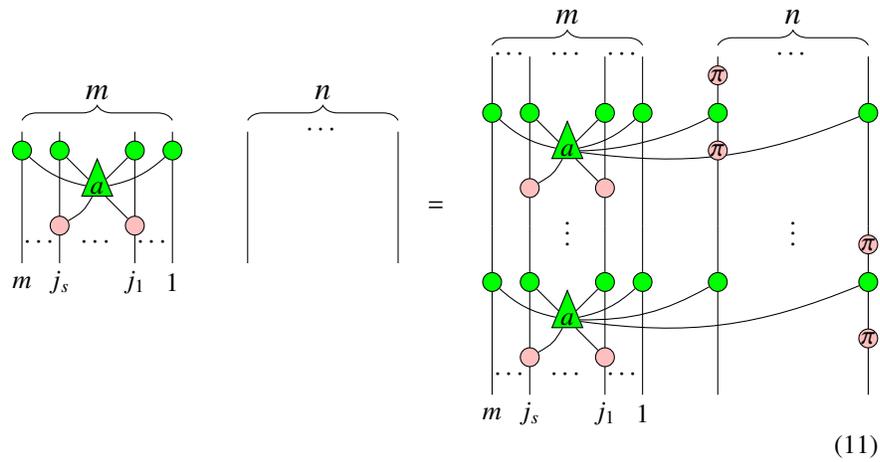
 \begin{corollary}\label{normalformtensornlines}
 \begin{equation}\label{normalformtensornlineseq}
	\beginpgfgraphicnamed{TikZit//normalformtensorsnlines}
	\InputIfFileExists{TikZit//normalformtensorsnlines.tikz}{}{\input{./figures/TikZit//normalformtensorsnlines.tikz}}%
	\endpgfgraphicnamed

  \end{equation}
 where on the RHD of (\ref{nlinestensornormalformeq}), there are $2^n$ green triangles labeled by $a$ on the left-most $m$ wires, each green triangle  is connected to the right-most $n$ wires via green dots which are surrounded by $k$ pairs of red $\pi$s with $0 \leq k  \leq n$, and different green triangles have different  distribution of pairs of red $\pi$s, that's why there are $\binom{n}{0}+\binom{n}{1} +\cdots + \binom{n}{n}=2^n$ green triangles labeled by $a$.
 \end{corollary} 



   \begin{proposition}\label{propadprime}
 For any $k\geq 1$, we have
 $$%
	\beginpgfgraphicnamed{TikZit//propaddprime}
	\InputIfFileExists{TikZit//propaddprime.tikz}{}{\input{./figures/TikZit//propaddprime.tikz}}%
	\endpgfgraphicnamed
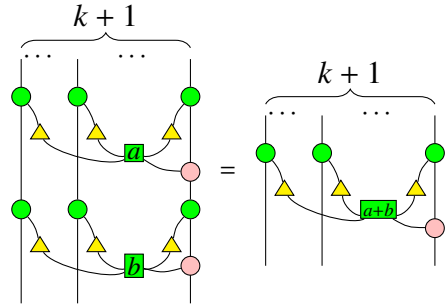
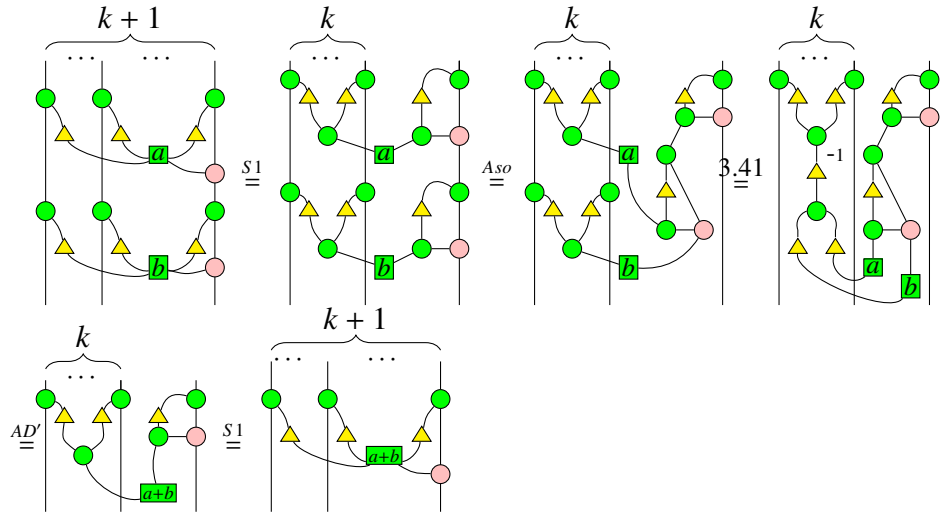
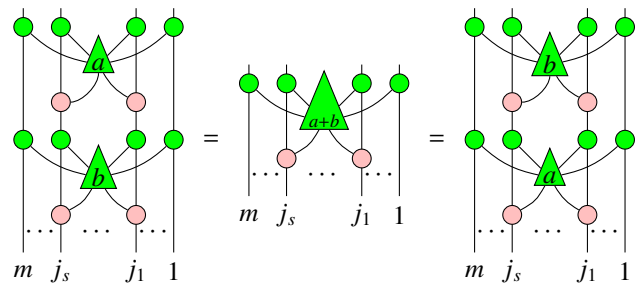
$$
 \end{proposition}

    \begin{proof}
$$%
	\beginpgfgraphicnamed{TikZit//propaddprimeprf}
	\InputIfFileExists{TikZit//propaddprimeprf.tikz}{}{\input{./figures/TikZit//propaddprimeprf.tikz}}%
	\endpgfgraphicnamed
  $$
   \end{proof}      
   
 \begin{corollary}\label{propadprimecro}
 $$ %
	\beginpgfgraphicnamed{TikZit//propaddprimecro}
	\InputIfFileExists{TikZit//propaddprimecro.tikz}{}{\input{./figures/TikZit//propaddprimecro.tikz}}%
	\endpgfgraphicnamed
$$
 \end{corollary} 
This can be directly obtained from Proposition \ref{propadprime} and Corollary \ref{picntcommutcro}.

   
   \begin{lemma}\label{ruletensorLsimpler}
	\beginpgfgraphicnamed{TikZit//ruletensorLsim}
	\InputIfFileExists{TikZit//ruletensorLsim.tikz}{}{\input{./figures/TikZit//ruletensorLsim.tikz}}%
	\endpgfgraphicnamed

 \end{lemma}
  \begin{proof}
   $$  %
	\beginpgfgraphicnamed{TikZit//ruletensorLsimprf}
	\InputIfFileExists{TikZit//ruletensorLsimprf.tikz}{}{\input{./figures/TikZit//ruletensorLsimprf.tikz}}%
	\endpgfgraphicnamed
 $$
  \end{proof}

     \begin{lemma}\label{ruletensor}
	\beginpgfgraphicnamed{TikZit//ruletensorL}
	\InputIfFileExists{TikZit//ruletensorL.tikz}{}{\input{./figures/TikZit//ruletensorL.tikz}}%
	\endpgfgraphicnamed
 
 \end{lemma}
  
   \begin{proof}
   $$  %
	\beginpgfgraphicnamed{TikZit//ruletensorLprf}
	\InputIfFileExists{TikZit//ruletensorLprf.tikz}{}{\input{./figures/TikZit//ruletensorLprf.tikz}}%
	\endpgfgraphicnamed
 $$
  \end{proof}

 \begin{proposition}\label{itensorand}
 For any $k\geq 1$, we have
 $$%
	\beginpgfgraphicnamed{TikZit//itensorandgt}
	\InputIfFileExists{TikZit//itensorandgt.tikz}{}{\input{./figures/TikZit//itensorandgt.tikz}}%
	\endpgfgraphicnamed
$$
 \end{proposition}   
     \begin{proof}
   $$  %
	\beginpgfgraphicnamed{TikZit//itensorandgtprf}
	\InputIfFileExists{TikZit//itensorandgtprf.tikz}{}{\input{./figures/TikZit//itensorandgtprf.tikz}}%
	\endpgfgraphicnamed
 $$
  \end{proof}
  \begin{corollary}\label{nlinestensornormalformadd}
 \begin{equation}\label{nlinestensornormalformaddeq}
	\beginpgfgraphicnamed{TikZit//nlinetensornormalformadd}
	\InputIfFileExists{TikZit//nlinetensornormalformadd.tikz}{}{\input{./figures/TikZit//nlinetensornormalformadd.tikz}}%
	\endpgfgraphicnamed

  \end{equation}
 where on the RHD of (\ref{nlinestensornormalformaddeq}), there are $2^n$ AND gates on the left-most $m$ wires, each AND gate  is accompanied by $k$ pairs of red $\pi$s on  the left-most $n$ wires with $0 \leq k  \leq n$, and different AND gates have different  distribution of pairs of red $\pi$s, that's why there are $\binom{n}{0}+\binom{n}{1} +\cdots + \binom{n}{n}=2^n$ AND gates.
 \end{corollary}

 
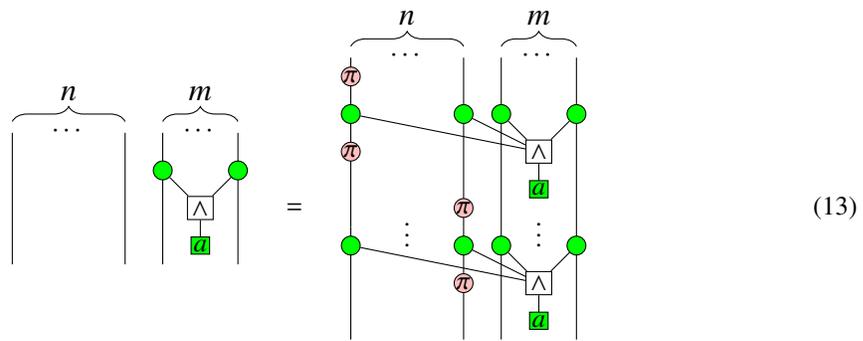
\begin{corollary}\label{nlinestensormmultiply}
 \begin{equation}\label{nlinestensormmultiplywq}
	\beginpgfgraphicnamed{TikZit//nlinestensormmultiplydm}
	\InputIfFileExists{TikZit//nlinestensormmultiplydm.tikz}{}{\input{./figures/TikZit//nlinestensormmultiplydm.tikz}}%
	\endpgfgraphicnamed

  \end{equation}
 where on the RHD of (\ref{nlinestensormmultiplywq}), there are $2^n$ AND gates on the right-most $m$ wires,each AND gate  is accompanied by $k$ pairs of red $\pi$s on  the left-most $n$ wires with $0 \leq k  \leq n$, and different AND gates have different  distribution of pairs of red $\pi$s, that's why there are $\binom{n}{0}+\binom{n}{1} +\cdots + \binom{n}{n}=2^n$ AND gates.
 \end{corollary}


     \begin{lemma}\label{raddcomplex}
	\beginpgfgraphicnamed{TikZit//raddcomutecomplex}
	\InputIfFileExists{TikZit//raddcomutecomplex.tikz}{}{\input{./figures/TikZit//raddcomutecomplex.tikz}}%
	\endpgfgraphicnamed

 \end{lemma}
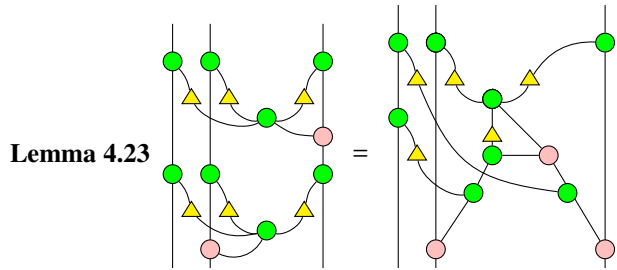
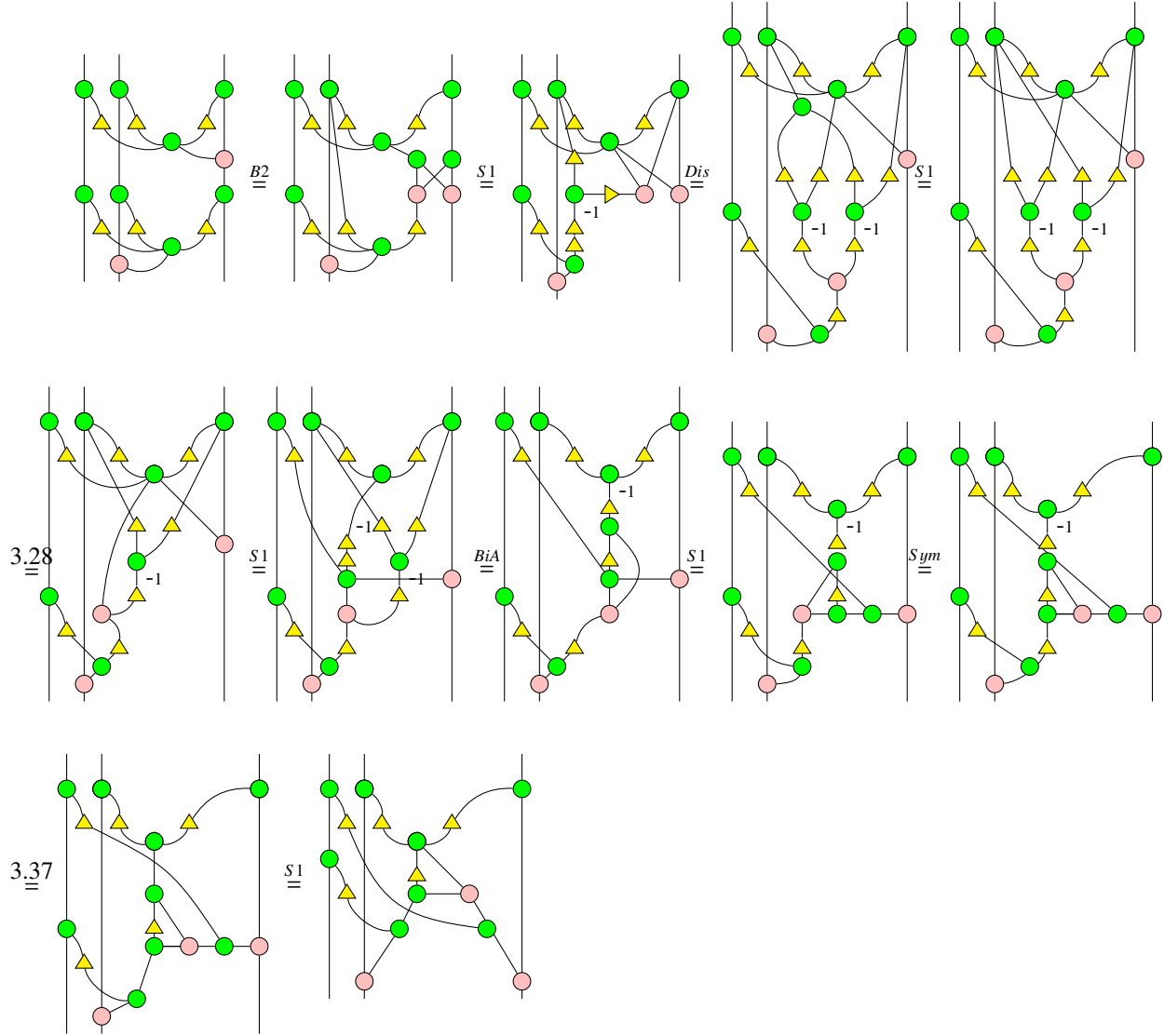
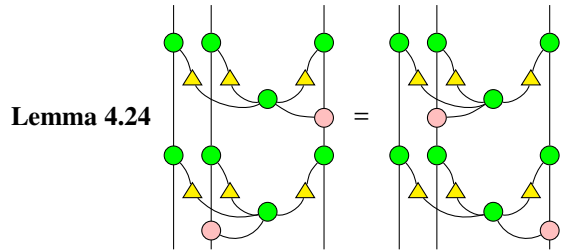
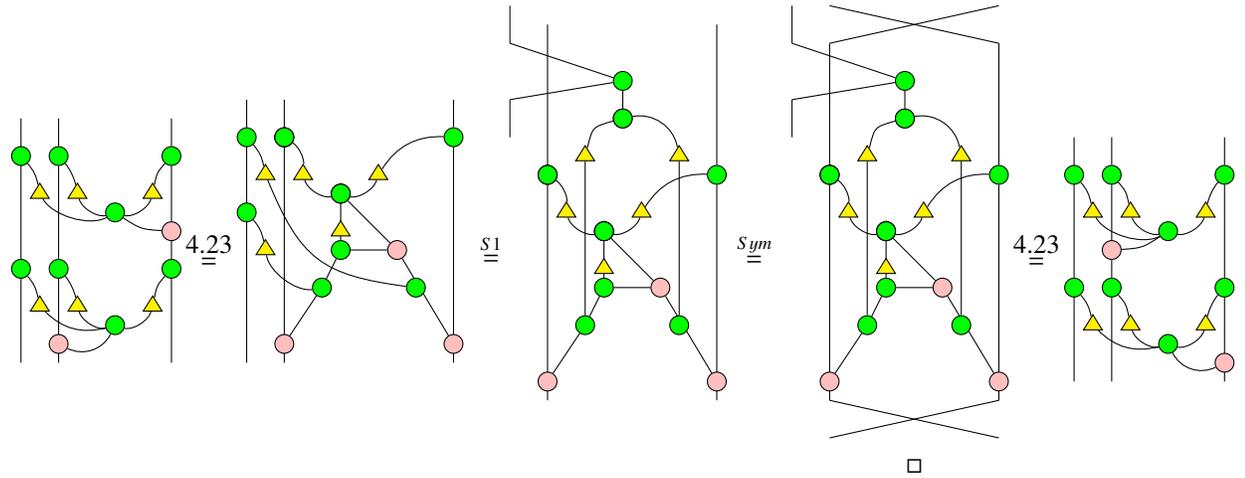
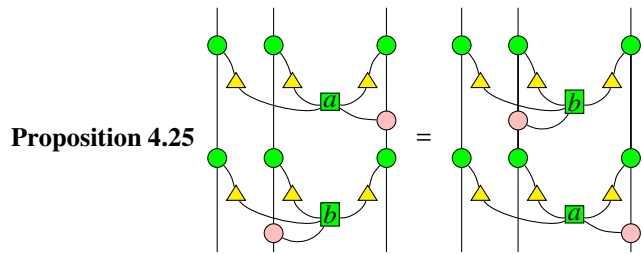
   \begin{proof}
   $$  %
	\beginpgfgraphicnamed{TikZit//raddcomutecomplexprf}
	\InputIfFileExists{TikZit//raddcomutecomplexprf.tikz}{}{\input{./figures/TikZit//raddcomutecomplexprf.tikz}}%
	\endpgfgraphicnamed
 $$
  \end{proof}

    \begin{lemma}\label{raddcomplexsym}
	\beginpgfgraphicnamed{TikZit//raddcomutecomplexsymtry}
	\InputIfFileExists{TikZit//raddcomutecomplexsymtry.tikz}{}{\input{./figures/TikZit//raddcomutecomplexsymtry.tikz}}%
	\endpgfgraphicnamed

 \end{lemma}
   \begin{proof}
   $$  %
	\beginpgfgraphicnamed{TikZit//raddcomutecomplexsymtryprf}
	\InputIfFileExists{TikZit//raddcomutecomplexsymtryprf.tikz}{}{\input{./figures/TikZit//raddcomutecomplexsymtryprf.tikz}}%
	\endpgfgraphicnamed
 $$
  \end{proof}

     \begin{proposition}\label{addcommutat}
	\beginpgfgraphicnamed{TikZit//addcommutation}
	\InputIfFileExists{TikZit//addcommutation.tikz}{}{\input{./figures/TikZit//addcommutation.tikz}}%
	\endpgfgraphicnamed
 \end{proposition}
  
   \begin{proof}
    If $a=0$ or $b=0$, then the equality holds trivially. Now we assume $ab\neq 0$. Then
   $$  %
	\beginpgfgraphicnamed{TikZit//addcommutationprof}
	\InputIfFileExists{TikZit//addcommutationprof.tikz}{}{\input{./figures/TikZit//addcommutationprof.tikz}}%
	\endpgfgraphicnamed
 $$
  \end{proof}

     \begin{proposition}\label{addcommutatgen}
     Let $n \geq 0$. Then we have
$$%
	\beginpgfgraphicnamed{TikZit//addcommutationgen}
	\InputIfFileExists{TikZit//addcommutationgen.tikz}{}{\input{./figures/TikZit//addcommutationgen.tikz}}%
	\endpgfgraphicnamed
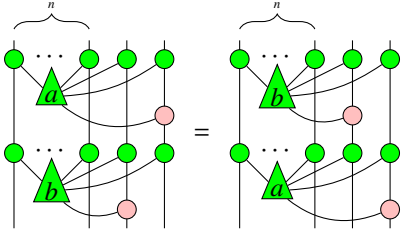
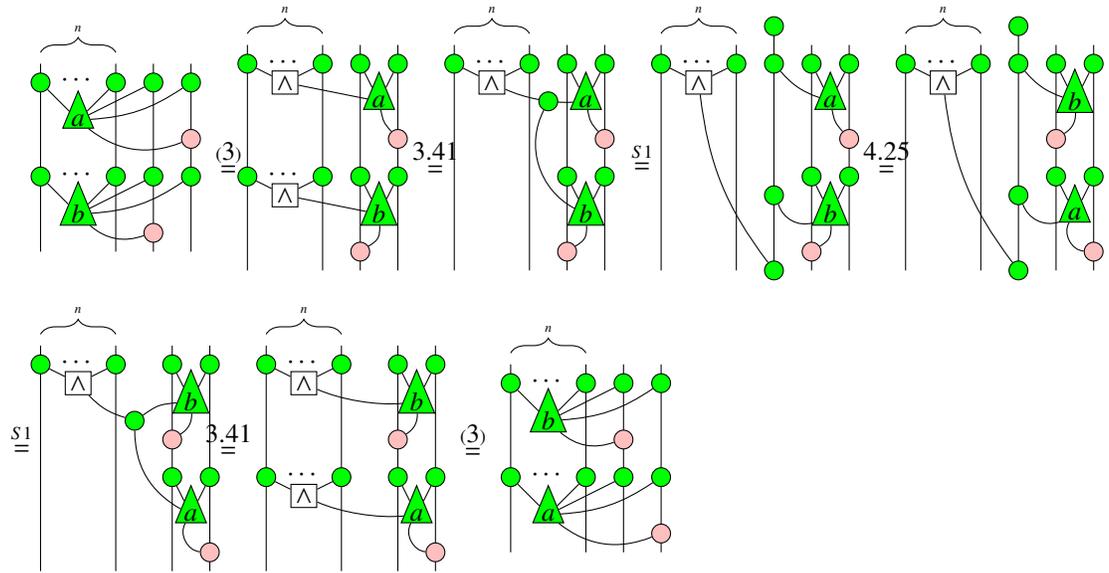
 $$
 \end{proposition}
  \begin{proof}
    $$  %
	\beginpgfgraphicnamed{TikZit//addcommutationgenprf}
	\InputIfFileExists{TikZit//addcommutationgenprf.tikz}{}{\input{./figures/TikZit//addcommutationgenprf.tikz}}%
	\endpgfgraphicnamed
 $$
  \end{proof}

       \begin{proposition}\label{addcommutatgencont}
     Assume that  node $a$ is connected to $ j_1, \cdots,  j_s $ via pink nodes and node $b$  is connected to $ i_1, \cdots,  i_t $ via pink nodes, where $i_1, \cdots,  i_t,  j_1, \cdots,  j_s \in \{1, \cdots, m\}$, and $\{i_1, \cdots,  i_t\} \neq \emptyset, \{ j_1, \cdots,  j_s\}\neq \emptyset$. Then we have
       \begin{equation}\label{addcommutatgenconteq}
	\beginpgfgraphicnamed{TikZit//addcommutationgenconct}
	\InputIfFileExists{TikZit//addcommutationgenconct.tikz}{}{\input{./figures/TikZit//addcommutationgenconct.tikz}}%
	\endpgfgraphicnamed
 
\end{equation}
 \end{proposition}
   \begin{proof}
   If $\{i_1, \cdots,  i_t\}=\{ j_1, \cdots,  j_s\} $, then  (\ref{addcommutatgenconteq}) follows directly from Corollary \ref{propadprimecro}. Below we assume that $\{i_1, \cdots,  i_t\}\neq\{ j_1, \cdots,  j_s\} $. If $\{i_1, \cdots,  i_t\} \subseteq\{ j_1, \cdots,  j_s\} $, then we assume w.l.o.g. that  $i_1=j_1, i_t=j_s,  j_2\notin \{i_1, \cdots,  i_t\}$, then we have 
        $$ %
	\beginpgfgraphicnamed{TikZit//addcommutationgenconctpf1}
	\InputIfFileExists{TikZit//addcommutationgenconctpf1.tikz}{}{\input{./figures/TikZit//addcommutationgenconctpf1.tikz}}%
	\endpgfgraphicnamed
 $$
     If symmetrically $\{ j_1, \cdots,  j_s\}\subseteq \{i_1, \cdots,  i_t\} $, 
then it can be proved similarly as the last case.
      If $\{i_1, \cdots,  i_t\} \nsubseteq\{ j_1, \cdots,  j_s\} $ and $\{ j_1, \cdots,  j_s\}\nsubseteq \{i_1, \cdots,  i_t\} $, then we assume w.l.o.g. that $j_1\notin \{i_1, \cdots,  i_t\}, i_1\notin \{ j_1, \cdots,  j_s\} $. Then we have 
  $$ %
	\beginpgfgraphicnamed{TikZit//addcommutationgenconctpf}
	\InputIfFileExists{TikZit//addcommutationgenconctpf.tikz}{}{\input{./figures/TikZit//addcommutationgenconctpf.tikz}}%
	\endpgfgraphicnamed
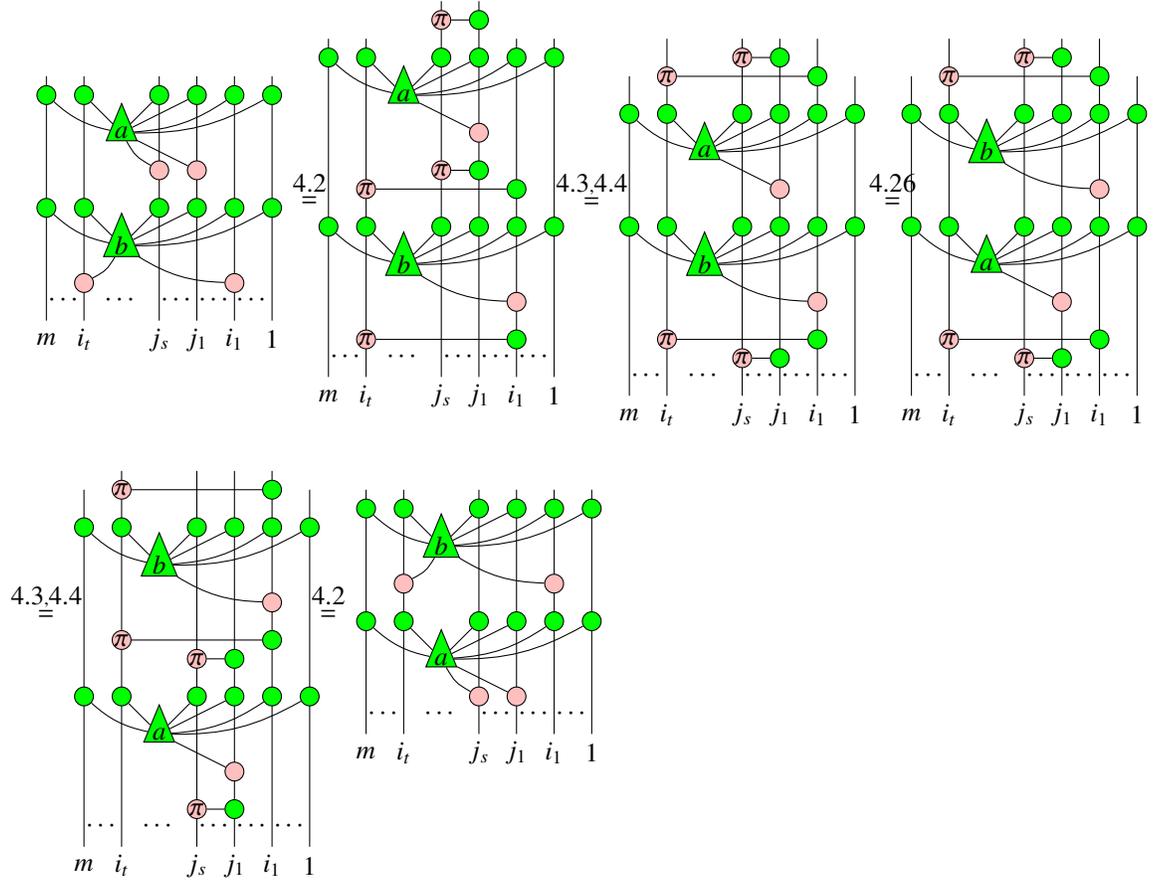
 $$
  \end{proof}
  



 \begin{proposition}\label{multiplypimulticommutg}
 Assume that $n \geq 0$. Then
$$%
	\beginpgfgraphicnamed{TikZit//multiplypimulticommuteg}
	\InputIfFileExists{TikZit//multiplypimulticommuteg.tikz}{}{\input{./figures/TikZit//multiplypimulticommuteg.tikz}}%
	\endpgfgraphicnamed
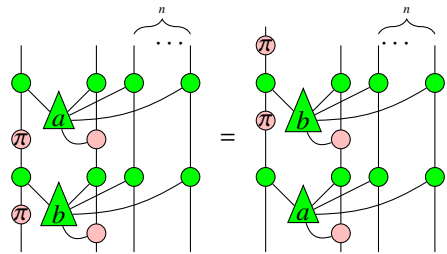
  $$
 \end{proposition}
\begin{proof}
Let $x=\frac{a+b}{2}, \quad y=\frac{b-a}{2}$. Then $a=x-y,\quad b=x+y$.
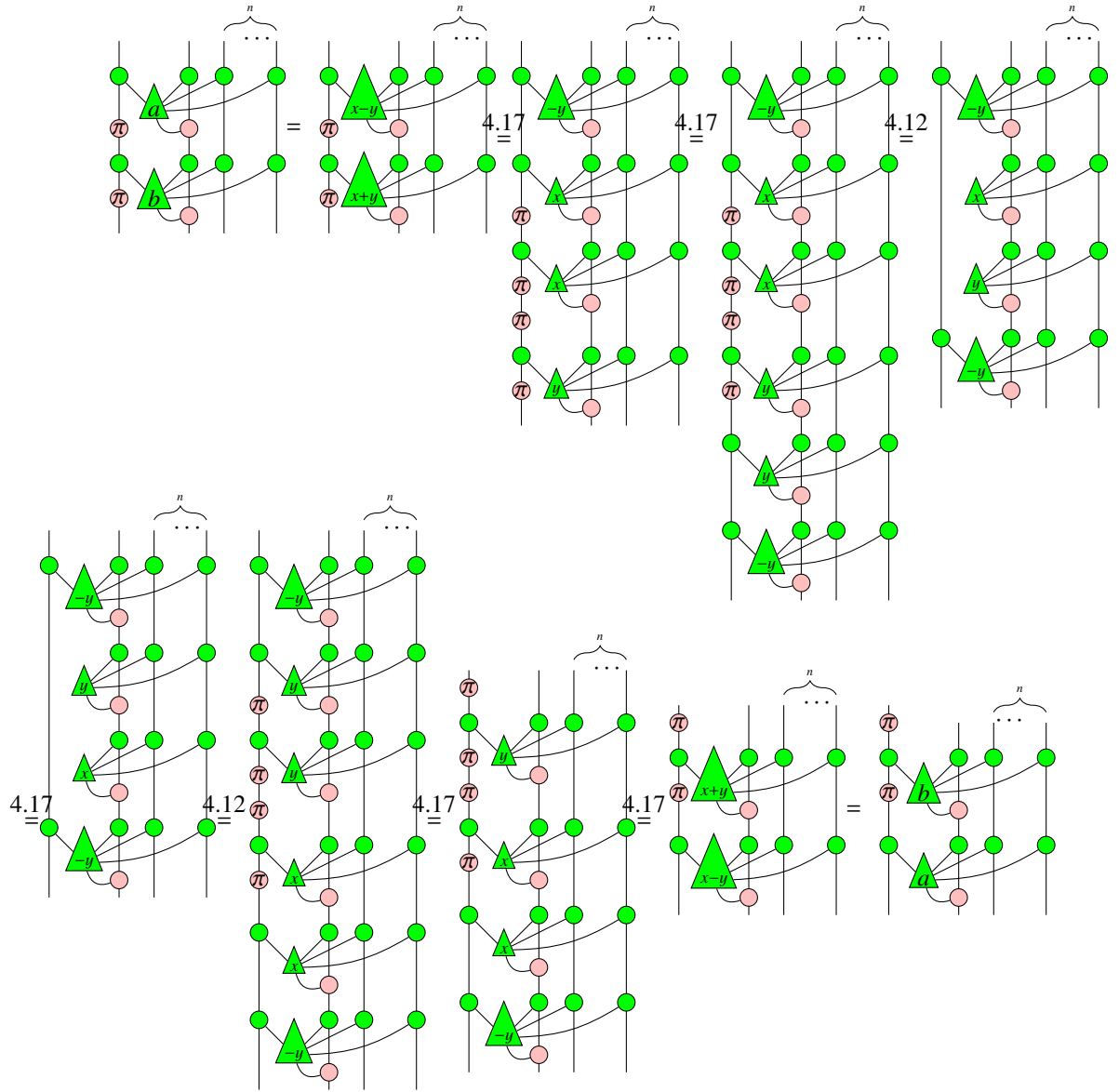
$$  %
	\beginpgfgraphicnamed{TikZit//multiplypimulticommutegprf}
	\InputIfFileExists{TikZit//multiplypimulticommutegprf.tikz}{}{\input{./figures/TikZit//multiplypimulticommutegprf.tikz}}%
	\endpgfgraphicnamed
 $$
 \end{proof}
 
 \begin{corollary}\label{multiplypimulticommutgcro}
  Assume that $n \geq 0$. Then
 $$ %
	\beginpgfgraphicnamed{TikZit//multiplypimulticommutegcro}
	\InputIfFileExists{TikZit//multiplypimulticommutegcro.tikz}{}{\input{./figures/TikZit//multiplypimulticommutegcro.tikz}}%
	\endpgfgraphicnamed
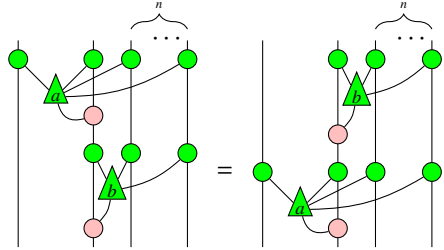
$$
 \end{corollary}
This follows directly from Proposition \ref{multiplypimulticommutg} and Corollary \ref{propo1cro2}.

 \begin{corollary}\label{multiplypimulticommutgcro2}
$$ %
	\beginpgfgraphicnamed{TikZit//multiplypimulticommutegcro2}
	\InputIfFileExists{TikZit//multiplypimulticommutegcro2.tikz}{}{\input{./figures/TikZit//multiplypimulticommutegcro2.tikz}}%
	\endpgfgraphicnamed
$$
where the node $a$ and $b$ are connected to $ j_1, \cdots,  j_s $ via pink nodes,  and two red $\pi$ nodes are located on the $i$-th line, $i\notin \{j_1, \cdots,  j_s\}$ or $i\in \{j_1, \cdots,  j_s\}, |\{j_1, \cdots,  j_s\}|\geq 2$. 
 \end{corollary}
\begin{proof}
If $i\notin \{j_1, \cdots,  j_s\}$, then we have
$$  %
	\beginpgfgraphicnamed{TikZit//multiplypimulticommutegcro2prf}
	\InputIfFileExists{TikZit//multiplypimulticommutegcro2prf.tikz}{}{\input{./figures/TikZit//multiplypimulticommutegcro2prf.tikz}}%
	\endpgfgraphicnamed
 $$
If $i\in \{j_1, \cdots,  j_s\}, |\{j_1, \cdots,  j_s\}|\geq 2$, we assume w.l.g that $i=j_s \neq j_1$. Then we have
$$  %
	\beginpgfgraphicnamed{TikZit//multiplypimulticommutegcro2prfc2}
	\InputIfFileExists{TikZit//multiplypimulticommutegcro2prfc2.tikz}{}{\input{./figures/TikZit//multiplypimulticommutegcro2prfc2.tikz}}%
	\endpgfgraphicnamed
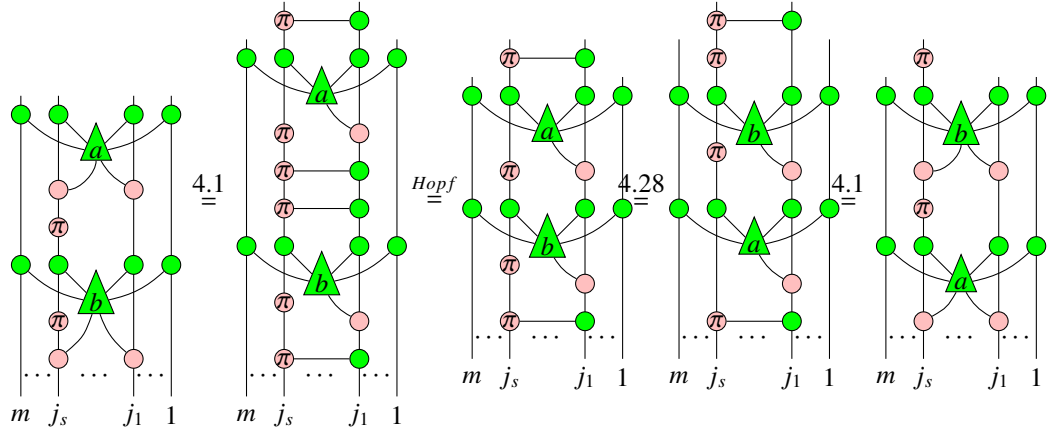
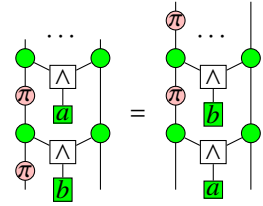
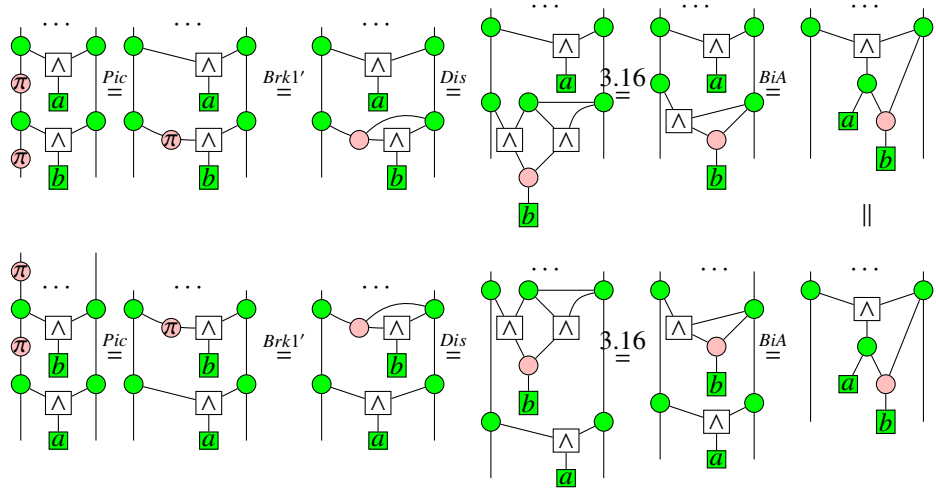
 $$
 \end{proof}


 \begin{proposition}\label{andpicomt}
 $$  %
	\beginpgfgraphicnamed{TikZit//andpicomute}
	\InputIfFileExists{TikZit//andpicomute.tikz}{}{\input{./figures/TikZit//andpicomute.tikz}}%
	\endpgfgraphicnamed
 $$
\end{proposition}
\begin{proof}
$$  %
	\beginpgfgraphicnamed{TikZit//andpicomuteprf}
	\InputIfFileExists{TikZit//andpicomuteprf.tikz}{}{\input{./figures/TikZit//andpicomuteprf.tikz}}%
	\endpgfgraphicnamed
 $$
 \end{proof}

  \begin{proposition}\label{addpidoublecom}
Suppose  the nodes $a$ and $b$ are connected to $ j_1, \cdots,  j_s $ via pink nodes,  pairs of red $\pi$ nodes separated by green nodes connected to $a$ are located on  $ i_1, \cdots,  i_t $, and pairs of red $\pi$ nodes separated by green nodes connected to $b$ are located on  $ k_1, \cdots,  k_l $, $ \{i_1, \cdots,  i_t\} \cap \{j_1, \cdots,  j_s\}= \emptyset, \{k_1, \cdots,  k_l\} \cap \{j_1, \cdots,  j_s\}= \emptyset. $ Then we have
 $$  %
	\beginpgfgraphicnamed{TikZit//addpidoublecomte}
	\InputIfFileExists{TikZit//addpidoublecomte.tikz}{}{\input{./figures/TikZit//addpidoublecomte.tikz}}%
	\endpgfgraphicnamed
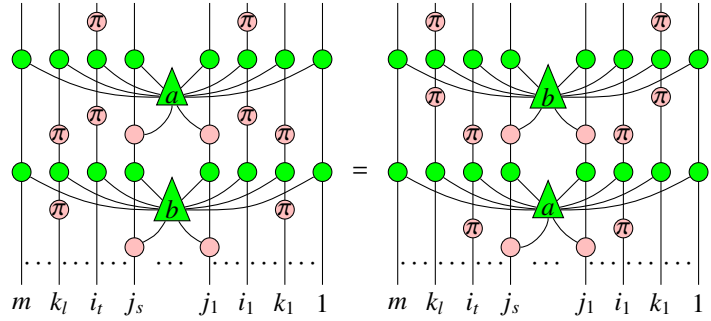
 $$
\end{proposition}
 \begin{proof}
 If $\{i_1, \cdots,  i_t\} \cup \{k_1, \cdots,  k_l\} = \emptyset$, then it is just the case of Corollary \ref{propadprimecro}. Otherwise, we assume w.l.g that $| \{k_1, \cdots,  k_l\} |\geq 1$. Then we have
$$  %
	\beginpgfgraphicnamed{TikZit//addpidoublecomteprf}
	\InputIfFileExists{TikZit//addpidoublecomteprf.tikz}{}{\input{./figures/TikZit//addpidoublecomteprf.tikz}}%
	\endpgfgraphicnamed
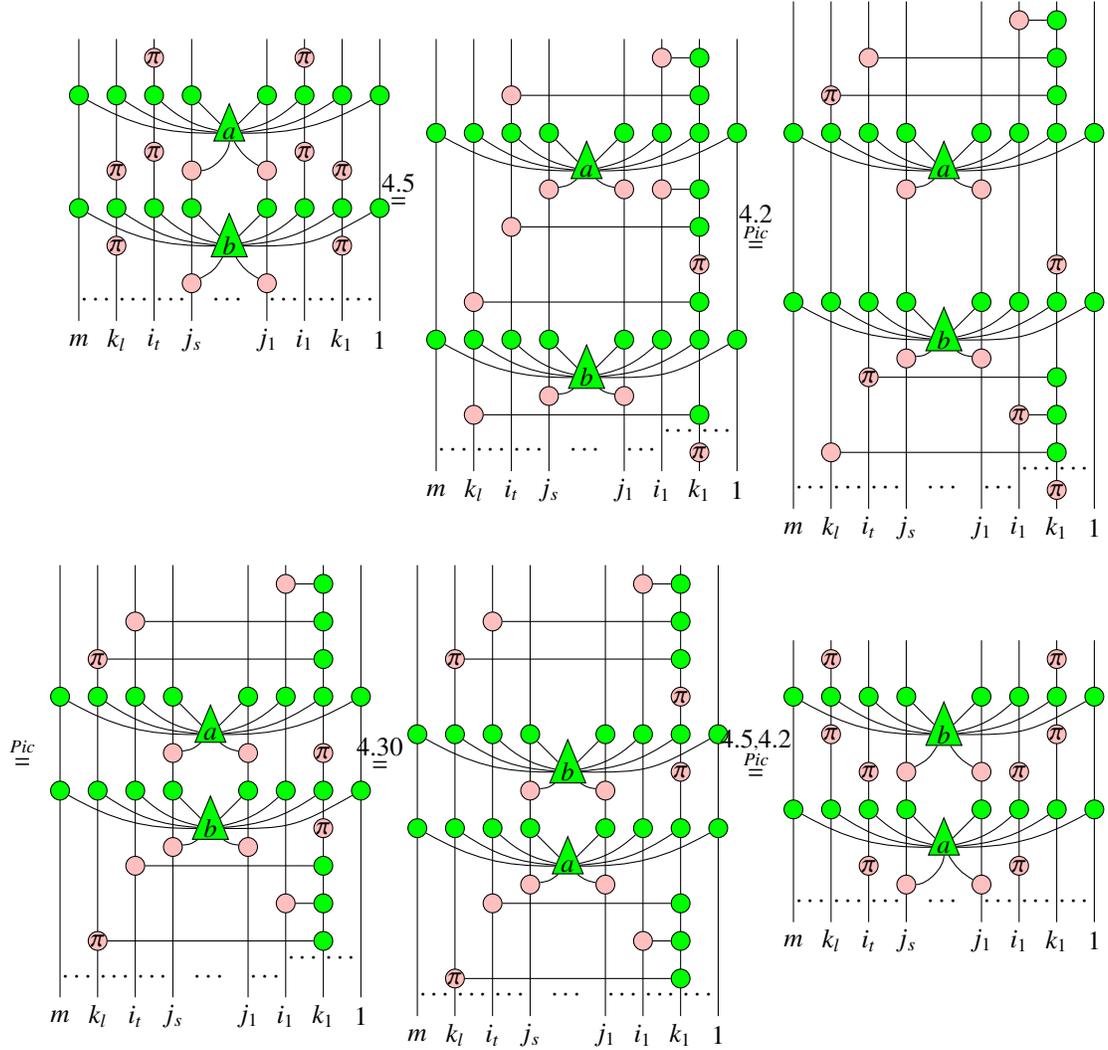
 $$
 \end{proof}
 
  \begin{proposition}\label{multipidoublecom}
Suppose  the  pairs of red $\pi$ nodes separated by green nodes connected to $a$ are located on  $ i_1, \cdots,  i_t $, and pairs of red $\pi$ nodes separated by green nodes connected to $b$ are located on  $ j_1, \cdots,  j_s $. Then we have
 $$  %
	\beginpgfgraphicnamed{TikZit//multipidoublecomte}
	\InputIfFileExists{TikZit//multipidoublecomte.tikz}{}{\input{./figures/TikZit//multipidoublecomte.tikz}}%
	\endpgfgraphicnamed
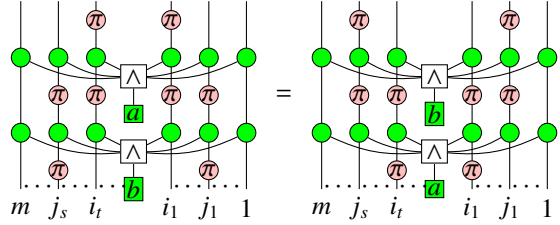
 $$
\end{proposition}
 \begin{proof}  
 If $\{i_1, \cdots,  i_t\} \cup \{j_1, \cdots,  j_s\} = \emptyset$, then it is just the case of Corollary \ref{propadprimecro}. Otherwise, it follows from Proposition\ref{andpicomt} using the same techniques as in the proof of  Proposition \ref{addpidoublecom}.
  \end{proof}

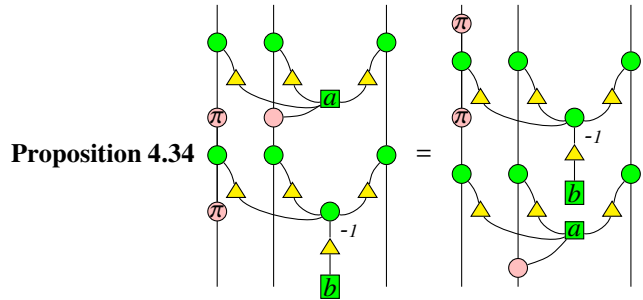
  \begin{proposition}\label{addpimultiplycommut}
	\beginpgfgraphicnamed{TikZit//addpimultiplycommute}
	\InputIfFileExists{TikZit//addpimultiplycommute.tikz}{}{\input{./figures/TikZit//addpimultiplycommute.tikz}}%
	\endpgfgraphicnamed
 
 \end{proposition}
    \begin{proof}
     If $a=0$, then the equality holds trivially. Now we assume $a\neq 0$. Then
   $$  %
	\beginpgfgraphicnamed{TikZit//addpimultiplycommuteprf}
	\InputIfFileExists{TikZit//addpimultiplycommuteprf.tikz}{}{\input{./figures/TikZit//addpimultiplycommuteprf.tikz}}%
	\endpgfgraphicnamed
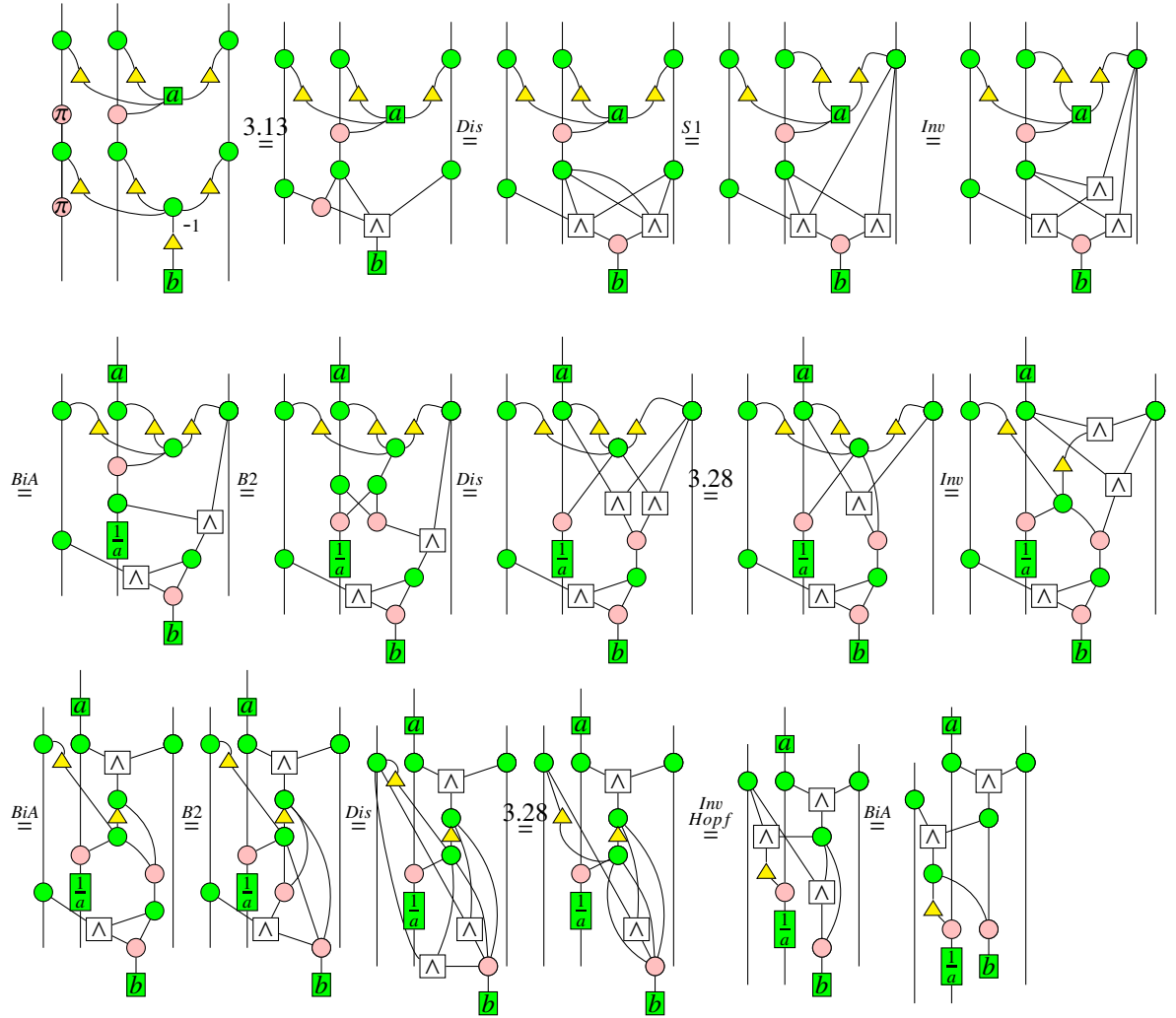
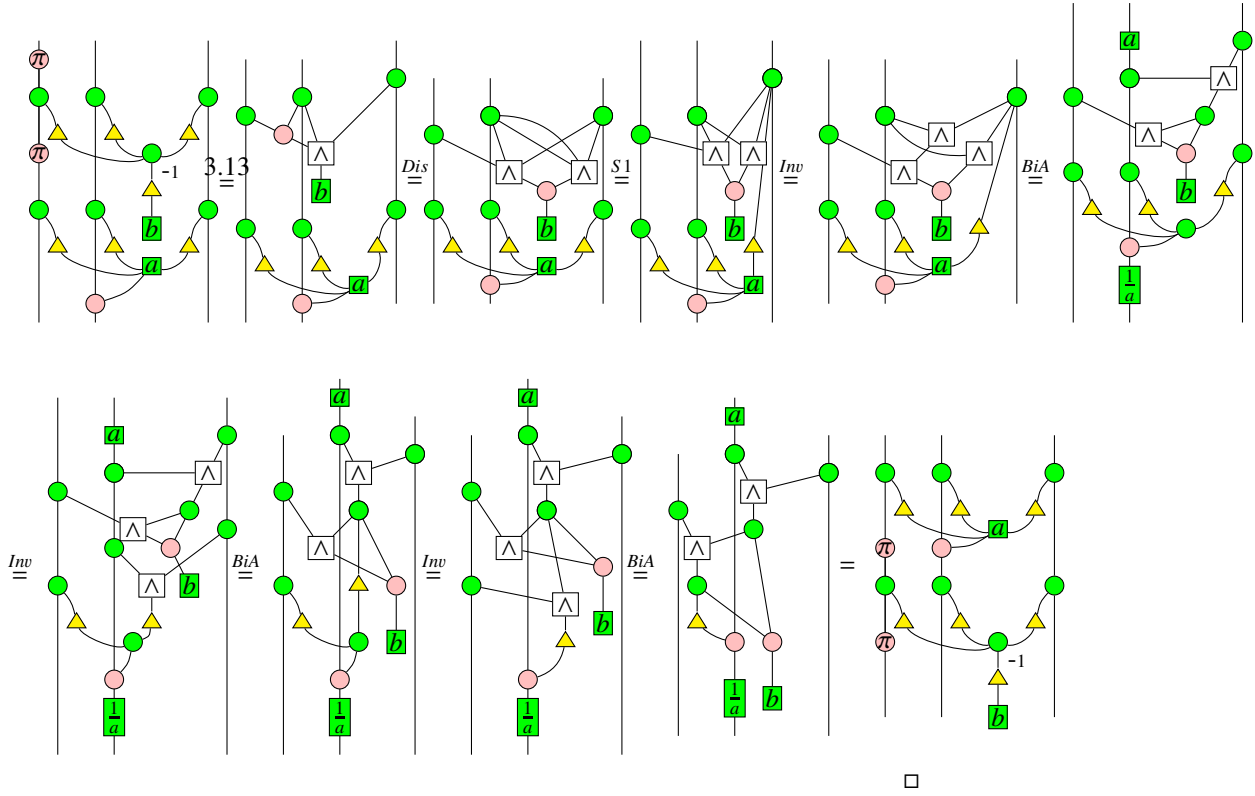
 $$
     $$  %
	\beginpgfgraphicnamed{TikZit//addpimultiplycommuteprf2}
	\InputIfFileExists{TikZit//addpimultiplycommuteprf2.tikz}{}{\input{./figures/TikZit//addpimultiplycommuteprf2.tikz}}%
	\endpgfgraphicnamed
 $$
  \end{proof}

  \begin{proposition}\label{addpimultiplycommutg}
 Suppose  the node $a$ is connected to $ j_1, \cdots,  j_s $ via pink nodes,  and a pair of red $\pi$ nodes separated by green nodes connected to $b$ are located on  $ k$, where $k\notin  \{j_1, \cdots,  j_s\},  |\{j_1, \cdots,  j_s\}|\geq 1$, or   $k\in  \{j_1, \cdots,  j_s\},  |\{j_1, \cdots,  j_s\}|\geq 2$. Then we have
  \begin{equation}\label{addpimultiplycommutgeq}
	\beginpgfgraphicnamed{TikZit//addpimultiplycommutegen}
	\InputIfFileExists{TikZit//addpimultiplycommutegen.tikz}{}{\input{./figures/TikZit//addpimultiplycommutegen.tikz}}%
	\endpgfgraphicnamed
 
  \end{equation}
 \end{proposition} 
  \begin{proof}
 If  $k\notin  \{j_1, \cdots,  j_s\},  |\{j_1, \cdots,  j_s\}|\geq 1$, then by rearranging the order of lines by swapping, (\ref{addpimultiplycommutgeq}) is equivalent to
 
  
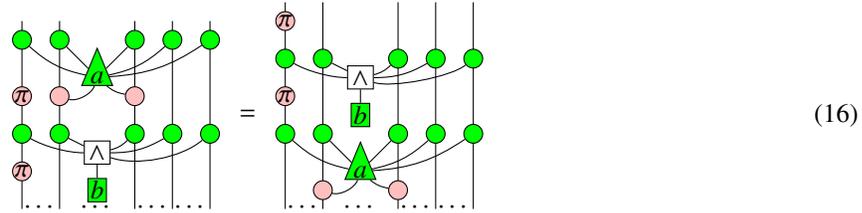
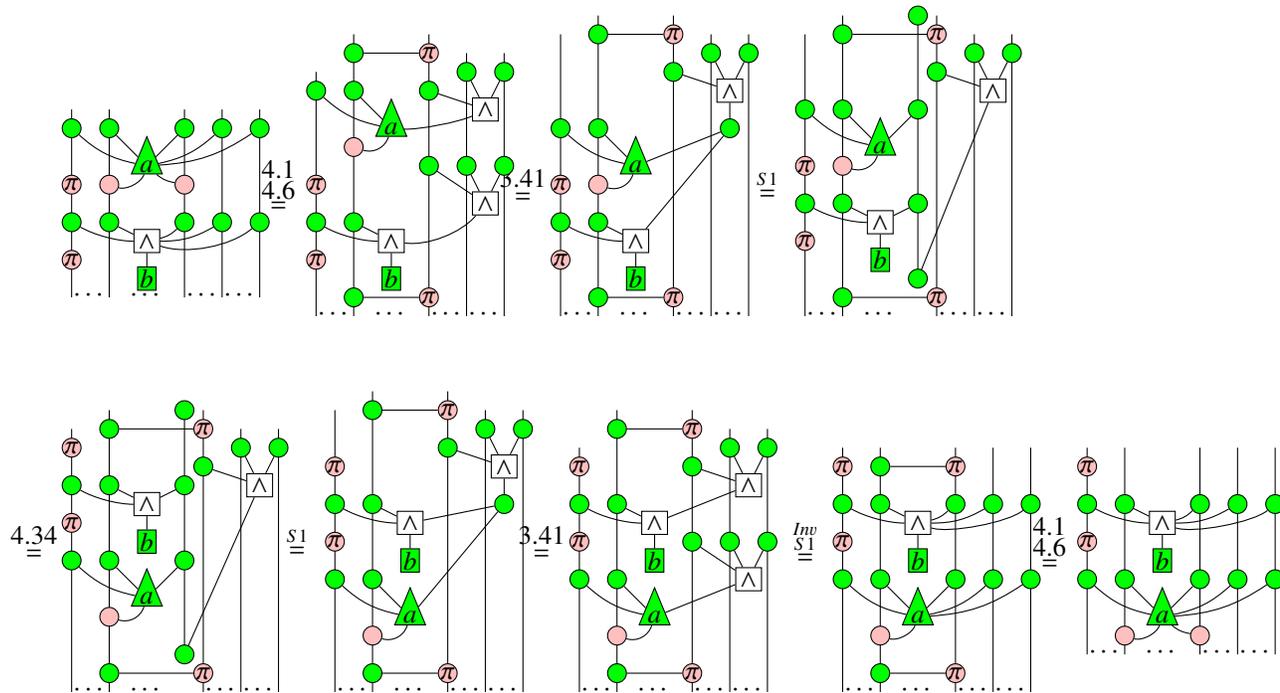
\begin{equation}\label{addpimultiplycommutegen2eq}
  $$%
	\beginpgfgraphicnamed{TikZit//addpimultiplycommutegen2}
	\InputIfFileExists{TikZit//addpimultiplycommutegen2.tikz}{}{\input{./figures/TikZit//addpimultiplycommutegen2.tikz}}%
	\endpgfgraphicnamed
 $$
    \end{equation}
 We have
   $$%
	\beginpgfgraphicnamed{TikZit//addpimultiplycommutegen2prf}
	\InputIfFileExists{TikZit//addpimultiplycommutegen2prf.tikz}{}{\input{./figures/TikZit//addpimultiplycommutegen2prf.tikz}}%
	\endpgfgraphicnamed
 $$
If   $k\in  \{j_1, \cdots,  j_s\},  |\{j_1, \cdots,  j_s\}|\geq 2$,  we 
assume w.l.g that $k=j_s$, then (\ref{addpimultiplycommutgeq}) is equivalent to
 
  
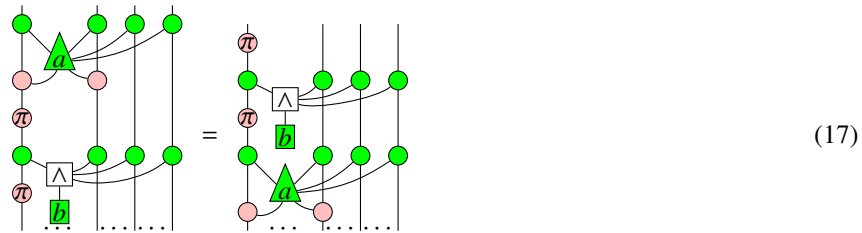
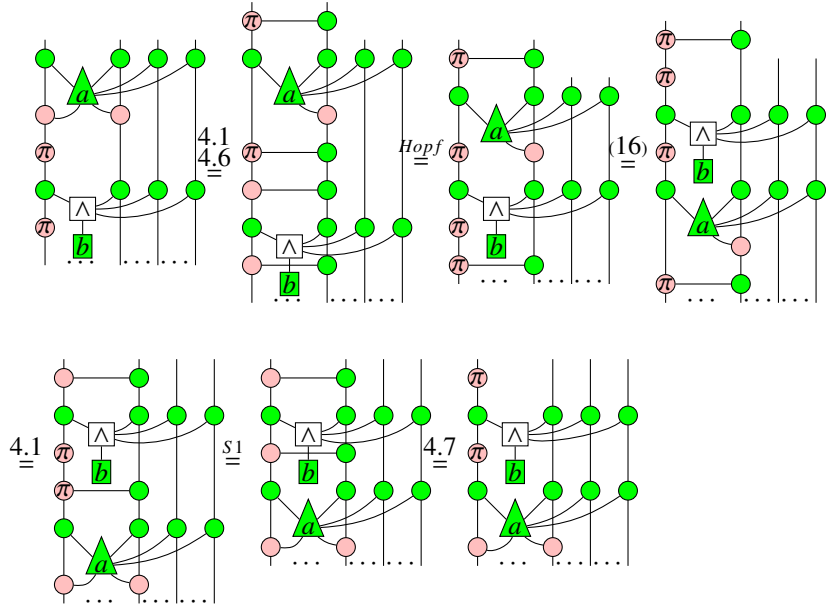
\begin{equation}\label{addpimultiplycommutegen3eq}
  $$%
	\beginpgfgraphicnamed{TikZit//addpimultiplycommutegen3}
	\InputIfFileExists{TikZit//addpimultiplycommutegen3.tikz}{}{\input{./figures/TikZit//addpimultiplycommutegen3.tikz}}%
	\endpgfgraphicnamed
 $$
    \end{equation}
We have
$$%
	\beginpgfgraphicnamed{TikZit//addpimultiplycommutegen3prf}
	\InputIfFileExists{TikZit//addpimultiplycommutegen3prf.tikz}{}{\input{./figures/TikZit//addpimultiplycommutegen3prf.tikz}}%
	\endpgfgraphicnamed
 $$
\end{proof}

 \begin{proposition}\label{addpipairmultiplycommutgp}
 Suppose  the node $a$ is connected to $ j_1, \cdots,  j_s $ via pink nodes,  pairs of red $\pi$ nodes separated by green nodes connected to $a$ are located on  $ i_1, \cdots,  i_t $, and pairs of red $\pi$ nodes separated by green nodes connected to $b$ are located on  $ k_1, \cdots,  k_l , \{i_1, \cdots,  i_t\} \neq  \{k_1, \cdots,  k_l\} $. Either 
$ \{k_1, \cdots,  k_l\} \neq \emptyset, \{k_1, \cdots,  k_l\} \cap \{j_1, \cdots,  j_s\}= \emptyset,  |\{j_1, \cdots,  j_s\}|\geq 1$; or $ \{k_1, \cdots,  k_l\} \neq \emptyset, |\{k_1, \cdots,  k_l\} \cap \{j_1, \cdots,  j_s\}|= 1,  |\{j_1, \cdots,  j_s\}|\geq 2$; or symmetrically, 
$ \{i_1, \cdots,  i_t\} \neq \emptyset, \{i_1, \cdots,  i_t\} \cap \{j_1, \cdots,  j_s\}= \emptyset,  |\{j_1, \cdots,  j_s\}|\geq 1$; or $ \{i_1, \cdots,  i_t\} \neq \emptyset, |\{i_1, \cdots,  i_t\} \cap \{j_1, \cdots,  j_s\}|= 1,  |\{j_1, \cdots,  j_s\}|\geq 2$.
  Then we have
$$%
	\beginpgfgraphicnamed{TikZit//addpipairmultiplycommutg}
	\InputIfFileExists{TikZit//addpipairmultiplycommutg.tikz}{}{\input{./figures/TikZit//addpipairmultiplycommutg.tikz}}%
	\endpgfgraphicnamed
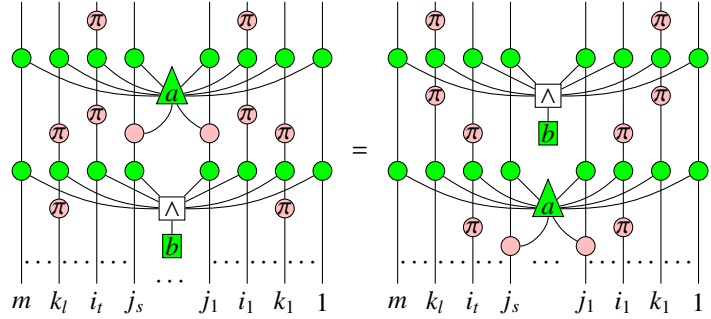
 $$
 \end{proposition}   
   \begin{proof}
    We assume w.l.g that $k_l\notin \{i_1, \cdots,  i_t\}$. If $k_l\notin  \{j_1, \cdots,  j_s\}$. Then we have 
      \begin{equation}\label{addpipairmultiplycommutgeq}
	\beginpgfgraphicnamed{TikZit//addpipairmultiplycommutgprf1}
	\InputIfFileExists{TikZit//addpipairmultiplycommutgprf1.tikz}{}{\input{./figures/TikZit//addpipairmultiplycommutgprf1.tikz}}%
	\endpgfgraphicnamed
 
  \end{equation}
If   $ |\{j_1, \cdots,  j_s\}|\geq 2$, $k_l\in  \{j_1, \cdots,  j_s\}$, we  assume w.l.g that $k_l=j_s\neq j_1$.
Then we have
$$%
	\beginpgfgraphicnamed{TikZit//addpipairmultiplycommutgprf2}
	\InputIfFileExists{TikZit//addpipairmultiplycommutgprf2.tikz}{}{\input{./figures/TikZit//addpipairmultiplycommutgprf2.tikz}}%
	\endpgfgraphicnamed
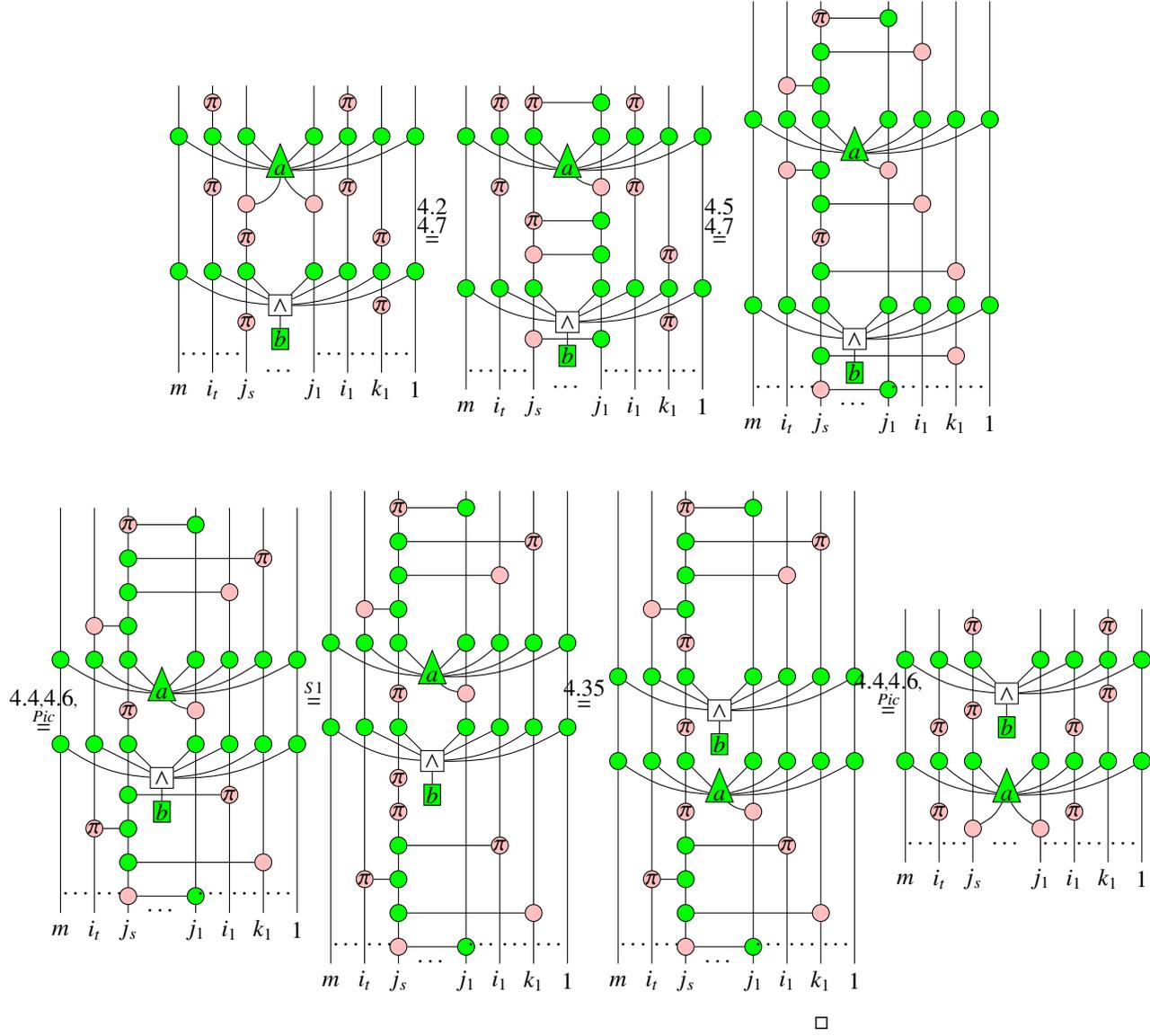
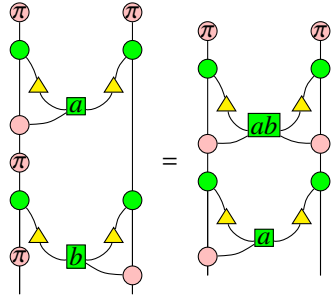
 $$
 \end{proof}

    \begin{proposition}\label{TR15}
$$ %
	\beginpgfgraphicnamed{TikZit//pimultiplyabsorbtion}
	\InputIfFileExists{TikZit//pimultiplyabsorbtion.tikz}{}{\input{./figures/TikZit//pimultiplyabsorbtion.tikz}}%
	\endpgfgraphicnamed
 $$
 \end{proposition}
 \begin{proof}
 If $ab=0$, it can be easily verified. Below we assume that $ab\neq 0$. Then we have
$$  %
	\beginpgfgraphicnamed{TikZit//pimultiplyabsorbtionprf}
	\InputIfFileExists{TikZit//pimultiplyabsorbtionprf.tikz}{}{\input{./figures/TikZit//pimultiplyabsorbtionprf.tikz}}%
	\endpgfgraphicnamed
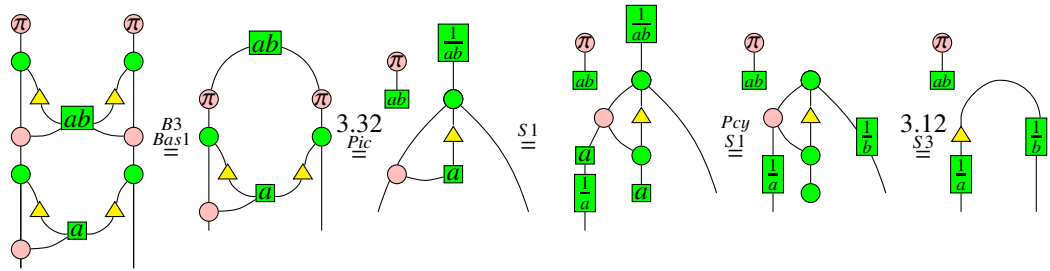
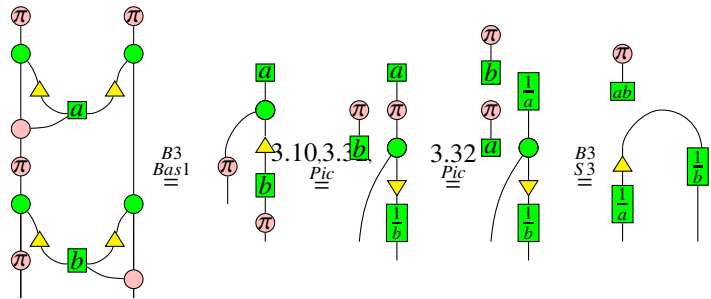
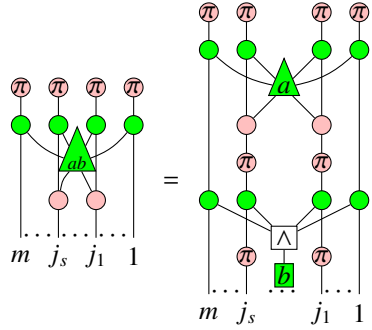
 $$
 \end{proof}

    \begin{proposition}\label{pimultiaddcombinepro}
$$ %
	\beginpgfgraphicnamed{TikZit//pimultiaddcombine}
	\InputIfFileExists{TikZit//pimultiaddcombine.tikz}{}{\input{./figures/TikZit//pimultiaddcombine.tikz}}%
	\endpgfgraphicnamed
 $$
where the node $ab$ is connected to $ j_1, \cdots,  j_s $ via pink nodes.
 \end{proposition}
 
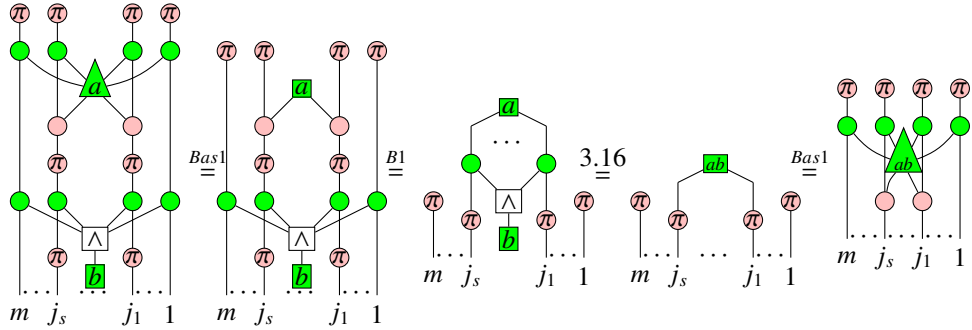
\begin{proof}
$$  %
	\beginpgfgraphicnamed{TikZit//pimultiaddcombineprf}
	\InputIfFileExists{TikZit//pimultiaddcombineprf.tikz}{}{\input{./figures/TikZit//pimultiaddcombineprf.tikz}}%
	\endpgfgraphicnamed
 $$
 \end{proof}

  \begin{proposition}\label{pitopaddpipaircommutprop}
  Suppose  the node $a$ is connected to $ j_1, \cdots,  j_s $ via pink nodes, the node $b$ is connected to $ i_1, \cdots,  i_t $ via pink nodes,  pairs of red $\pi$ nodes separated by green nodes connected to $b$ are located on  $ j_1, \cdots,  j_s $. Furthermore,  $\emptyset\neq \{i_1, \cdots,  i_t\} \subseteq  \{1, \cdots,  n\}, \emptyset\neq \{j_1, \cdots,  j_s\} \subseteq  \{n+1, \cdots,  n+m\}$.
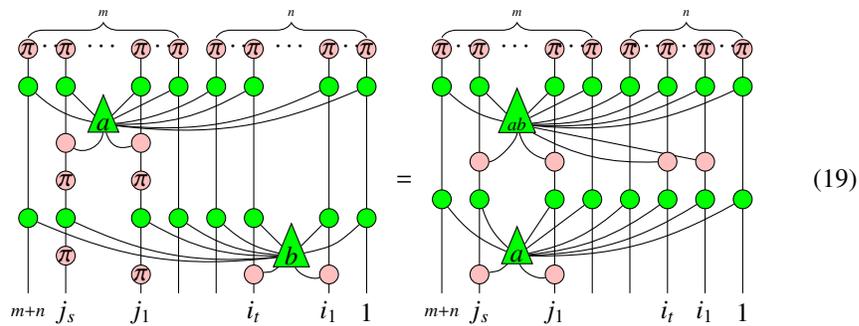
     \begin{equation}\label{pitopaddpipaircommuteq}
	\beginpgfgraphicnamed{TikZit//pitopaddpipaircommut}
	\InputIfFileExists{TikZit//pitopaddpipaircommut.tikz}{}{\input{./figures/TikZit//pitopaddpipaircommut.tikz}}%
	\endpgfgraphicnamed
 
    \end{equation}
  \end{proposition}  
   \begin{proof}
   If $ab=0$, then it is easy to prove. We assume that $ab\neq 0$. Then
$$  %
	\beginpgfgraphicnamed{TikZit//pitopaddpipaircommutprf}
	\InputIfFileExists{TikZit//pitopaddpipaircommutprf.tikz}{}{\input{./figures/TikZit//pitopaddpipaircommutprf.tikz}}%
	\endpgfgraphicnamed
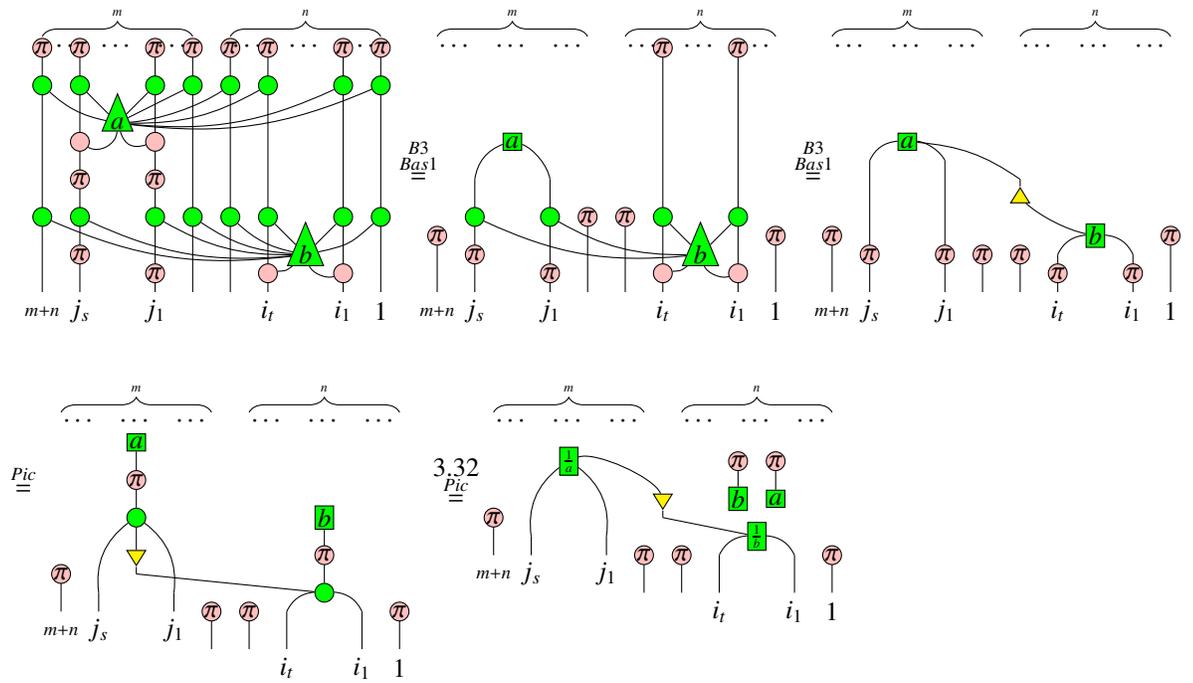
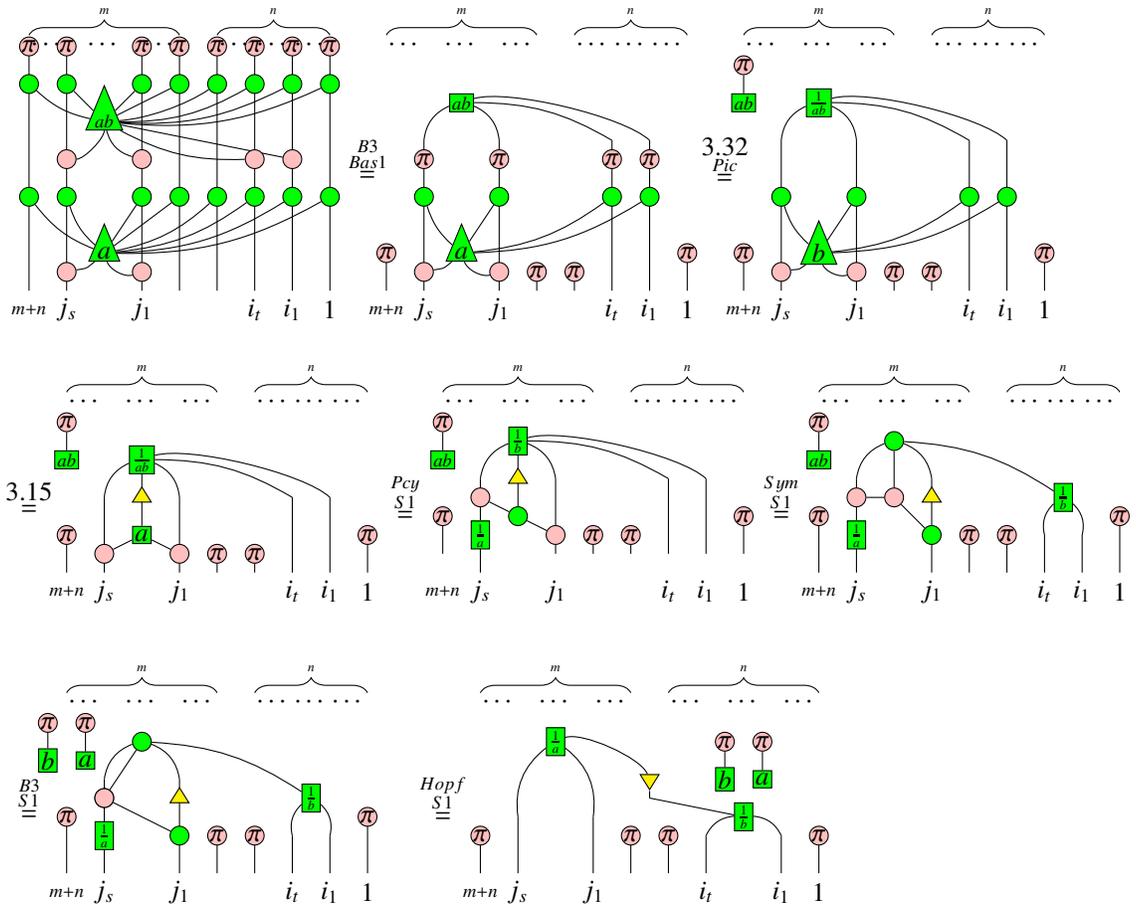
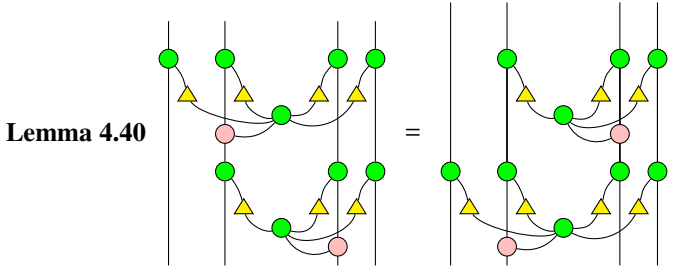
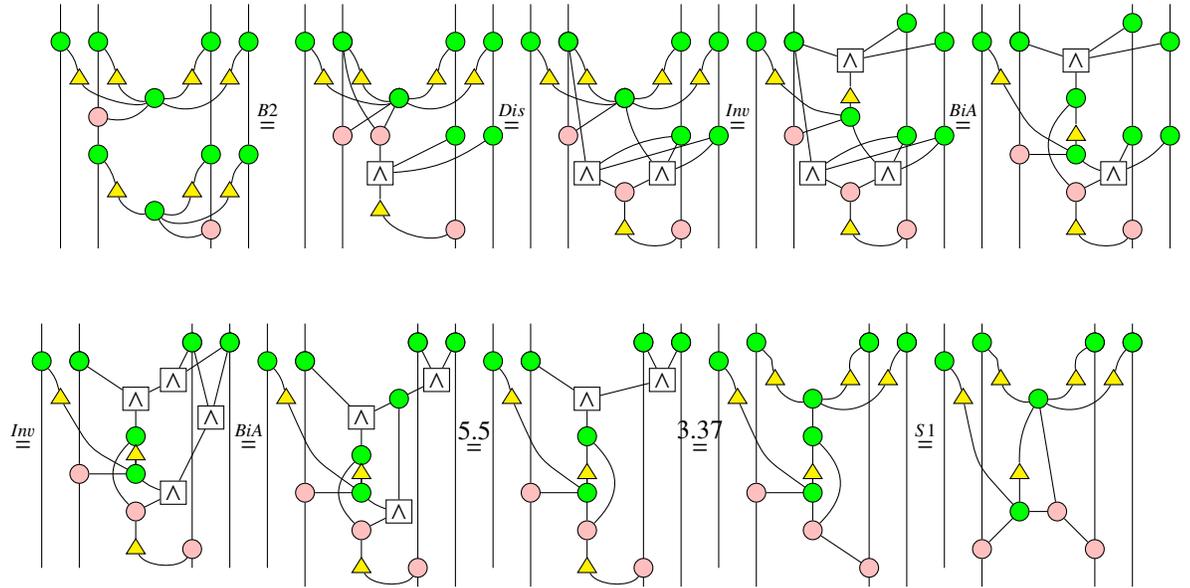
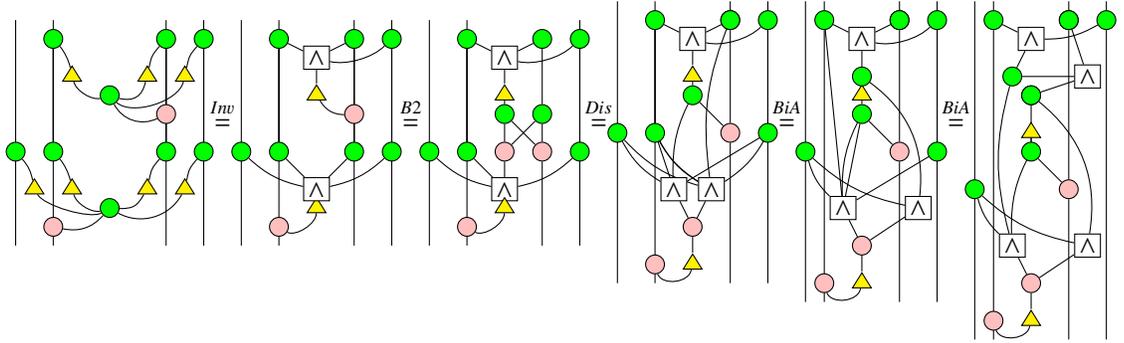
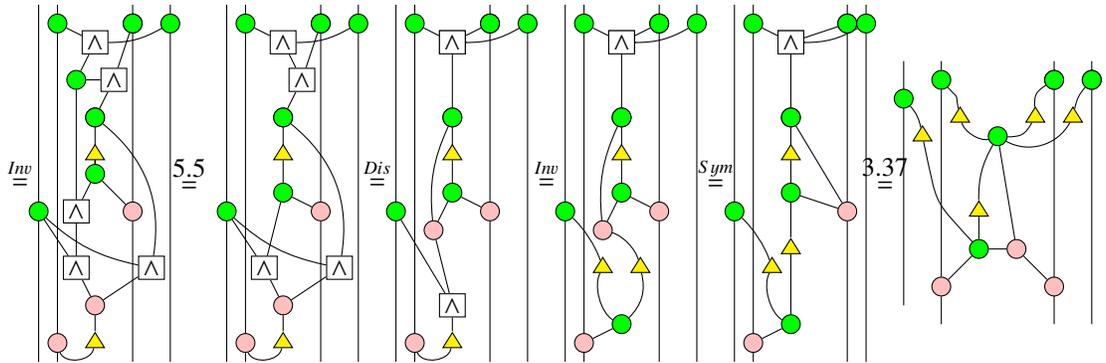
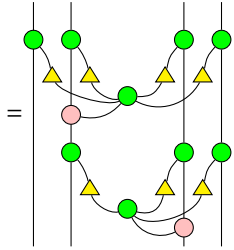
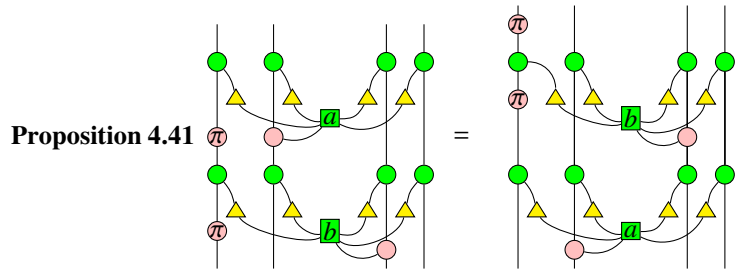
 $$
$$  %
	\beginpgfgraphicnamed{TikZit//pitopaddpipaircommutprf2}
	\InputIfFileExists{TikZit//pitopaddpipaircommutprf2.tikz}{}{\input{./figures/TikZit//pitopaddpipaircommutprf2.tikz}}%
	\endpgfgraphicnamed
 $$
 \end{proof}

   \begin{lemma}\label{multiplypimulticommutesim}
	\beginpgfgraphicnamed{TikZit//multiplypimulticommutesimp}
	\InputIfFileExists{TikZit//multiplypimulticommutesimp.tikz}{}{\input{./figures/TikZit//multiplypimulticommutesimp.tikz}}%
	\endpgfgraphicnamed
 
         \end{lemma}
   \begin{proof}
      $$  %
	\beginpgfgraphicnamed{TikZit//multiplypimulticommutesimpprf}
	\InputIfFileExists{TikZit//multiplypimulticommutesimpprf.tikz}{}{\input{./figures/TikZit//multiplypimulticommutesimpprf.tikz}}%
	\endpgfgraphicnamed
 $$
       $$  %
	\beginpgfgraphicnamed{TikZit//multiplypimulticommutesimpprf2}
	\InputIfFileExists{TikZit//multiplypimulticommutesimpprf2.tikz}{}{\input{./figures/TikZit//multiplypimulticommutesimpprf2.tikz}}%
	\endpgfgraphicnamed
 $$

  \end{proof}

  \begin{proposition}\label{multiplypimulticommute}
	\beginpgfgraphicnamed{TikZit//multiplypimulticommute}
	\InputIfFileExists{TikZit//multiplypimulticommute.tikz}{}{\input{./figures/TikZit//multiplypimulticommute.tikz}}%
	\endpgfgraphicnamed
         \end{proposition}
            \begin{proof}
              If $a=0$ or $b=0$, then the equality holds trivially. Now we assume $ab\neq 0$. 
 First we note that the following equalities hold:
     \begin{equation}\label{multiplypimulticommuteeq}
	\beginpgfgraphicnamed{TikZit//addinverse}
	\InputIfFileExists{TikZit//addinverse.tikz}{}{\input{./figures/TikZit//addinverse.tikz}}%
	\endpgfgraphicnamed
, \quad\quad \quad\quad %
	\beginpgfgraphicnamed{TikZit//seperate}
	\InputIfFileExists{TikZit//seperate.tikz}{}{\input{./figures/TikZit//seperate.tikz}}%
	\endpgfgraphicnamed
 
            \end{equation}
           Since 
            $$   %
	\beginpgfgraphicnamed{TikZit//multiplypimulticommuteprf1}
	\InputIfFileExists{TikZit//multiplypimulticommuteprf1.tikz}{}{\input{./figures/TikZit//multiplypimulticommuteprf1.tikz}}%
	\endpgfgraphicnamed
 $$
             $$   %
	\beginpgfgraphicnamed{TikZit//multiplypimulticommuteprf2}
	\InputIfFileExists{TikZit//multiplypimulticommuteprf2.tikz}{}{\input{./figures/TikZit//multiplypimulticommuteprf2.tikz}}%
	\endpgfgraphicnamed
 $$
          Then   it is equivalent to prove
          $$    %
	\beginpgfgraphicnamed{TikZit//multiplypimulticommuteequiv1}
	\InputIfFileExists{TikZit//multiplypimulticommuteequiv1.tikz}{}{\input{./figures/TikZit//multiplypimulticommuteequiv1.tikz}}%
	\endpgfgraphicnamed
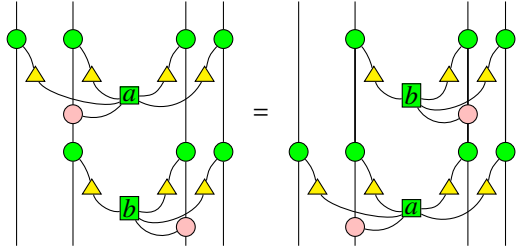
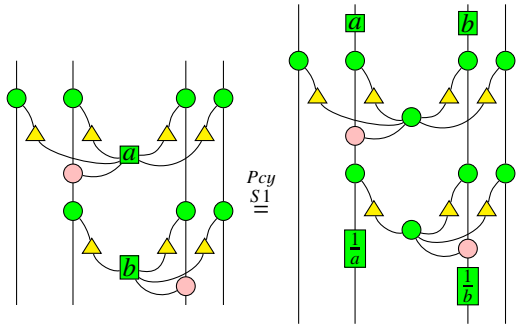
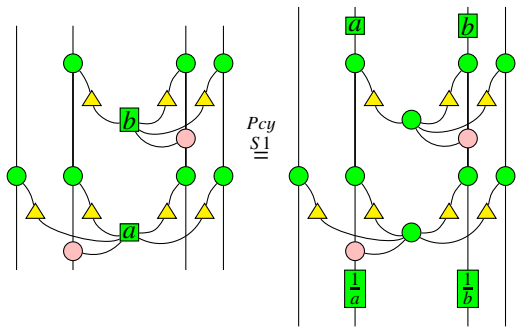
 $$
           
          Furthermore,
              $$    %
	\beginpgfgraphicnamed{TikZit//multiplypimulticommuteequiv2}
	\InputIfFileExists{TikZit//multiplypimulticommuteequiv2.tikz}{}{\input{./figures/TikZit//multiplypimulticommuteequiv2.tikz}}%
	\endpgfgraphicnamed
 $$
              Therefore it is equivalent to prove Lemma \ref{multiplypimulticommutesim} which has been done.
                \end{proof}

  \begin{corollary}\label{multiplypimulticommutecro} 
    $$    %
	\beginpgfgraphicnamed{TikZit//multiplypimulticommutecrory}
	\InputIfFileExists{TikZit//multiplypimulticommutecrory.tikz}{}{\input{./figures/TikZit//multiplypimulticommutecrory.tikz}}%
	\endpgfgraphicnamed
 $$
    \end{corollary}

 \begin{proposition}\label{addpipair2sidecommutprop}
   Suppose  the node $a$ is connected to $ j_1, \cdots,  j_s $ via pink nodes, the node $b$ is connected to $ i_1, \cdots,  i_t $ via pink nodes,  pairs of red $\pi$ nodes separated by green nodes connected to $a$ are located on  $ h_1, \cdots,  h_u $,  pairs of red $\pi$ nodes separated by green nodes connected to $b$ are located on  $ k_1, \cdots,  k_l $. Furthermore,  $\emptyset\neq \{i_1, \cdots,  i_t\} \subseteq  \{1, \cdots,  n\}, \emptyset\neq \{j_1, \cdots,  j_s\} \subseteq  \{1, \cdots,  n\},   \{h_1, \cdots,  h_u\} \subseteq  \{n+1, \cdots,  n+m\},  \{k_1, \cdots,  k_l\} \subseteq  \{n+1, \cdots,  n+m\}$. Then
 $$  %
	\beginpgfgraphicnamed{TikZit//addpipair2sidecommut}
	\InputIfFileExists{TikZit//addpipair2sidecommut.tikz}{}{\input{./figures/TikZit//addpipair2sidecommut.tikz}}%
	\endpgfgraphicnamed
 $$
  \end{proposition}   
    \begin{proof}
    If $\{i_1, \cdots,  i_t\} =\{j_1, \cdots,  j_s\} $, then it is just the case of Proposition \ref{addpidoublecom}.  If  $\{h_1, \cdots,  h_u\}=\{k_1, \cdots,  k_l\}$, then it is just the case of Proposition \ref{addcommutatgencont}. Therefore, we assume w.l.g that $k_1\notin \{h_1, \cdots,  h_u\},   i_t\notin \{j_1, \cdots,  j_s\} $. Then we have
$$  %
	\beginpgfgraphicnamed{TikZit//addpipair2sidecommutprf}
	\InputIfFileExists{TikZit//addpipair2sidecommutprf.tikz}{}{\input{./figures/TikZit//addpipair2sidecommutprf.tikz}}%
	\endpgfgraphicnamed
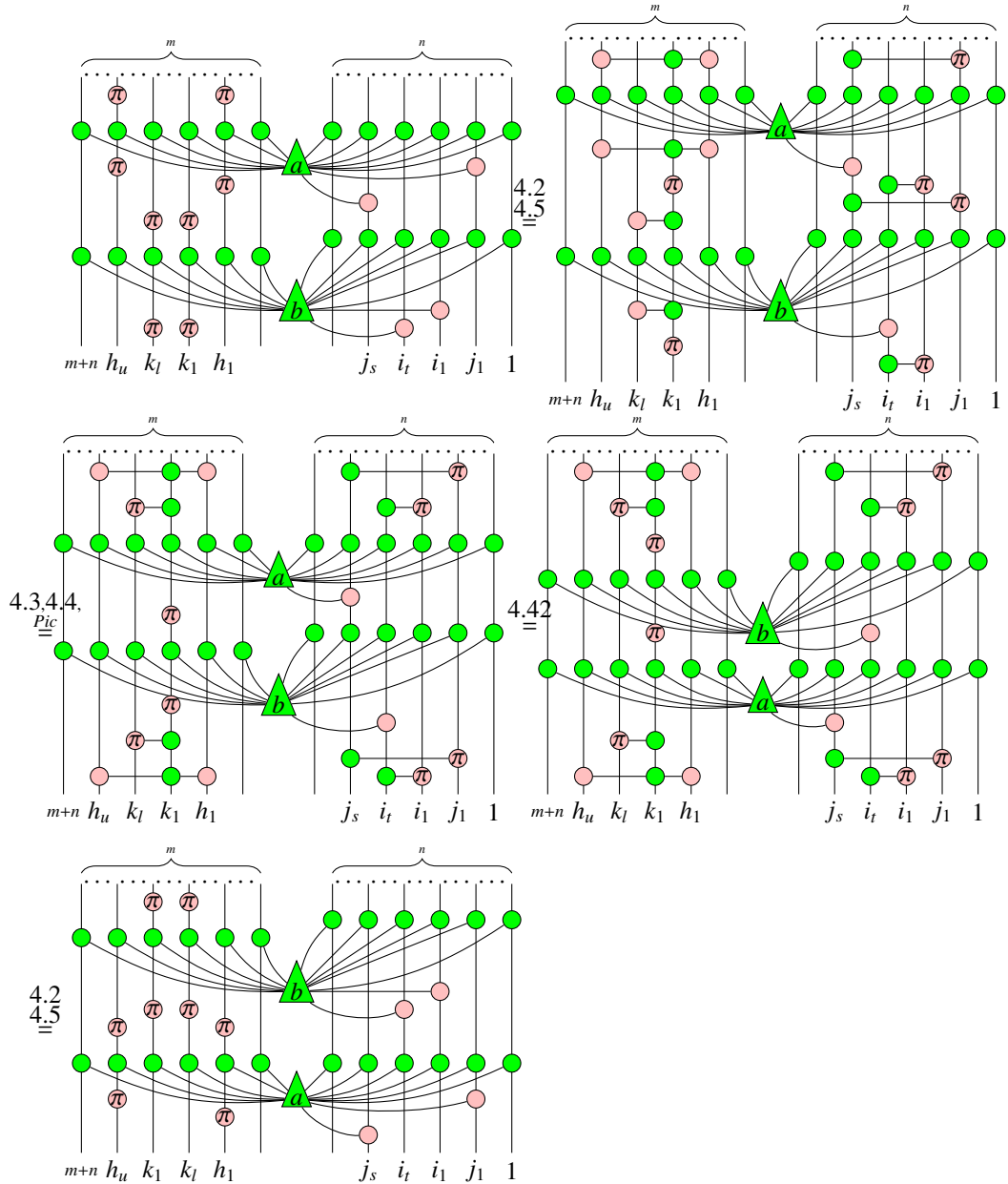
 $$
 \end{proof}

   \begin{corollary}\label{addpipair2sidecommutcro} 
    Suppose  the node $a$ is connected to $ h_1, \cdots,  h_u $ via pink nodes, the node $b$ is connected to $ k_1, \cdots,  k_l $ via pink nodes,  pairs of red $\pi$ nodes separated by green nodes connected to $a$ are located on  $ j_1, \cdots,  j_s $,  pairs of red $\pi$ nodes separated by green nodes connected to $b$ are located on  $ i_1, \cdots,  i_t $. Furthermore,  $\{i_1, \cdots,  i_t\} \subseteq  \{1, \cdots,  n\},  \{j_1, \cdots,  j_s\} \subseteq  \{1, \cdots,  n\},  \emptyset\neq  \{h_1, \cdots,  h_u\} \subseteq  \{n+1, \cdots,  n+m\},  \emptyset\neq\{k_1, \cdots,  k_l\} \subseteq  \{n+1, \cdots,  n+m\}$. Then

    $$    %
	\beginpgfgraphicnamed{TikZit//addpipair2sidecommutcroy}
	\InputIfFileExists{TikZit//addpipair2sidecommutcroy.tikz}{}{\input{./figures/TikZit//addpipair2sidecommutcroy.tikz}}%
	\endpgfgraphicnamed
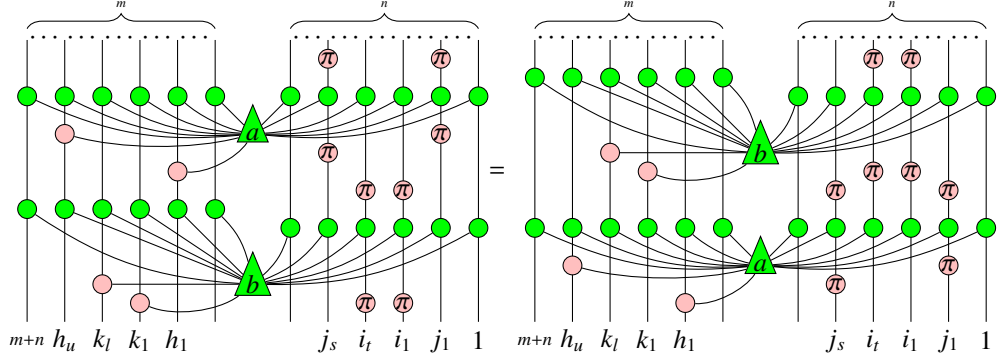
 $$
    \end{corollary}
 This can be obtained from Proposition \ref{addpipair2sidecommutprop} by swapping the left $m$ wires to the right.

    
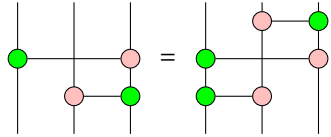
\begin{lemma}\label{cnotscomutelm}
  $$  %
	\beginpgfgraphicnamed{TikZit//cnotscomute}
	\InputIfFileExists{TikZit//cnotscomute.tikz}{}{\input{./figures/TikZit//cnotscomute.tikz}}%
	\endpgfgraphicnamed
 $$
     \end{lemma}
   This can be easily obtained mainly by using the (B2) rule.
   
 \begin{proposition}\label{addpipair2sidecommutprop28}
   Suppose  the node $a$ is connected to $ h_1, \cdots,  h_u $ via pink nodes, the node $b$ is connected to $ i_1, \cdots,  i_t $ via pink nodes,  pairs of red $\pi$ nodes separated by green nodes connected to $a$ are located on  $ j_1, \cdots,  j_s $,  pairs of red $\pi$ nodes separated by green nodes connected to $b$ are located on  $ k_1, \cdots,  k_l $. Furthermore,  $\emptyset\neq \{i_1, \cdots,  i_t\} \subseteq  \{1, \cdots,  n\}, \{j_1, \cdots,  j_s\} \subseteq  \{1, \cdots,  n\},  \emptyset\neq \{h_1, \cdots,  h_u\} \subseteq  \{n+1, \cdots,  n+m\},  \emptyset\neq \{k_1, \cdots,  k_l\} \subseteq  \{n+1, \cdots,  n+m\}, \{h_1, \cdots,  h_u\}\neq \{k_1, \cdots,  k_l\} $. Then
   $$  %
	\beginpgfgraphicnamed{TikZit//addpipair2sidecommutprop28}
	\InputIfFileExists{TikZit//addpipair2sidecommutprop28.tikz}{}{\input{./figures/TikZit//addpipair2sidecommutprop28.tikz}}%
	\endpgfgraphicnamed
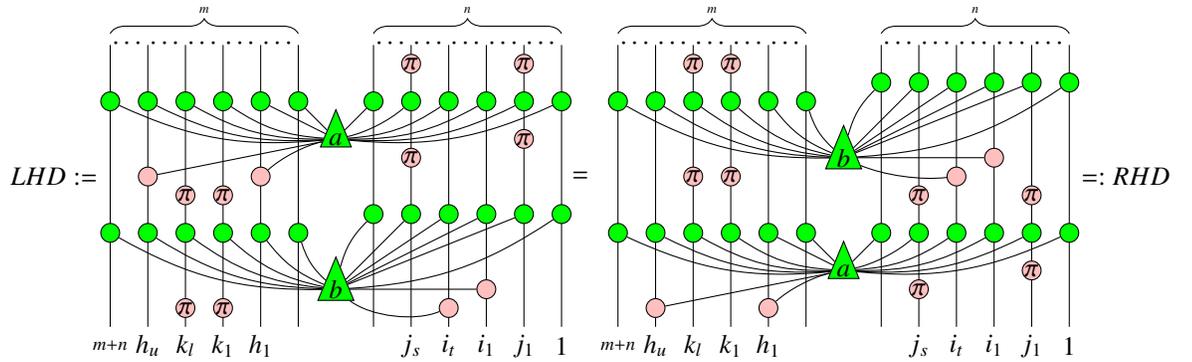
 $$
    \end{proposition}   
 \begin{proof}
 We assume w.l.g that $k_1\notin \{h_1, \cdots,  h_u\} $. If $\{i_1, \cdots,  i_t\} \cap \{j_1, \cdots,  j_s\}=\emptyset$, then
   $$  %
	\beginpgfgraphicnamed{TikZit//addpipair2sidecommutprop28prf}
	\InputIfFileExists{TikZit//addpipair2sidecommutprop28prf.tikz}{}{\input{./figures/TikZit//addpipair2sidecommutprop28prf.tikz}}%
	\endpgfgraphicnamed
 $$
    If $\{i_1, \cdots,  i_t\} \cap \{j_1, \cdots,  j_s\}\neq\emptyset$,  we assume w.l.g that $j_1=i_1$,  then
   $$  %
	\beginpgfgraphicnamed{TikZit//addpipair2sidecommutprop28prf2}
	\InputIfFileExists{TikZit//addpipair2sidecommutprop28prf2.tikz}{}{\input{./figures/TikZit//addpipair2sidecommutprop28prf2.tikz}}%
	\endpgfgraphicnamed
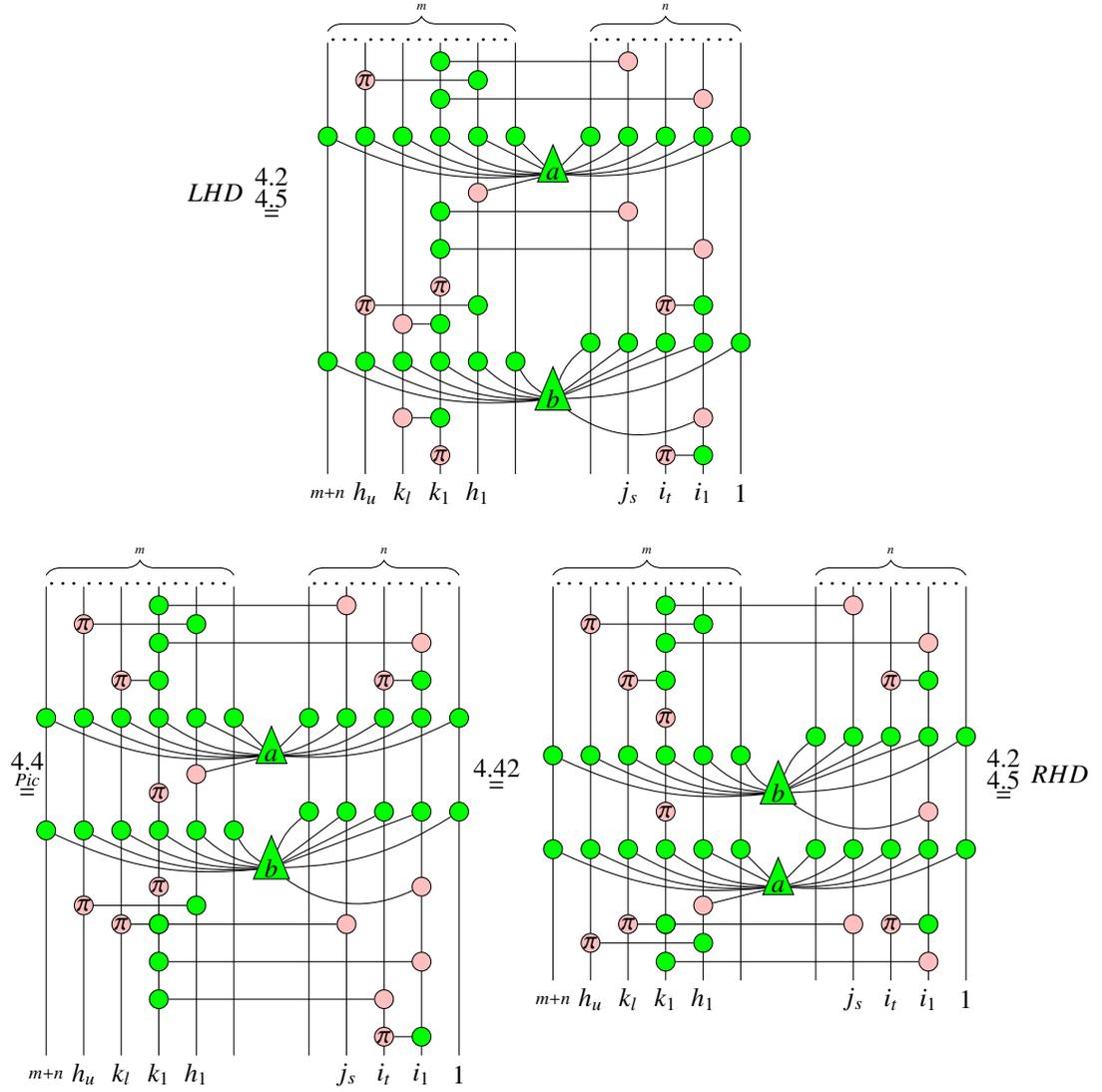
 $$

   \end{proof}

 \begin{proposition}\label{addpipair2sidecommutprop29}  
   Suppose  the node $a$ is connected to $ h_1, \cdots,  h_u $ on the left $m$ wires via pink nodes, and  is connected to $ j_1, \cdots,  j_s $ on the right $n$ wires via pink nodes; the node $b$ is connected to $ k_1, \cdots,  k_l $ via pink nodes,  pairs of red $\pi$ nodes separated by green nodes connected to $b$ are located on  $ i_1, \cdots,  i_t $. Furthermore,  $ \{i_1, \cdots,  i_t\} \subseteq  \{1, \cdots,  n\}, \emptyset\neq\{j_1, \cdots,  j_s\} \subseteq  \{1, \cdots,  n\},  \emptyset\neq \{h_1, \cdots,  h_u\} \subseteq  \{n+1, \cdots,  n+m\},  \emptyset\neq \{k_1, \cdots,  k_l\} \subseteq  \{n+1, \cdots,  n+m\}, \{h_1, \cdots,  h_u\}\neq \{k_1, \cdots,  k_l\} $. Then
   $$  %
	\beginpgfgraphicnamed{TikZit//addpipair2sidecommuteprop29}
	\InputIfFileExists{TikZit//addpipair2sidecommuteprop29.tikz}{}{\input{./figures/TikZit//addpipair2sidecommuteprop29.tikz}}%
	\endpgfgraphicnamed
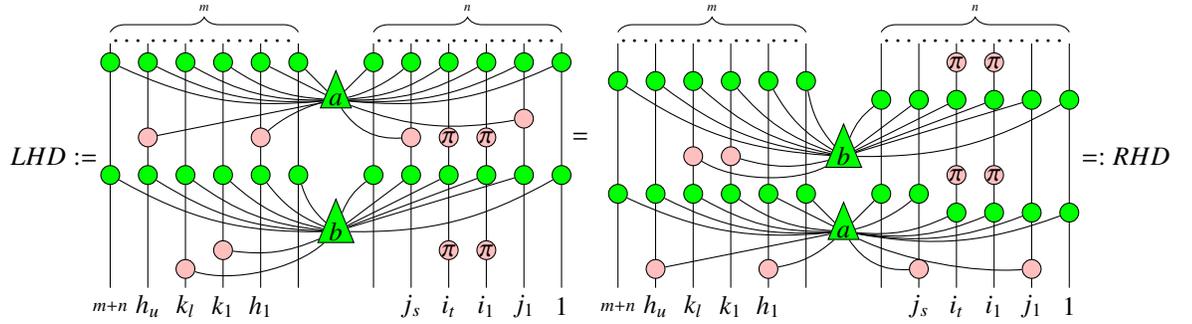
 $$
   \end{proposition}   
  
  \begin{proof} If  $\{h_1, \cdots,  h_u\}\nsubseteq \{k_1, \cdots,  k_l\} $,  we assume w.l.g that $h_1\notin \{k_1, \cdots,  k_l\} $. Then
   $$  %
	\beginpgfgraphicnamed{TikZit//addpipair2sidecommuteprop29prf0}
	\InputIfFileExists{TikZit//addpipair2sidecommuteprop29prf0.tikz}{}{\input{./figures/TikZit//addpipair2sidecommuteprop29prf0.tikz}}%
	\endpgfgraphicnamed
 $$
   
  If  $\{h_1, \cdots,  h_u\}\subseteq \{k_1, \cdots,  k_l\} $, we assume w.l.g that $k_1\notin \{h_1, \cdots,  h_u\} $. Here are two cases. (1)  
  $ \{i_1, \cdots,  i_t\} \nsubseteq \{j_1, \cdots,  j_s\}$, we assume w.l.g that $i_1\notin \{j_1, \cdots,  j_s\}$. Then
  $$  %
	\beginpgfgraphicnamed{TikZit//addpipair2sidecommuteprop29prf10}
	\InputIfFileExists{TikZit//addpipair2sidecommuteprop29prf10.tikz}{}{\input{./figures/TikZit//addpipair2sidecommuteprop29prf10.tikz}}%
	\endpgfgraphicnamed
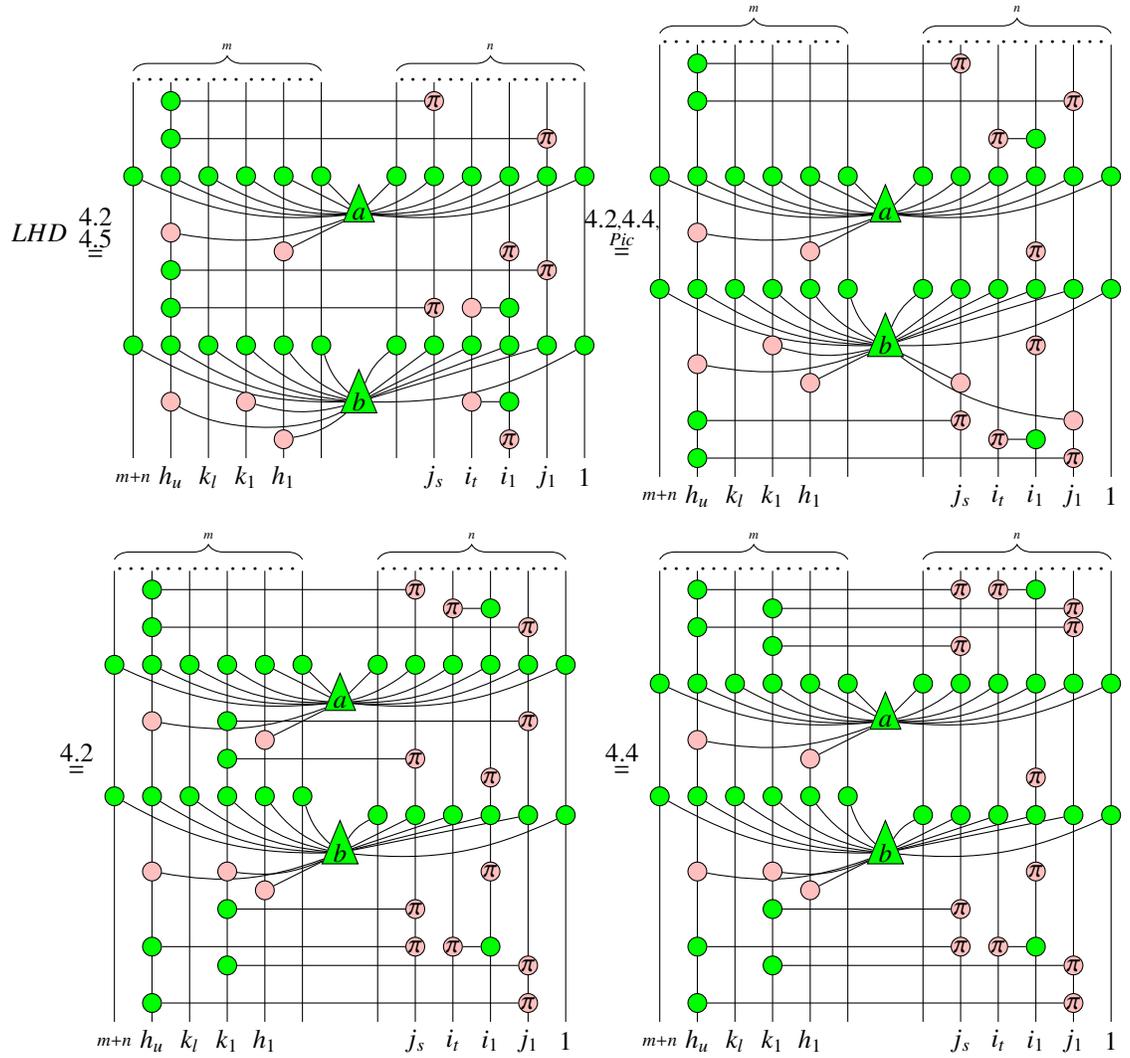
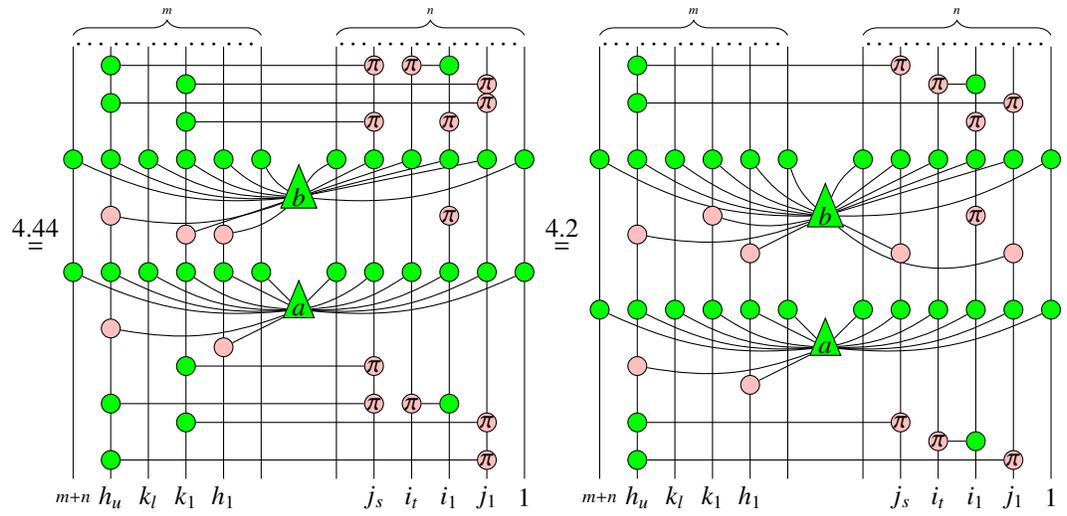
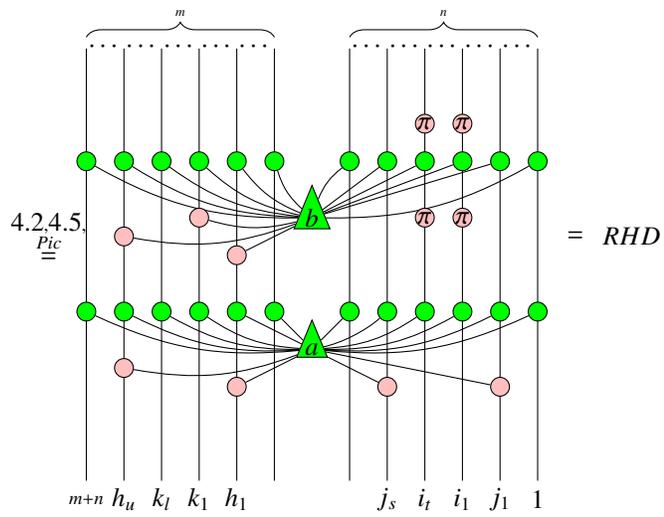
 $$ 
   $$  %
	\beginpgfgraphicnamed{TikZit//addpipair2sidecommuteprop29prf1}
	\InputIfFileExists{TikZit//addpipair2sidecommuteprop29prf1.tikz}{}{\input{./figures/TikZit//addpipair2sidecommuteprop29prf1.tikz}}%
	\endpgfgraphicnamed
 $$ 
  (2)  
  $ \{i_1, \cdots,  i_t\} \subseteq \{j_1, \cdots,  j_s\}$. Then
   $$  %
	\beginpgfgraphicnamed{TikZit//addpipair2sidecommuteprop29prf2}
	\InputIfFileExists{TikZit//addpipair2sidecommuteprop29prf2.tikz}{}{\input{./figures/TikZit//addpipair2sidecommuteprop29prf2.tikz}}%
	\endpgfgraphicnamed
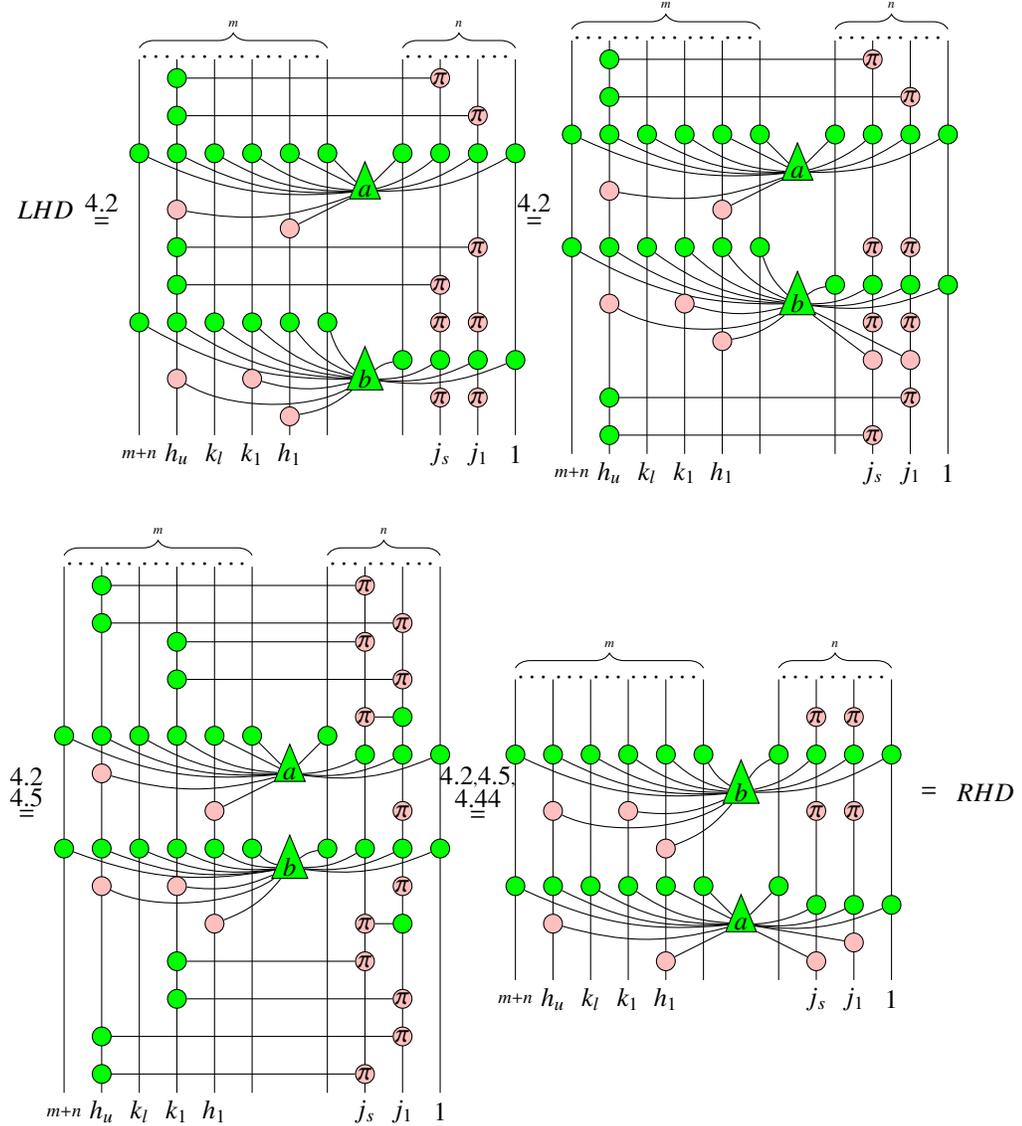
 $$  
    \end{proof}

   \begin{proposition}\label{addpipair2sidecommutprop29b}  
    Suppose  the node $a$ is connected to $ h_1, \cdots,  h_u $ on the left $m$ wires via pink nodes, and  is connected to $ j_1, \cdots,  j_s $ on the right $n$ wires via pink nodes; the node $b$ is connected to  $ i_1, \cdots,  i_t $ via pink nodes,  pairs of red $\pi$ nodes separated by green nodes connected to $b$ are located on $ k_1, \cdots,  k_l $ . Furthermore,  $\emptyset\neq \{i_1, \cdots,  i_t\} \subseteq  \{1, \cdots,  n\}, \emptyset\neq\{j_1, \cdots,  j_s\} \subseteq  \{1, \cdots,  n\},  \emptyset\neq \{h_1, \cdots,  h_u\} \subseteq  \{n+1, \cdots,  n+m\}, \{k_1, \cdots,  k_l\} \subseteq  \{n+1, \cdots,  n+m\}, \{h_1, \cdots,  h_u\}\neq \{k_1, \cdots,  k_l\} $. Then
   
 $$  %
	\beginpgfgraphicnamed{TikZit//addpipair2sidecommuteprop29b}
	\InputIfFileExists{TikZit//addpipair2sidecommuteprop29b.tikz}{}{\input{./figures/TikZit//addpipair2sidecommuteprop29b.tikz}}%
	\endpgfgraphicnamed
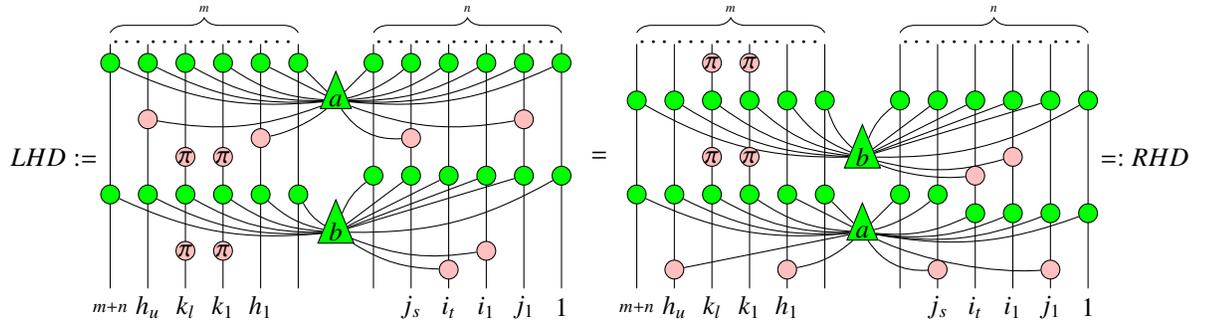
 $$
   \end{proposition}  
    \begin{proof} 
    If $\{k_1, \cdots,  k_l\} = \emptyset$, then it is just the case of Proposition \ref{addcommutatgencont}. Below we assume  that  $\{k_1, \cdots,  k_l\} \neq \emptyset$. Since $\emptyset\neq \{h_1, \cdots,  h_u\}$, we assume that  $\{h_1\} \neq \emptyset$. 
Then we have
$$  %
	\beginpgfgraphicnamed{TikZit//addpipair2sidecommuteprop29bprf00}
	\InputIfFileExists{TikZit//addpipair2sidecommuteprop29bprf00.tikz}{}{\input{./figures/TikZit//addpipair2sidecommuteprop29bprf00.tikz}}%
	\endpgfgraphicnamed
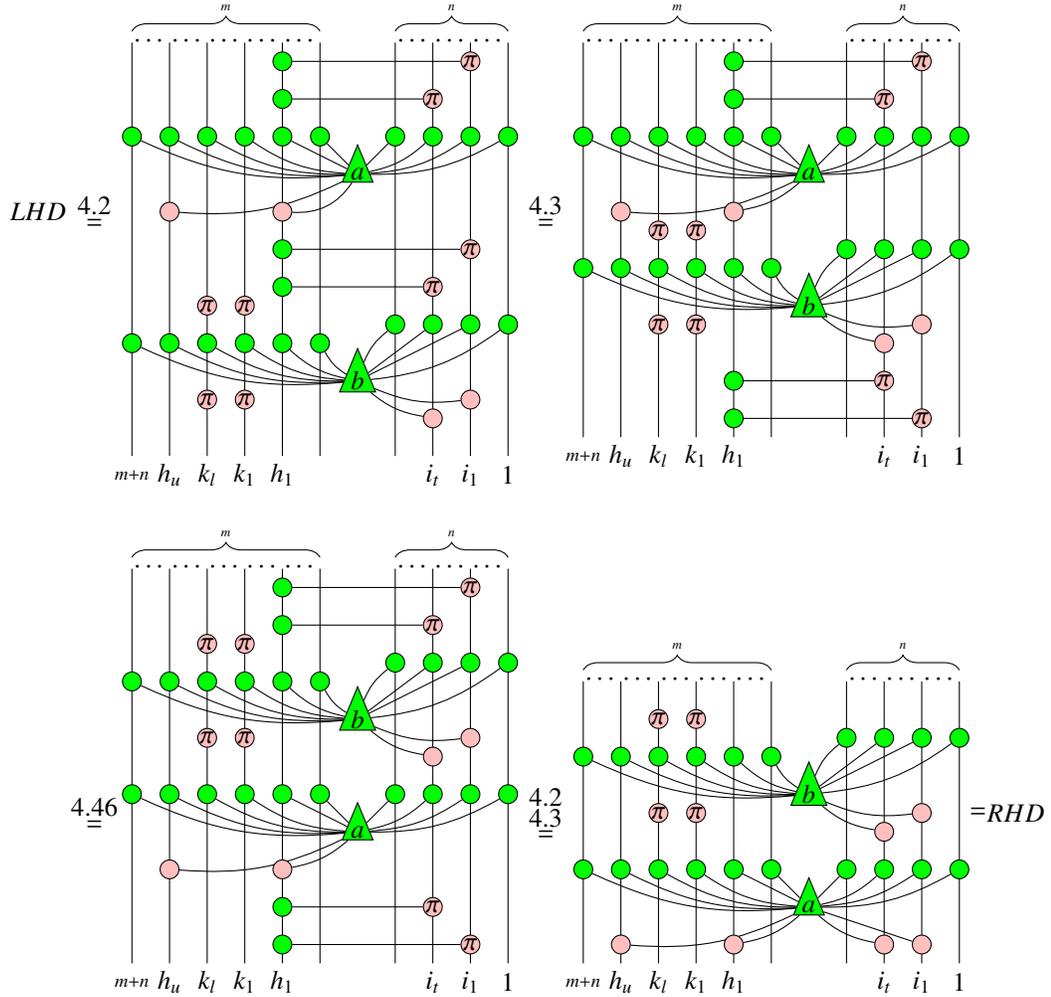
 $$
   \end{proof}

 \begin{proposition}\label{addpipairmulcommutprop30a}  
    Suppose  the node $a$ is connected to $ j_1, \cdots,  j_s $ on the  right $n$  wires via pink nodes,
 pairs of red $\pi$ nodes separated by green nodes connected to $a$ are located on $ k_1, \cdots,  k_l $,   pairs of red $\pi$ nodes separated by green nodes connected to $b$ are located on $ i_1, \cdots,  i_t $. Furthermore,  $ \{i_1, \cdots,  i_t\} \subseteq  \{1, \cdots,  n\}, \emptyset\neq\{j_1, \cdots,  j_s\} \subseteq  \{1, \cdots,  n\},  \emptyset\neq\{k_1, \cdots,  k_l\} \subseteq  \{n+1, \cdots,  n+m\}$. Then
    $$  %
	\beginpgfgraphicnamed{TikZit//addpipairmultcommutprop30a}
	\InputIfFileExists{TikZit//addpipairmultcommutprop30a.tikz}{}{\input{./figures/TikZit//addpipairmultcommutprop30a.tikz}}%
	\endpgfgraphicnamed
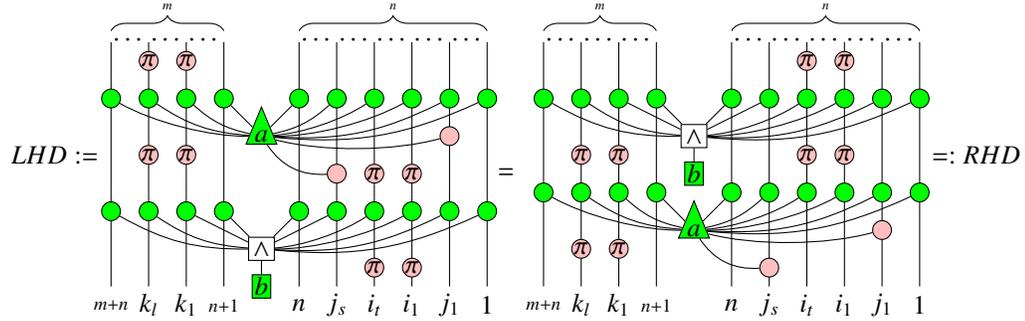
 $$
   \end{proposition}  
 \begin{proof} 
 Since $\{k_1, \cdots,  k_l\}\cap \{j_1, \cdots,  j_s\}=  \emptyset$, it is just one of the cases of Proposition \ref{addpipairmultiplycommutgp}, so we are done.
   \end{proof}

 \begin{proposition}\label{addpipairmulcommutprop30b}  
    Suppose  the node $a$ is connected to $ j_1, \cdots,  j_s $ on the  right $n$  wires via pink nodes, and connected to $ k_1, \cdots,  k_l $ on the  left $m$  wires via pink nodes;
 pairs of red $\pi$ nodes separated by green nodes connected to $b$ are located on $ i_1, \cdots,  i_t $. Furthermore,  $ \emptyset\neq\{i_1, \cdots,  i_t\} \subseteq  \{1, \cdots,  n\}, \emptyset\neq\{j_1, \cdots,  j_s\} \subseteq  \{1, \cdots,  n\},  \emptyset\neq\{k_1, \cdots,  k_l\} \subseteq  \{n+1, \cdots,  n+m\}$. Then
    $$  %
	\beginpgfgraphicnamed{TikZit//addpipairmultcommutprop30b}
	\InputIfFileExists{TikZit//addpipairmultcommutprop30b.tikz}{}{\input{./figures/TikZit//addpipairmultcommutprop30b.tikz}}%
	\endpgfgraphicnamed
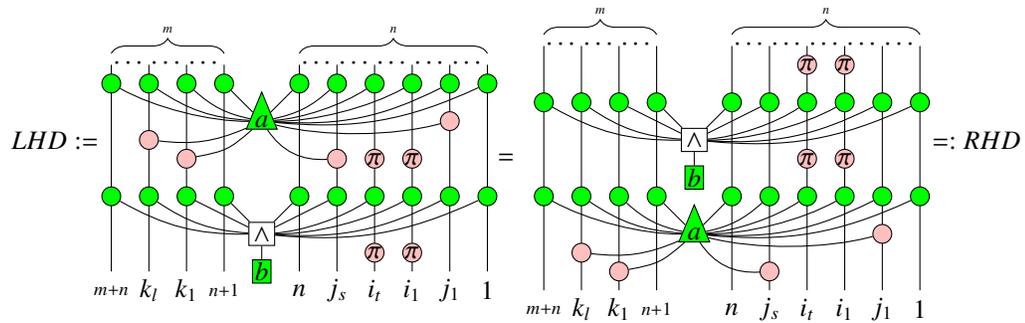
 $$
   \end{proposition}  
 \begin{proof} 
 If $|\{i_1, \cdots,  i_t\}\cap \{j_1, \cdots,  j_s\}| \leq 1$, then it is just one of the cases of Proposition \ref{addpipairmultiplycommutgp}, so we are done. Below we assume that $|\{i_1, \cdots,  i_t\}\cap \{j_1, \cdots,  j_s\}| \geq 2$. 
 If 
 $\{i_1, \cdots,  i_t\}\subseteq \{j_1, \cdots,  j_s\}$, we assume  w.l.g that $i_1=j_1, i_t=j_u$. Then 
  we have
    \begin{equation}\label{addpipairmulcommutprop30beq} 
	\beginpgfgraphicnamed{TikZit//addpipairmultcommutprop30bprf1}
	\InputIfFileExists{TikZit//addpipairmultcommutprop30bprf1.tikz}{}{\input{./figures/TikZit//addpipairmultcommutprop30bprf1.tikz}}%
	\endpgfgraphicnamed
 
    \end{equation}
  If $\{i_1, \cdots,  i_t\}\nsubseteq \{j_1, \cdots,  j_s\}$, we assume  w.l.g that $i_1=j_1, i_t=j_s, i_2\notin \{j_1, \cdots,  j_s\}$. Then by the proof of the last case we have
  $$  %
	\beginpgfgraphicnamed{TikZit//addpipairmultcommutprop30bprf2}
	\InputIfFileExists{TikZit//addpipairmultcommutprop30bprf2.tikz}{}{\input{./figures/TikZit//addpipairmultcommutprop30bprf2.tikz}}%
	\endpgfgraphicnamed
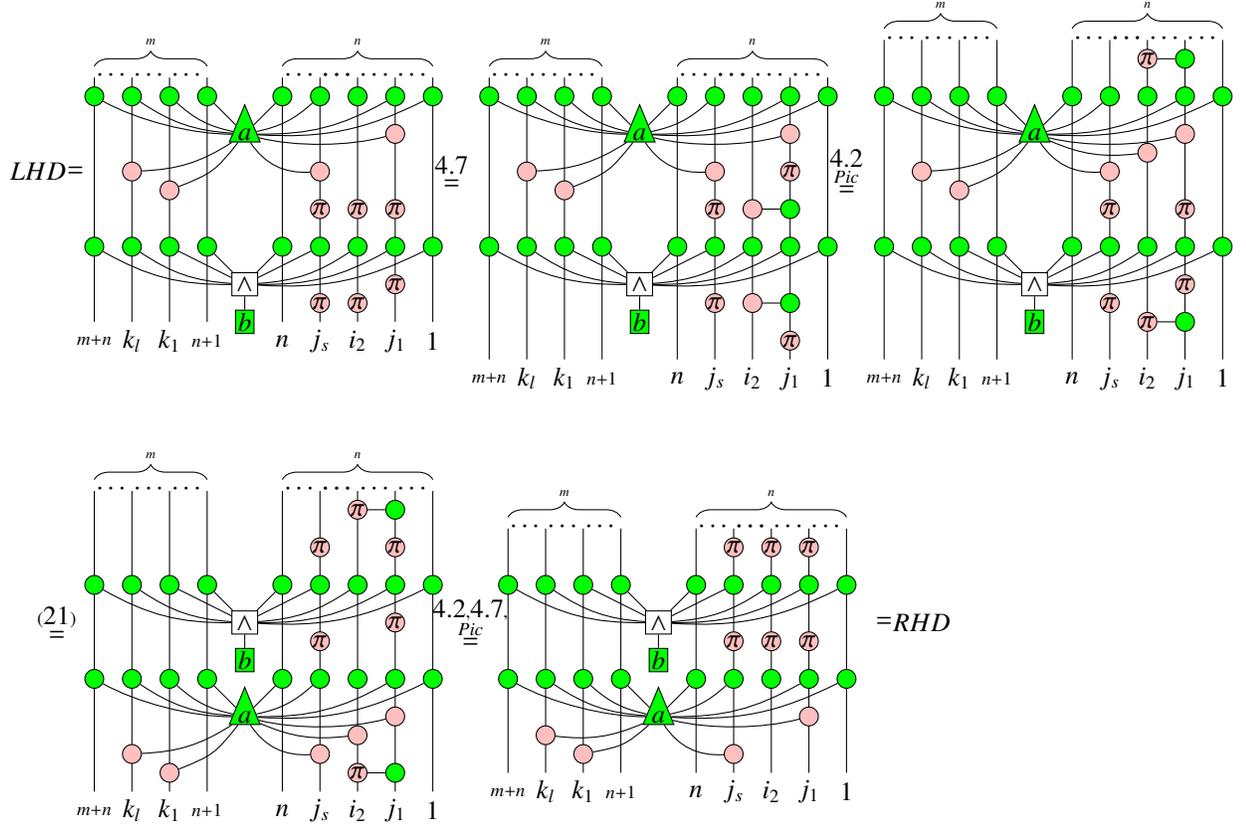
 $$
  \end{proof}

    \begin{corollary}\label{addpipairmulcommutprop30bcro} 
      Suppose  the node $a$ is connected to $ j_1, \cdots,  j_s $ on the  right $n$  wires via pink nodes, and connected to $ h_1, \cdots,  h_u $ on the  left $m$  wires via pink nodes;
   pairs of red $\pi$ nodes separated by green nodes connected to $b$ are located on $ k_1, \cdots,  k_l $. Furthermore,  $ \emptyset\neq\{j_1, \cdots,  j_s\} \subseteq  \{1, \cdots,  n\},  \emptyset\neq\{k_1, \cdots,  k_l\} \subseteq  \{n+1, \cdots,  n+m\}, \emptyset\neq\{h_1, \cdots,  h_u\} \subseteq  \{n+1, \cdots,  n+m\}$. Then
       $$  %
	\beginpgfgraphicnamed{TikZit//addpipairmultcommutprop30bcro}
	\InputIfFileExists{TikZit//addpipairmultcommutprop30bcro.tikz}{}{\input{./figures/TikZit//addpipairmultcommutprop30bcro.tikz}}%
	\endpgfgraphicnamed
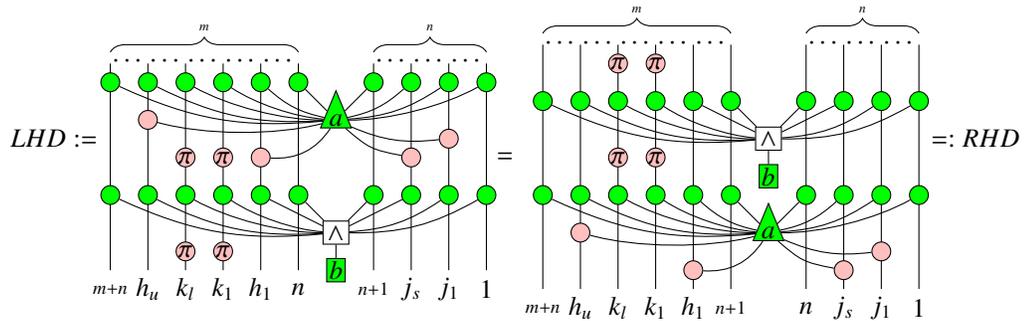
 $$
   \end{corollary}  
   This is just the symmetric case of Proposition \ref{addpipairmulcommutprop30b}, thus can be obtained by swapping wires.

  \begin{proposition}\label{addpipairmulcommutprop30c}  
    Suppose  the node $a$ is connected to $ j_1, \cdots,  j_s $ on the  right $n$  wires via pink nodes, 
 pairs of red $\pi$ nodes separated by green nodes connected to $b$ are located on $ i_1, \cdots,  i_t $. Furthermore,  $ \emptyset\neq\{i_1, \cdots,  i_t\} \subseteq  \{1, \cdots,  n\}, \emptyset\neq\{j_1, \cdots,  j_s\} \subseteq  \{1, \cdots,  n\},  \{i_1, \cdots,  i_t\}\neq\{j_1, \cdots,  j_s\}  $. Then
    $$  %
	\beginpgfgraphicnamed{TikZit//addpipairmultcommutprop30c}
	\InputIfFileExists{TikZit//addpipairmultcommutprop30c.tikz}{}{\input{./figures/TikZit//addpipairmultcommutprop30c.tikz}}%
	\endpgfgraphicnamed
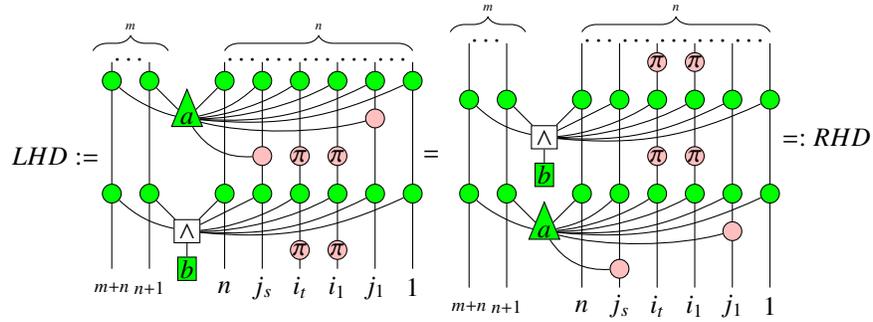
 $$
   \end{proposition}    
  \begin{proof} 
  If $\{i_1, \cdots,  i_t\}\nsubseteq \{j_1, \cdots,  j_s\}$, we assume  w.l.g that $i_1\notin \{j_1, \cdots,  j_s\}$. Then by Proposition \ref{addpipairmultiplycommutgp} we have
  $$  %
	\beginpgfgraphicnamed{TikZit//addpipairmultcommutprop30cprf1}
	\InputIfFileExists{TikZit//addpipairmultcommutprop30cprf1.tikz}{}{\input{./figures/TikZit//addpipairmultcommutprop30cprf1.tikz}}%
	\endpgfgraphicnamed
 $$
   If $\{i_1, \cdots,  i_t\}\subseteq \{j_1, \cdots,  j_s\}$, since $\{i_1, \cdots,  i_t\}\neq\{j_1, \cdots,  j_s\}  $, we assume  w.l.g that $i_1=j_1, i_t=j_s, j_2\notin \{i_1, \cdots,  i_t\}$. Then by Proposition \ref{addpipairmultiplycommutgp} we have
  $$  %
	\beginpgfgraphicnamed{TikZit//addpipairmultcommutprop30cprf2}
	\InputIfFileExists{TikZit//addpipairmultcommutprop30cprf2.tikz}{}{\input{./figures/TikZit//addpipairmultcommutprop30cprf2.tikz}}%
	\endpgfgraphicnamed
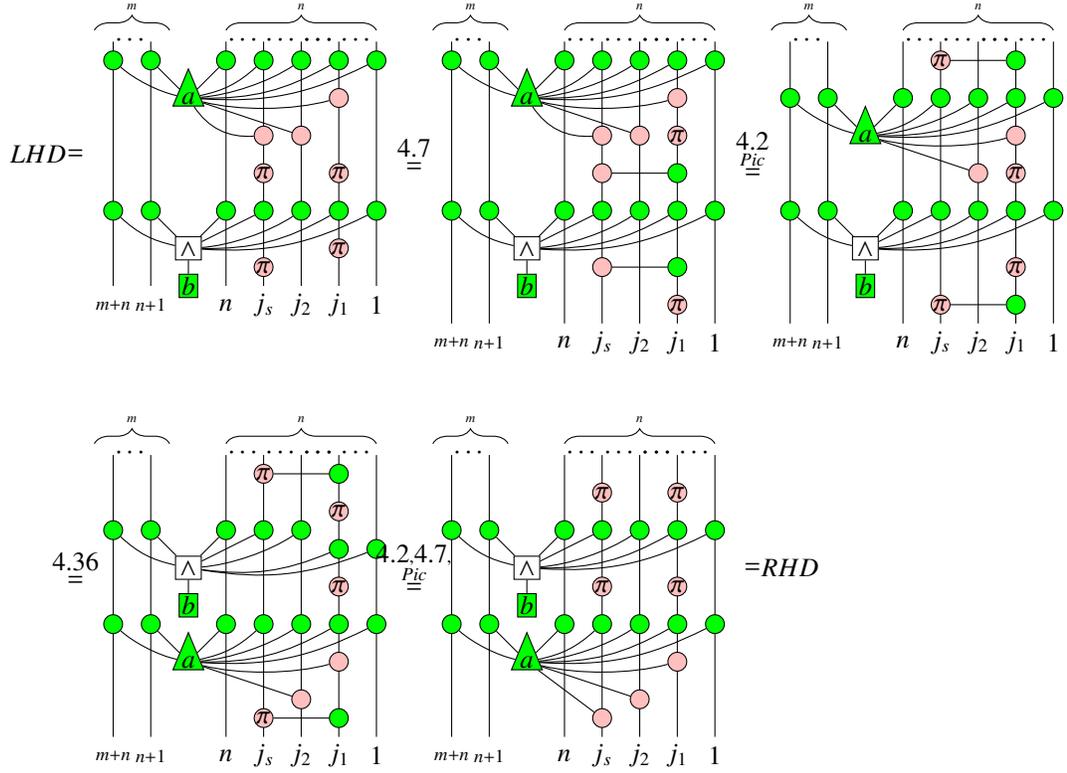
 $$ 
     \end{proof}

    \begin{corollary}\label{addpipairmulcommutprop30ccro}  
    Suppose  the node $a$ is connected to $ j_1, \cdots,  j_s $ on the  left $m$  wires via pink nodes, 
 pairs of red $\pi$ nodes separated by green nodes connected to $b$ are located on $ i_1, \cdots,  i_t $. Furthermore,  $ \emptyset\neq\{i_1, \cdots,  i_t\} \subseteq  \{1+n, \cdots,  m+n\}, \emptyset\neq\{j_1, \cdots,  j_s\} \subseteq  \{1+n, \cdots,  m+n\},  \{i_1, \cdots,  i_t\}\neq\{j_1, \cdots,  j_s\}  $. Then
    $$  %
	\beginpgfgraphicnamed{TikZit//addpipairmultcommutprop30ccro}
	\InputIfFileExists{TikZit//addpipairmultcommutprop30ccro.tikz}{}{\input{./figures/TikZit//addpipairmultcommutprop30ccro.tikz}}%
	\endpgfgraphicnamed
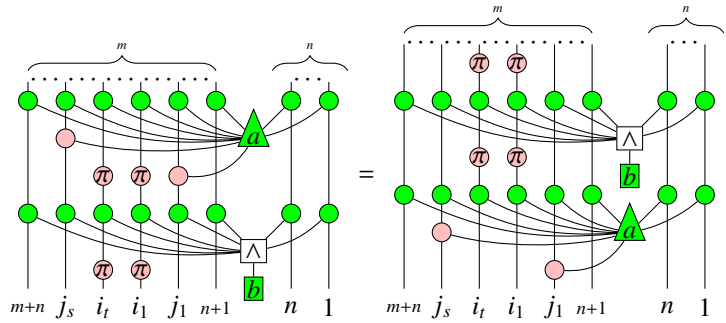
 $$
   \end{corollary}   
     This is just the symmetric case of Proposition \ref{addpipairmulcommutprop30c}, thus can be obtained by swapping wires.

 \section{Completeness of ZX-calculus via elementary transformations}
 In this section, we prove the following main theorem based on a normal form of an arbitrary vector which is obtained via elementary transformations.

 \begin{theorem} \label{main}  
 The ZX-calculus is complete with respect to the rules in Figure \ref{figurealgebra}.
   \end{theorem}      
     
  \subsection{Normal form}
  
   
  
  Suppose $\{e_{k}| 0\leq k \leq 2^{m}-1\}$ are the $2^m$-dimensional standard unit  column vectors (with entries all 0s except for one 1):
 \[
e_k=\begin{blockarray}{cl}
\begin{block}{(c)l}
     0 &r_0\\
      \vdots    &  \\
      1&r_k\\
     \vdots &  \\
       0&r_{2^m-1}\\
\end{block}
\end{blockarray}
 \]
 where $r_i$ denote the i-th row, $0\leq i \leq 2^{m}-1,  m\geq 1$.
 
 Let 
 $$ \ket{0}= \begin{pmatrix}
        1  \\
        0  \\
 \end{pmatrix}, \ket{1}= \begin{pmatrix}
        0  \\
        1  \\
 \end{pmatrix}.$$
Then 
\begin{equation}\label{qbitstovector}
$$\ket{a_{m-1}\cdots a_i \cdots a_0}=e_{\sum_{i=0}^{m-1}a_i2^i},$$
\end{equation}
where $a_i\in \{0, 1\}, 0\leq i \leq m-1$ \cite{Mermin}.

 To have the normal form, we first need to have a graphical representation of elementary row addition. Suppose that $0\leq j_1< \cdots < j_s \leq m-1, 1\leq s \leq m$. We claim that the following diagram
 \begin{equation}\label{rowaddrepresentation}
	\beginpgfgraphicnamed{TikZit//rowaddrepresent}
	\InputIfFileExists{TikZit//rowaddrepresent.tikz}{}{\input{./figures/TikZit//rowaddrepresent.tikz}}%
	\endpgfgraphicnamed

 \end{equation}
where $a_j$ connects to  $ j_1, \cdots,  j_s $ via red dots represents the $2^m\times 2^m$ row-addition elementary matrix: 

 \[
A_j=\begin{blockarray}{cccccl}
\begin{block}{(ccccc)l}
     1 & \cdots & 0 &\cdots & 0 &r_0\\
     \vdots    & \ddots & &&  \vdots&  \\
        0   & \cdots & 1 & \cdots& a_j&r_j\\
       \vdots    & &  & \ddots&  \vdots &  \\
        0   & \cdots &0 & \cdots& 1&r_{2^m-1}\\
\end{block}
\end{blockarray}
 \]
 where $a_j$ lies in the $r_j$ row, $j=2^m-1-(2^{j_1}+\cdots+2^{j_s})$.
 We demonstrate this in the following.  Note that all the columns $\{A_{jk}| 0\leq k \leq 2^{m-1}\}$ of $A_j$ are standard unit vectors except the last  (counted from left to right) column. This means 
 $$A_{jk}=A_je_{k}=\left\{
\begin{array}{ll}
e_k, 0\leq k \leq 2^{m}-2\\
e_{2^m-1}+a_je_j, k = 2^{m}-1\\
\end{array}
\right.$$ 
where 
 \[
A_{j2^{m}-1}=\begin{blockarray}{cl}
\begin{block}{(c)l}
     0 &r_0\\
      \vdots    &  \\
      a_j&r_j\\
     \vdots &  \\
       1&r_{2^m-1}\\
\end{block}
\end{blockarray}
 \]
By the equality (\ref{qbitstovector}), we let  $e_k=\ket{a_{m-1}\cdots a_i \cdots a_0}$. Clearly, if $0\leq k \leq 2^{m}-2$, then some $a_i$ must be $0, 0\leq i \leq m-1$. Suppose $a_{j_1}=0$, then $e_k$ can be represented by the  following diagram:
 $$ %
	\beginpgfgraphicnamed{TikZit//rowaddverify0}
	\InputIfFileExists{TikZit//rowaddverify0.tikz}{}{\input{./figures/TikZit//rowaddverify0.tikz}}%
	\endpgfgraphicnamed
$$
 where $\alpha_i \in \{0, \pi \}, 0\leq i \leq m-1.$ Then for $0\leq k \leq 2^{m}-2$, we have
\[
\left\llbracket%
	\beginpgfgraphicnamed{TikZit//rowaddverify1}
	\InputIfFileExists{TikZit//rowaddverify1.tikz}{}{\input{./figures/TikZit//rowaddverify1.tikz}}%
	\endpgfgraphicnamed
\right\rrbracket =e_k
\]

For $k = 2^{m}-1$, we have

 $$\left\llbracket%
	\beginpgfgraphicnamed{TikZit//standardunit}
	\InputIfFileExists{TikZit//standardunit.tikz}{}{\input{./figures/TikZit//standardunit.tikz}}%
	\endpgfgraphicnamed
\right\rrbracket=e_{2^m-1}$$  
\[
\begin{array}{rl}
\left\llbracket%
	\beginpgfgraphicnamed{TikZit//rowaddverify2}
	\InputIfFileExists{TikZit//rowaddverify2.tikz}{}{\input{./figures/TikZit//rowaddverify2.tikz}}%
	\endpgfgraphicnamed
\right\rrbracket &=\ket{\overbrace{1\cdots 1}^m}+a_j\ket{\underset{m-1}{1}\cdots \underset{j_s}{0}\cdots 1\cdots \underset{j_1}{0}\cdots \underset{0}{1}}\\
&=e_{2^m-1}+a_je_j
\end{array}
\]
where we used 
( \ref{qbitstovector}) for the last equality. Therefore, we have shown that the diagram (\ref{rowaddrepresentation}) represents the matrix $A_j$.

Similarly we can check that the following diagram 
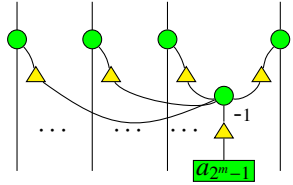
\[
	\beginpgfgraphicnamed{TikZit//rowmultring}
	\InputIfFileExists{TikZit//rowmultring.tikz}{}{\input{./figures/TikZit//rowmultring.tikz}}%
	\endpgfgraphicnamed

\]
represents the row-multiplication matrix:
 \[
M=\begin{blockarray}{cccccl}
\begin{block}{(ccccc)l}
     1 & \cdots & 0 &\cdots & 0 &r_0\\
     \vdots    & \ddots & &&  \vdots&  \\
        0   & \cdots & 1 & \cdots& 0&r_k\\
       \vdots    & &  & \ddots&  \vdots &  \\
        0   & \cdots &0 & \cdots& a_{2^m-1}&r_{2^m-1}\\
\end{block}
\end{blockarray}
 \]

Given an arbitrary  vector $(a_0, a_1, \cdots, a_{2^m-1})^T$ with $a_i \in \mathbb{C}, m\geq 1$, we claim that it can be uniquely represented by the following normal form:
 \begin{equation}\label{normalfring}
	\beginpgfgraphicnamed{TikZit//normalform3}
	\InputIfFileExists{TikZit//normalform3.tikz}{}{\input{./figures/TikZit//normalform3.tikz}}%
	\endpgfgraphicnamed

 \end{equation}	
where for those diagrams which represent row additions, $a_i$ connects to wires with pink nodes depending on $i$, and all possible connections are included in the normal form. Actually, one can check that there are $\binom{m}{1}+\binom{m}{2} +\cdots + \binom{m}{m}=2^m-1$ row additions in the normal form. By Proposition \ref{addcommutatgencont}, all the  row addition diagrams are commutative with each other.

The normal form (\ref{normalfring}) is actually obtained via the following processes:
\[
 \begin{pmatrix}
        0  \\ 0\\
        \vdots \\
        1 \end{pmatrix} \xrightarrow[\text{addition}]{\text{row}}  \begin{pmatrix}
        a_0  \\ 0\\
        \vdots \\
        1 \end{pmatrix} \xrightarrow[\text{addition}]{\text{row}}  \begin{pmatrix}
        a_0  \\  a_1\\
        \vdots \\ a_{2^m-2}\\
        1 \end{pmatrix}  \xrightarrow[\text{multiplication}]{\text{row}}   \begin{pmatrix}
        a_0  \\  a_1\\
        \vdots \\ a_{2^m-2}\\
        a_{2^m-1} \end{pmatrix}
\]
In the case of $m=0$, for any complex number $a$, its normal form is defined as
  $$ %
	\beginpgfgraphicnamed{TikZit//scalarnorm}
	\begin{tikzpicture}
	\begin{pgfonlayer}{nodelayer}
		\node [style=rn] (0) at (0, 0.25) {$\pi$};
		\node [style=gbox, scale=1] (1) at (0, -0.25) {$a$};
	\end{pgfonlayer}
	\begin{pgfonlayer}{edgelayer}
		\draw (0) to (1);
	\end{pgfonlayer}
\end{tikzpicture}}%
	\endpgfgraphicnamed
$$
where 
 $$\left\llbracket%
	\beginpgfgraphicnamed{TikZit//scalarnorm}
	\begin{tikzpicture}
	\begin{pgfonlayer}{nodelayer}
		\node [style=rn] (0) at (0, 0.25) {$\pi$};
		\node [style=gbox, scale=1] (1) at (0, -0.25) {$a$};
	\end{pgfonlayer}
	\begin{pgfonlayer}{edgelayer}
		\draw (0) to (1);
	\end{pgfonlayer}
\end{tikzpicture}}%
	\endpgfgraphicnamed
\right\rrbracket=a$$  
By the map-state duality as given in (\ref{maptostate}), we obtain the universality of ZX-calculus over  $\mathbb{C}$:  any $2^m \times 2^n$ matrix $A$ with $m, n \geq 0$ can be represented by a ZX diagram. 

Now we make the following conventions due to the rule (Aso).
  $$ %
	\beginpgfgraphicnamed{TikZit//wntoalgzx}
	\InputIfFileExists{TikZit//wntoalgzx.tikz}{}{\input{./figures/TikZit//wntoalgzx.tikz}}%
	\endpgfgraphicnamed
 \quad\quad\quad\quad  %
	\beginpgfgraphicnamed{TikZit//mlegsblackspider}
	\InputIfFileExists{TikZit//mlegsblackspider.tikz}{}{\input{./figures/TikZit//mlegsblackspider.tikz}}%
	\endpgfgraphicnamed
$$

At the end of this subsection, we give another form to the  normal form (\ref{normalfring}) as follows. 

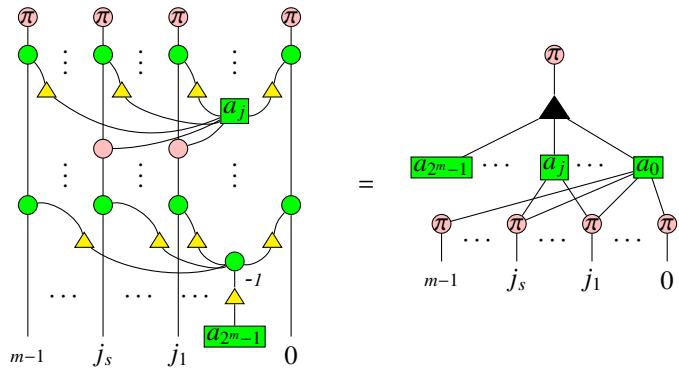
\begin{proposition}\label{normalformwithwpro}
 $$ %
	\beginpgfgraphicnamed{TikZit//normalformwithw}
	\InputIfFileExists{TikZit//normalformwithw.tikz}{}{\input{./figures/TikZit//normalformwithw.tikz}}%
	\endpgfgraphicnamed
$$
 where each green box $a_j$ connects to pink $\pi$ nodes at the locations labelled as  $j_1,\cdots, j_s$ with $j=2^m-1-(2^{j_1}+\cdots+2^{j_s}), 0\leq j_1< \cdots < j_s \leq m-1$. In another words,  $j_1,\cdots, j_s$ are exactly the locations where the coefficients are $0$ in the binary expansion $j=\sum_{i=0}^{m-1}a_i2^i$.
 \end{proposition}
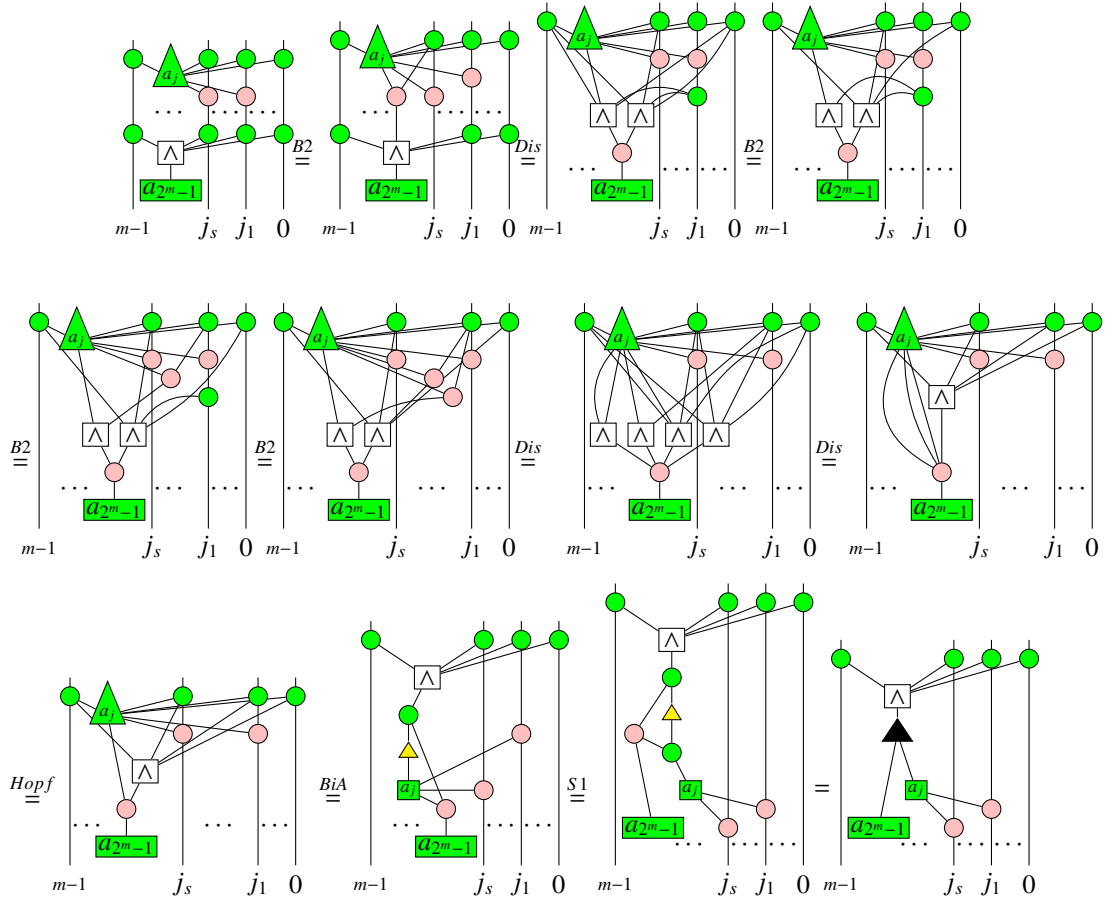
 \begin{proof}
  \[
	\beginpgfgraphicnamed{TikZit//normalformwithwprf}
	\InputIfFileExists{TikZit//normalformwithwprf.tikz}{}{\input{./figures/TikZit//normalformwithwprf.tikz}}%
	\endpgfgraphicnamed

   \]
     \[
	\beginpgfgraphicnamed{TikZit//normalformwithwprf2}
	\InputIfFileExists{TikZit//normalformwithwprf2.tikz}{}{\input{./figures/TikZit//normalformwithwprf2.tikz}}%
	\endpgfgraphicnamed

   \]
   Repeat this processes and by the associative rule (Aso) and the $\pi$ copy rule (B3) we then finish the proof.
  \end{proof}

  \subsection{Completeness}
  
 Completeness means for any two diagrams  $D_1$ and $D_2$, if $\left\llbracket D_1 \right\rrbracket = \left\llbracket D_2 \right\rrbracket $, then $D_1=D_2$ can be derived from the ZX rules. Because of the following map-state duality, we can assume that each diagram is a  state diagram (diagram without any input):  
   \begin{equation}\label{maptostate}
	\beginpgfgraphicnamed{TikZit/mapstatedual2}
	\InputIfFileExists{TikZit/mapstatedual2.tikz}{}{\input{./figures/TikZit/mapstatedual2.tikz}}%
	\endpgfgraphicnamed

\end{equation}
For proof of completeness, assume that  $\left\llbracket D_1 \right\rrbracket = \left\llbracket D_2 \right\rrbracket $. 
The proof strategy is to show that each state diagram (including scalars) can be written into a normal form. Then $D_1$ and $D_2$ must have the same normal form, since $\left\llbracket D_1 \right\rrbracket = \left\llbracket D_2 \right\rrbracket $. So we can rewrite $D_1$ into $D_2$ by first rewriting $D_1$ into its normal form then inverting the rewriting (which is a series of equalities) from $D_2$ to its normal form, which means $D_1=D_2$ can be derived from the ZX rules. 
  
  In general, each state diagram has the following form:
     \begin{equation}\label{generalstatediam}
	\beginpgfgraphicnamed{TikZit/generalstatedm}
	\InputIfFileExists{TikZit/generalstatedm.tikz}{}{\input{./figures/TikZit/generalstatedm.tikz}}%
	\endpgfgraphicnamed

\end{equation}
  where $A_1, A_2, \cdots, A_n$ are diagrams composed of generators in parallel.
 Since $A_1$ is a state diagram, it must be composed of cap or  %
	\beginpgfgraphicnamed{TikZit/greendot}
	\begin{tikzpicture}
	\begin{pgfonlayer}{nodelayer}
		\node [style=gn] (0) at (0, 0) {};
		\node [style=none] (1) at (0, -0.25) {};
	\end{pgfonlayer}
	\begin{pgfonlayer}{edgelayer}
		\draw (0) to (1.center);
	\end{pgfonlayer}
\end{tikzpicture}}%
	\endpgfgraphicnamed
. Also, we have 
 
 $$  %
	\beginpgfgraphicnamed{TikZit/sequentialbend}
	\InputIfFileExists{TikZit/sequentialbend.tikz}{}{\input{./figures/TikZit/sequentialbend.tikz}}%
	\endpgfgraphicnamed
$$
 which means if the tensor of two normal forms can be rewritten into a normal form, a bending of any generator can be rewritten into a normal form, and a self-plugging on a normal form (connecting two outputs of a normal form with a cup) can  be rewritten into a normal form, then a normal form sequentially connected with a diagram composed of tensors of generators can be rewritten into a normal form. As a consequence, any state diagram as depicted in (\ref{generalstatediam}) can be rewritten into a normal form. 
  

To summarise, we need to prove the following statements:
 \begin{enumerate}
 \item the juxtaposition of any two diagrams in normal form can be rewritten
into a normal form.
 \item a self-plugging on a diagram in normal
form can be rewritten into a normal form.
 \item all generators bended in state diagrams or already being state diagrams can be rewritten into normal forms.
  \end{enumerate}
    
  \subsubsection{ Rewrite the tensor product of  two normal forms 
into a normal form}  
First, we note that for $m=0$,  the tensor of two scalar diagrams in normal form is still a normal form:
   $$ %
	\beginpgfgraphicnamed{TikZit//scalarsum}
	\InputIfFileExists{TikZit//scalarsum.tikz}{}{\input{./figures/TikZit//scalarsum.tikz}}%
	\endpgfgraphicnamed
$$
  \begin{proposition}\label{piredonpairpidm}  
 \[
	\beginpgfgraphicnamed{TikZit//piredonpairpidmdm}
	\InputIfFileExists{TikZit//piredonpairpidmdm.tikz}{}{\input{./figures/TikZit//piredonpairpidmdm.tikz}}%
	\endpgfgraphicnamed

\]
    \end{proposition}    
 

Given two normal forms such that

 \begin{equation}\label{normalformeq}
\left\llbracket%
	\beginpgfgraphicnamed{TikZit//normalform3}
	\InputIfFileExists{TikZit//normalform3.tikz}{}{\input{./figures/TikZit//normalform3.tikz}}%
	\endpgfgraphicnamed
\right\rrbracket = \begin{pmatrix}
        a_0  \\  a_1\\
        \vdots \\ a_{2^m-2}\\
        a_{2^m-1} \end{pmatrix}
 \end{equation}	
 and
  \begin{equation}\label{normalformbeq}
	\left\llbracket%
	\beginpgfgraphicnamed{TikZit//normalform3b}
	\InputIfFileExists{TikZit//normalform3b.tikz}{}{\input{./figures/TikZit//normalform3b.tikz}}%
	\endpgfgraphicnamed
\right\rrbracket = \begin{pmatrix}
        b_0  \\  b_1\\
        \vdots \\ b_{2^n-2}\\
        b_{2^n-1} \end{pmatrix}
 \end{equation}	

where $m, n $ are positive integers, $a_i, b_j \in \mathbb{C}$. We want to rewrite the following tensor product of two normal forms into a single normal form:
 \begin{equation}\label{normalformtensoreq}
	\beginpgfgraphicnamed{TikZit//normalform3}
	\InputIfFileExists{TikZit//normalform3.tikz}{}{\input{./figures/TikZit//normalform3.tikz}}%
	\endpgfgraphicnamed
\quad %
	\beginpgfgraphicnamed{TikZit//normalform3b}
	\InputIfFileExists{TikZit//normalform3b.tikz}{}{\input{./figures/TikZit//normalform3b.tikz}}%
	\endpgfgraphicnamed

 \end{equation}	
By naturality of the ambient category for ZX-calculus,  each diagram can move  freely upside and down without going across any other diagrams. Furthermore, the row addition diagrams can commute with each other. Therefore, we can rearrange the order of the diagrams of  elementary transformations in the following way:

 \begin{equation}\label{normalformreplaceeq}
	\beginpgfgraphicnamed{TikZit//normalformreplacedm}
	\InputIfFileExists{TikZit//normalformreplacedm.tikz}{}{\input{./figures/TikZit//normalformreplacedm.tikz}}%
	\endpgfgraphicnamed

  \end{equation}	
where for simplicity we use labelled boxes to represent the corresponding diagrams of  elementary transformations. Then we apply the equalities (\ref{nlinestensornormalformeq}), (\ref{normalformtensornlineseq}), (\ref{nlinestensornormalformaddeq}), (\ref{nlinestensormmultiplywq}) which turn the tensor of diagrammatic elementary transformation and parallel wires into a connected diagram, and we get 
 \begin{equation}\label{normalformreplaceexpeq}
	\beginpgfgraphicnamed{TikZit//normalformreplaceexpdm}
	\InputIfFileExists{TikZit//normalformreplaceexpdm.tikz}{}{\input{./figures/TikZit//normalformreplaceexpdm.tikz}}%
	\endpgfgraphicnamed

  \end{equation}	
where for $0\leq i \leq 2^{m}-1$, $\{a_i\}$ represents the family of $2^n$ elementary transformations labelled by $a_i$ but with different distribution of red $\pi$ pairs on the right-most $n$ wires; and for $0\leq j \leq 2^{n}-1$, $\{b_j\}$ represents the family of $2^m$ elementary transformations labelled by $b_j$ but with different distribution of red $\pi$ pairs on the left-most $m$ wires. By Proposition \ref{addpidoublecom} and  Proposition \ref{multipidoublecom}, diagrams in any  $\{a_i\}$ family or  $\{b_j\}$ family are commutative to each other.  First let  the $a_{2^m-1}$ diagram without any red $\pi$ pair commute with other diagrams in the  $\{a_{2^m-1}\}$ family until it is placed at the bottom of the diagram in (\ref{normalformreplaceexpeq}). Then let the  $b_{2^n-1}$ diagram without any  red $\pi$ pair  first commute with other diagrams in the  $\{b_{2^n-1}\}$ family, and furthermore  commute with all the $a_{2^m-1}$ diagrams with red $\pi$ pairs  by Proposition \ref{addpipairmultiplycommutgp}, finally reaches the place right above the $a_{2^m-1}$ diagram without any red $\pi$ pair. By Corollary \ref{andadditionco},  these two diagrams which have no  red $\pi$ pairs can combined into a single  $a_{2^m-1}b_{2^n-1}$ diagram   without any red $\pi$ pair. 

Below we show step by step how to get all the $a_{i}b_{j}$ diagrams ($0\leq i \leq 2^{m}-2, 0\leq j \leq 2^{n}-2$) without any red $\pi$ pair while cancelling out other diagrams with red $\pi$ pairs.

\begin{enumerate}
\item  \label{it0}   In the $\{a_0\}$ family, by Proposition \ref{piredonpairpidm}, each diagram with  red $\pi$ pair  will be annihilated by the red $\pi$ nodes on top, thus only  $\{a_0\}$ with no  red $\pi$ pair left, which will be denoted as $A_0$.

\item  \label{itd}  In the $\{b_0\}$ family, there are $2^m-1$ diagrams with red $\pi$ pairs, but only one of them has the same location distribution (on the left-most $m$ wires) as that of the pink nodes connected with the output of the $a_0$ box in $A_0$, we denote this diagram by $b_0^{a_0}$.  Then  $b_0^{a_0}$ commutes with other diagrams of  the $\{b_0\}$ family. It also commutes with each diagram of the families $\{a_1\}, \{a_2\}, \cdots, \{a_{2^m-2}\}$, following Proposition \ref{addpipair2sidecommutprop28}. Then  $b_0^{a_0}$ can commute to be right below $A_0$.  Now by  Proposition \ref{pitopaddpipaircommutprop}, we get the diagram $a_0b_0$  (without any  red $\pi$ pair)  and $A_0$ again, while $b_0^{a_0}$ is consumed. By Proposition \ref{addcommutatgencont},  $a_0b_0$ commutes with $A_0$. Furthermore,  $a_0b_0$ commutes with each diagram in the $\{a_i\}$ family, $0\leq i \leq 2^{m}-2$, by Proposition \ref{addpipair2sidecommutprop29};  commutes with each diagram in the $\{b_i\}$ family, $0\leq i \leq 2^{n}-2$, by Proposition \ref{addpipair2sidecommutprop29b};   commutes with each diagram in the $\{b_{2^n-1}\}$ family (with red $\pi$ pairs), by Corollary \ref{addpipairmulcommutprop30bcro};  and commutes with each diagram in the $\{a_{2^m-1}\}$ family (with red $\pi$ pairs), by Proposition \ref{addpipairmulcommutprop30b}. Therefore $a_0b_0$ can be commuted to be right above the diagram  $a_{2^m-1}b_{2^n-1}$.

\item \label{it2}   Similarly, we can choose  $b_i^{a_0}$ from the $\{b_i\}$ family, for $1\leq i \leq 2^{n}-2$, where  $b_i^{a_0}$  has  the same location distribution (on the left-most $m$ wires) of red $\pi$ pairs as that of the pink nodes connected with the output of the $a_0$ box in $A_0$. Repeating the method of step  \ref{itd}, we get diagrams  $a_0b_i$ located sequentially right above $a_0b_0$,  for $1\leq i \leq 2^{n}-2$, with all those $b_i^{a_0}$ eliminated. 

\item \label{it3}  Now choose $b_{2^n-1}^{a_0}$ from the $\{b_{2^n-1}\}$ family, where $b_{2^n-1}^{a_0}$ has  the same location distribution (on the left-most $m$ wires) of red $\pi$ pairs as that of the pink nodes connected with the output of the $a_0$ box in $A_0$. Since all the $b_i^{a_0}, 0\leq i \leq 2^{n}-2$ have been eliminated, by  Proposition \ref{addpipairmultiplycommutgp}, $b_{2^n-1}^{a_0}$ can commute with all the remaining diagrams of the families $b_i, 0\leq i \leq 2^{n}-2$, and commute with all the diagrams with red $\pi$ pairs in the $\{a_i\}$ families, for $1\leq i \leq 2^{m}-2$. Furthermore,   the location distribution (on the left-most $m$ wires) of red $\pi$ pairs in $b_{2^n-1}^{a_0}$ is different from the  location distribution (on the left-most $m$ wires) of pink nodes connected with the output of the $a_i$ for  $1\leq i \leq 2^{m}-2$. Then  by  Corollary \ref{addpipairmulcommutprop30ccro},   $b_{2^n-1}^{a_0}$  commutes with each diagram with no red $\pi$ pairs in the families $\{a_i\}$  for $1\leq i \leq 2^{m}-2$. 
Therefore,  $b_{2^n-1}^{a_0}$ can commute to be right below $A_0$. By  Proposition \ref{pimultiaddcombinepro}, we get $a_0b_{2^n-1}$  with no red $\pi$ pairs, while $b_{2^n-1}^{a_0}$ and $A_0$ are eliminated. Then following the step  \label{itd}, $a_0b_{2^n-1}$ can be commuted to the location right above $a_0b_i, 0\leq i \leq 2^{n}-2$.    So far, we get the products  $a_0b_i, 0\leq i \leq 2^{n}-1$, and all the diagrams in the  $\{a_0\}$ family have been eliminated,  the diagrams $b_i^{a_0},$ $1\leq i \leq 2^{n}-1$, are also eliminated, while all the diagrams in the families $\{a_i \}, 1\leq i \leq 2^{m}-2, $ remain unchanged. 

\item Repeat the steps  \ref{it0},  \ref{itd},  \ref{it2},  \ref{it3} for the families $\{a_i \}, 1\leq i \leq 2^{m}-2, $ and $\{b_j \}, 0\leq j \leq 2^{n}-1, $ we  get the products  $a_ib_j, 1\leq i \leq 2^{m}-2, 0\leq j \leq 2^{n}-1$ at the bottom part. Now what remains is the following diagrams (from top to bottom):  families $\{b_j \}, 0\leq j \leq 2^{n}-2, $ without red $\pi$ pairs; family $\{a_{2^m-1}\}$  with red $\pi$ pairs; the products  $a_ib_j, 0\leq i \leq 2^{m}-2, 0\leq j \leq 2^{n}-1$; and the product  $a_{2^m-1}b_{2^n-1}$. We still need to get the products $a_{2^m-1}b_j,  0\leq j \leq 2^{n}-2$ and eliminate the non-product diagrams.  

\item \label{it4} Now each  $\{b_i \}$ family without red $\pi$ pairs just has one diagram $b_i, 0\leq i \leq 2^{n}-2$. Denote by $a_{2^m-1}^{b_i}$  the diagram in the  family $\{a_{2^m-1}\}$  that has the same location distribution (on the right-most $n$ wires) of red $\pi$ pairs as that of the pink nodes connected with the output of  $b_i, 0\leq i \leq 2^{n}-2$.  By  Proposition \ref{addpipairmulcommutprop30c}, $a_{2^m-1}^{b_0}$ can commute with  each $b_i, 1\leq i \leq 2^{n}-2$ to the location just below $b_0$. Then by Proposition \ref{pimultiaddcombinepro}, we get $a_{2^m-1}b_0$ with no red $\pi$ pairs, while $a_{2^m-1}^{b_0}$ and $b_0$ are eliminated. $a_{2^m-1}b_0$ can commute with each $b_i, 1\leq i \leq 2^{n}-2$, and commute with the remaining diagrams in family $\{a_{2^m-1}\}$  with red $\pi$ pairs, by Proposition \ref{addpipairmulcommutprop30c}. Therefore, $a_{2^m-1}b_0$  can commute to the location just above the diagrams with products. 

\item Repeat the step \ref{it4} for $b_i, 1\leq i \leq 2^{n}-2$, we can eliminate 
$a_{2^m-1}^{b_i}$ and $b_i$, while get $a_{2^m-1}b_i$ for $1\leq i \leq 2^{n}-2$. Now what remains is the red $\pi$ nodes on top and all the products $a_ib_j, 0\leq i \leq 2^{m}-1, 0\leq j \leq 2^{n}-1$, which is a normal form.
\end{enumerate}

\subsubsection{ Self-plugging on a normal form}
  \begin{lemma}\label{3and3gdotcirc}
	\beginpgfgraphicnamed{TikZit//3and3gdotcircs}
	\InputIfFileExists{TikZit//3and3gdotcircs.tikz}{}{\input{./figures/TikZit//3and3gdotcircs.tikz}}%
	\endpgfgraphicnamed

 \end{lemma}
  \begin{proof}
   $$  %
	\beginpgfgraphicnamed{TikZit//3and3gdotcircsprf}
	\InputIfFileExists{TikZit//3and3gdotcircsprf.tikz}{}{\input{./figures/TikZit//3and3gdotcircsprf.tikz}}%
	\endpgfgraphicnamed
 $$
  \end{proof}

   \begin{lemma}\label{3and3gdotcircsimp}
	\beginpgfgraphicnamed{TikZit//3and3gdotcircssimp}
	\InputIfFileExists{TikZit//3and3gdotcircssimp.tikz}{}{\input{./figures/TikZit//3and3gdotcircssimp.tikz}}%
	\endpgfgraphicnamed

 \end{lemma}
  \begin{proof}
 $$  %
	\beginpgfgraphicnamed{TikZit//3and3gdotcircssimpprf}
	\InputIfFileExists{TikZit//3and3gdotcircssimpprf.tikz}{}{\input{./figures/TikZit//3and3gdotcircssimpprf.tikz}}%
	\endpgfgraphicnamed
 $$
  \end{proof}  

 \begin{proposition}\label{rule10}  
  \begin{equation}\label{TR17}
	\beginpgfgraphicnamed{TikZit//phaseaddition}
	\InputIfFileExists{TikZit//phaseaddition.tikz}{}{\input{./figures/TikZit//phaseaddition.tikz}}%
	\endpgfgraphicnamed
 \end{equation}
  \end{proposition}
 \begin{proof}
\[  %
	\beginpgfgraphicnamed{TikZit//phaseadditionprf}
	\InputIfFileExists{TikZit//phaseadditionprf.tikz}{}{\input{./figures/TikZit//phaseadditionprf.tikz}}%
	\endpgfgraphicnamed
 \]
 \end{proof}

   \begin{corollary}\label{rule10exten}
\[    %
	\beginpgfgraphicnamed{TikZit//phaseadditionexten}
	\InputIfFileExists{TikZit//phaseadditionexten.tikz}{}{\input{./figures/TikZit//phaseadditionexten.tikz}}%
	\endpgfgraphicnamed
 \]
  \end{corollary}

 \begin{proof}
\[  %
	\beginpgfgraphicnamed{TikZit//phaseadditionextenprf}
	\InputIfFileExists{TikZit//phaseadditionextenprf.tikz}{}{\input{./figures/TikZit//phaseadditionextenprf.tikz}}%
	\endpgfgraphicnamed
 \]
 \end{proof}


    \begin{proposition}\label{rule12th}
    \[ %
	\beginpgfgraphicnamed{TikZit//rule12th}
	\InputIfFileExists{TikZit//rule12th.tikz}{}{\input{./figures/TikZit//rule12th.tikz}}%
	\endpgfgraphicnamed
     \]
         \end{proposition}
      \begin{proof}
     If $a=0$, then the equality holds trivially. Now we assume $a\neq 0$. Then
   $$  %
	\beginpgfgraphicnamed{TikZit//rule12thprf}
	\InputIfFileExists{TikZit//rule12thprf.tikz}{}{\input{./figures/TikZit//rule12thprf.tikz}}%
	\endpgfgraphicnamed
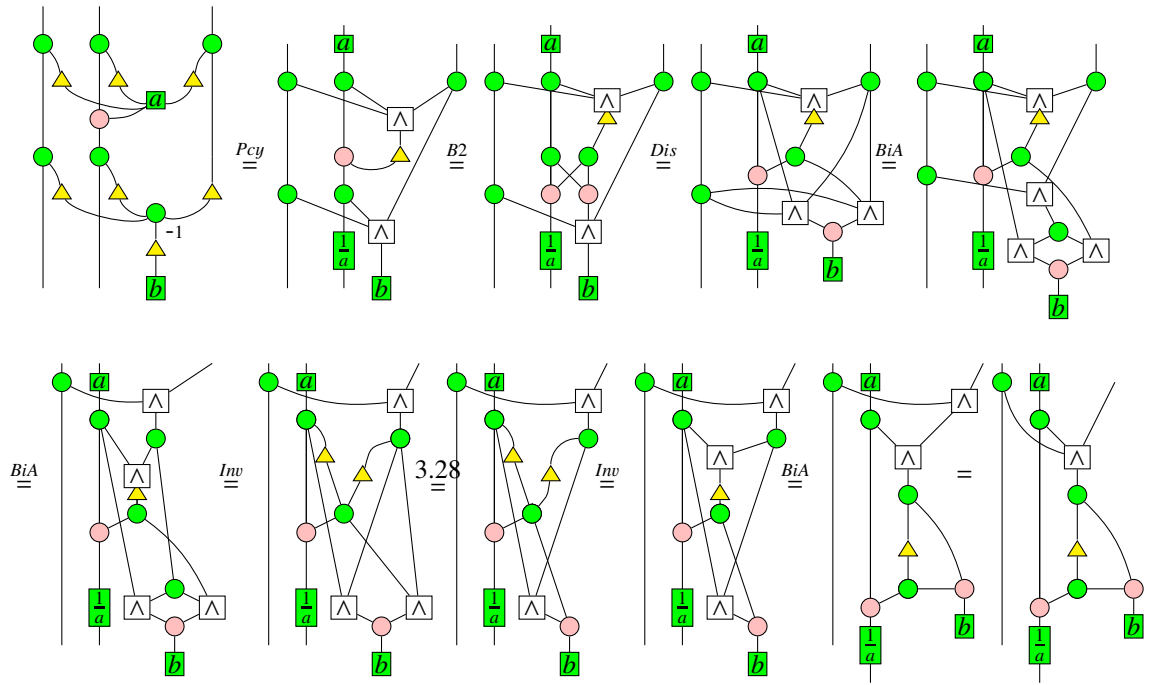
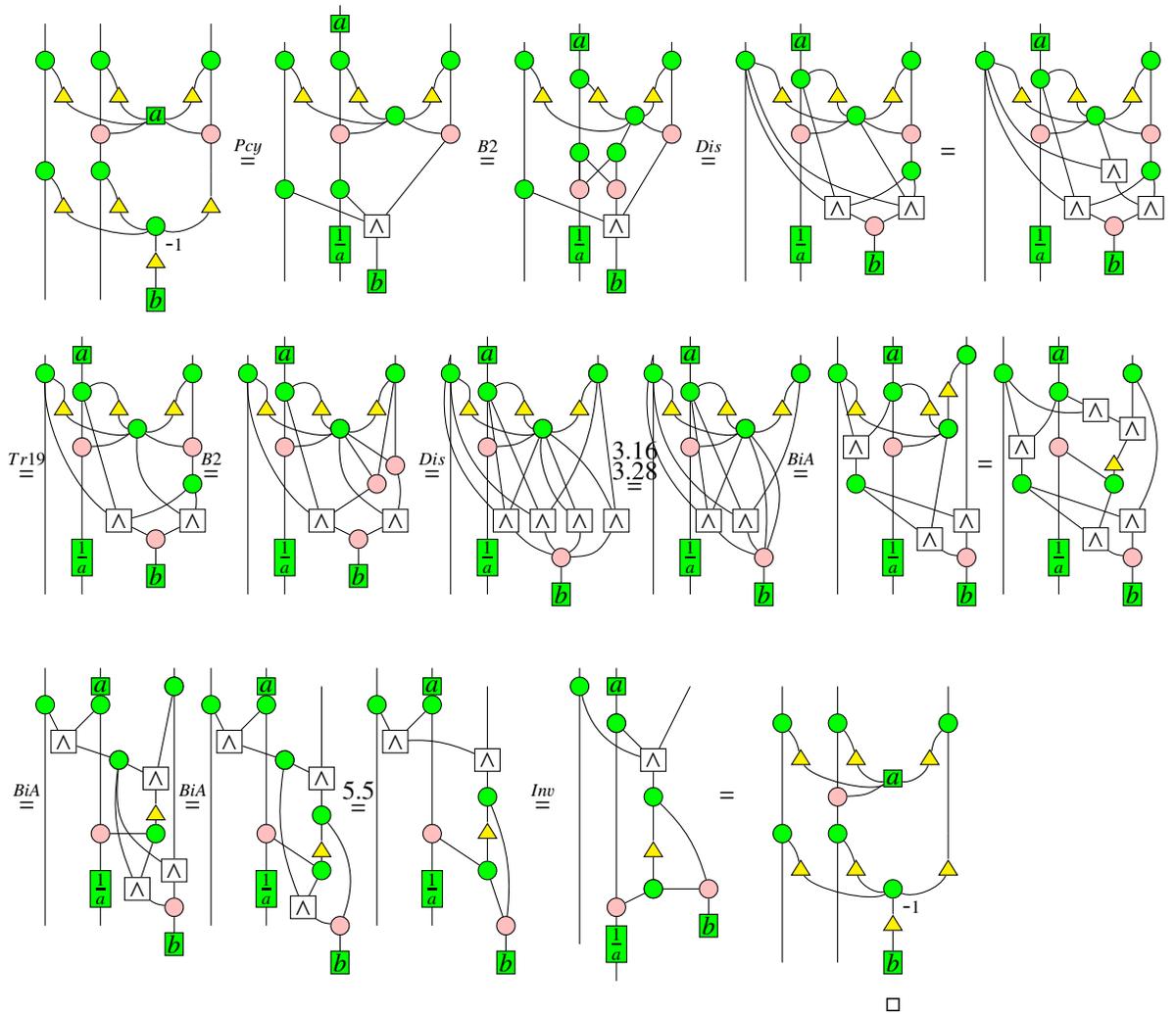
 $$
    $$  %
	\beginpgfgraphicnamed{TikZit//rule12thprf2}
	\InputIfFileExists{TikZit//rule12thprf2.tikz}{}{\input{./figures/TikZit//rule12thprf2.tikz}}%
	\endpgfgraphicnamed
 $$
  \end{proof}
  
    \begin{corollary}\label{rule12thexten}
    Suppose $a$ is connected to the $i$-th line via a pink node ($i>1$). Then
      \begin{equation}\label{rule12thexteneq}
	\beginpgfgraphicnamed{TikZit//rule12thextendm}
	\InputIfFileExists{TikZit//rule12thextendm.tikz}{}{\input{./figures/TikZit//rule12thextendm.tikz}}%
	\endpgfgraphicnamed
 
     \end{equation}
    where on the left side of  (\ref{rule12thexteneq}) the node $a$ is connected to the $i$-th line and the 1st line via two pink nodes.
         \end{corollary}
   \begin{proof}
  For simplicity, we prove with the short notation as shown in  (\ref{andshortnotationeq}).  We assume w.l.g that $i=2$, otherwise we could swap the $i$-th line and the 2nd line on both sides of  (\ref{rule12thexteneq}).
   Therefore
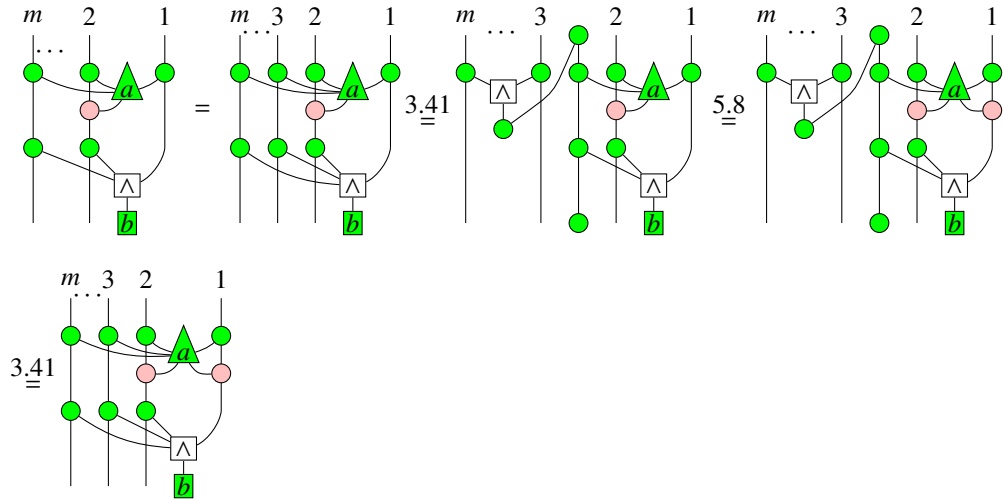
    \[ %
	\beginpgfgraphicnamed{TikZit//rule12thextendmprf}
	\InputIfFileExists{TikZit//rule12thextendmprf.tikz}{}{\input{./figures/TikZit//rule12thextendmprf.tikz}}%
	\endpgfgraphicnamed
 \]
    \end{proof}

      \begin{proposition}\label{rule12extengen}
       Suppose $a$ is connected to the $i_1, \cdots, i_s$  via  pink nodes ($i_1>1$). Then one more connection can be added on the right most line:
         \begin{equation}\label{rule12extengeneq}
	\beginpgfgraphicnamed{TikZit//rule12extengendm}
	\InputIfFileExists{TikZit//rule12extengendm.tikz}{}{\input{./figures/TikZit//rule12extengendm.tikz}}%
	\endpgfgraphicnamed
 
       \end{equation}
         \end{proposition}
      \begin{proof}
 We assume w.l.g that $i_1=2$, otherwise we could swap the $i_1$-th line and the 2nd line on both sides of  (\ref{rule12extengeneq}).
   Therefore
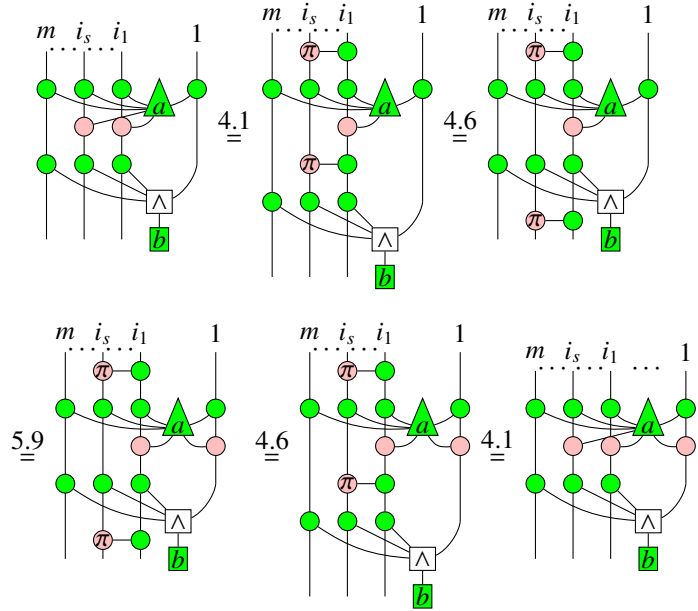
    \[ %
	\beginpgfgraphicnamed{TikZit//rule12extengendmprf}
	\InputIfFileExists{TikZit//rule12extengendmprf.tikz}{}{\input{./figures/TikZit//rule12extengendmprf.tikz}}%
	\endpgfgraphicnamed
 \]
    \end{proof}
     
 \begin{theorem} \label{selfplug}  
 A self-plugged normal form can be rewritten into a normal form. 
   \end{theorem}      
    \begin{proof} 
    Given a normal form such that

 \begin{equation}\label{normalform3snteq}
\left\llbracket%
	\beginpgfgraphicnamed{TikZit//normalform3snt}
	\InputIfFileExists{TikZit//normalform3snt.tikz}{}{\input{./figures/TikZit//normalform3snt.tikz}}%
	\endpgfgraphicnamed
\right\rrbracket = \begin{pmatrix}
        a_0  \\  a_1\\
        \vdots \\ a_{2^m-2}\\
        a_{2^m-1} \end{pmatrix}
 \end{equation}	
Because of symmetry, we can assume that  the normal form is self-plugged on the right-most two lines and $m\geq 2$:
 \begin{equation}\label{normalform3snteq2}
	\beginpgfgraphicnamed{TikZit//normalform3sntplug}
	\InputIfFileExists{TikZit//normalform3sntplug.tikz}{}{\input{./figures/TikZit//normalform3sntplug.tikz}}%
	\endpgfgraphicnamed

       \end{equation}	
We do the rewriting in an order from bottom to top.
\begin{enumerate}
\item \label{2m1}
Consider the following 
rewriting:
 \begin{equation}\label{normalformptrace1eq}
	\beginpgfgraphicnamed{TikZit//normalformptrace1dm}
	\InputIfFileExists{TikZit//normalformptrace1dm.tikz}{}{\input{./figures/TikZit//normalformptrace1dm.tikz}}%
	\endpgfgraphicnamed

       \end{equation}	
       
\item \label{2m4} Consider the $a_{2^m-4}$ term in addition. Since $2^m-4=2^m-1-(2^0+2^1)$, $a_{2^m-4}$ must connect to the 0-th and 1-st lines on the right-most. Then based on (\ref{normalformptrace1eq}), we have 
 \begin{equation}\label{normalformptrace2eq}
	\beginpgfgraphicnamed{TikZit//normalformptrace2dm}
	\InputIfFileExists{TikZit//normalformptrace2dm.tikz}{}{\input{./figures/TikZit//normalformptrace2dm.tikz}}%
	\endpgfgraphicnamed

       \end{equation}
where $c=a_{2^m-4}+a_{2^m-1}$.
 
\item Consider $a_{4k+1}, 0\leq k \leq 2^{m-2}-1$.   Since $4k+1=2^m-1-(2^{j_1}+\cdots+2^{j_s})=4\cdot2^{m-2}+1-(2+2^{j_1}+\cdots+2^{j_s}), j_1< \cdots < j_s$, there must exist some $j_i$ such that $j_i=1$, and $j_k\neq 0, k=1, \cdots, s$. This means $a_{4k+1}$ must connect to the 1-st line via a pink node while not connected to the  0-th line. Then we have
\begin{equation}\label{normalformptrace4eq}
	\beginpgfgraphicnamed{TikZit//normalformptrace4dm}
	\InputIfFileExists{TikZit//normalformptrace4dm.tikz}{}{\input{./figures/TikZit//normalformptrace4dm.tikz}}%
	\endpgfgraphicnamed

       \end{equation}
where $c=a_{2^m-4}+a_{2^m-1}$. So $a_{4k+1}$ are all eliminated.

\item Consider $a_{4k+2}, 0\leq k \leq 2^{m-2}-1$.   Since $4k+2=2^m-1-(2^{j_1}+\cdots+2^{j_s})=4\cdot2^{m-2}+2-(3+2^{j_1}+\cdots+2^{j_s}), j_1< \cdots < j_s$, there must exist some $j_i$ such that $j_i=0$, and $j_k\neq 1, k=1, \cdots, s$. This means $a_{4k+2}$ must connect to the 0-th line via a pink node while not connected to the  1-st  line. Then we have
\begin{equation}\label{normalformptrace5eq}
	\beginpgfgraphicnamed{TikZit//normalformptrace5dm}
	\InputIfFileExists{TikZit//normalformptrace5dm.tikz}{}{\input{./figures/TikZit//normalformptrace5dm.tikz}}%
	\endpgfgraphicnamed

       \end{equation}
where $c=a_{2^m-4}+a_{2^m-1}$. So $a_{4k+2}$ are all eliminated.

\item Now we consider the case $m=2$. From the previous cases, $a_1$ and $a_2$ are eliminated, $a_0$ and $a_3$ are combined as in (\ref{normalformptrace2eq}). Therefor we have
\begin{equation}\label{normalformptracemis2}
	\beginpgfgraphicnamed{TikZit//normalformptracemis2dm}
	\InputIfFileExists{TikZit//normalformptracemis2dm.tikz}{}{\input{./figures/TikZit//normalformptracemis2dm.tikz}}%
	\endpgfgraphicnamed

       \end{equation}
 where $c=a_{0}+a_{3}$.  Thus for $m=2$, we get a normal form.
 
\item \label{4k} Consider  $a_{4k}, 0\leq k \leq 2^{m-2}-1$. If $m=2$, it is just the case \ref{2m4} which we already considered. Now we consider $m\geq 3$ and $4k\neq 2^m-4$.     Since $4k=2^m-1-(2^{j_1}+\cdots+2^{j_s})=4\cdot2^{m-2}-(1+2^{j_1}+\cdots+2^{j_s}), j_1< \cdots < j_s$, there must be that $j_1=0, j_2=1$, and there is another $j_k \geq 2$, i.e., $4k=2^m-4-(2^{j_{k_1}}+\cdots+2^{j_{k_t}}), j_{k_i} \geq 2$. This means  $a_{4k}$ is connected with the right-most two wires and at least one wire among the left $m-2$ wires via pink nodes. Therefore,
   \begin{equation}\label{normalformptrace3eq}
	\beginpgfgraphicnamed{TikZit//normalformptrace3dm}
	\InputIfFileExists{TikZit//normalformptrace3dm.tikz}{}{\input{./figures/TikZit//normalformptrace3dm.tikz}}%
	\endpgfgraphicnamed

       \end{equation}
where $c=a_{2^m-4}+a_{2^m-1}$.

\item  Consider  $a_{4k+3}, 0\leq k \leq 2^{m-2}-1$. If $m=2$, it is just the case \ref{2m1} which we already considered. Now we consider $m\geq 3$ and $0\leq k < 2^{m-2}-1$.   By case \ref{4k}, $4k+3=2^m-4-(2^{j_{k_1}}+\cdots+2^{j_{k_t}})+3=2^m-1-(2^{j_{k_1}}+\cdots+2^{j_{k_t}}), j_{k_i} \geq 2$. 
This means  $a_{4k+3}$ is not connected with the right-most two wires but with at least one wire among the left $m-2$ wires via pink nodes in the same way as the case of $a_{4k}$. Therefore,
   \begin{equation}\label{normalformptrace6eq}
	\beginpgfgraphicnamed{TikZit//normalformptrace6dm}
	\InputIfFileExists{TikZit//normalformptrace6dm.tikz}{}{\input{./figures/TikZit//normalformptrace6dm.tikz}}%
	\endpgfgraphicnamed

       \end{equation}
where $c=a_{2^m-4}+a_{2^m-1}, b_k=a_{4k+3}+a_{4k},  0\leq k \leq 2^{m-2}-2$.

\item Finally the self-plugged normal form (\ref{normalform3snteq2}) becomes
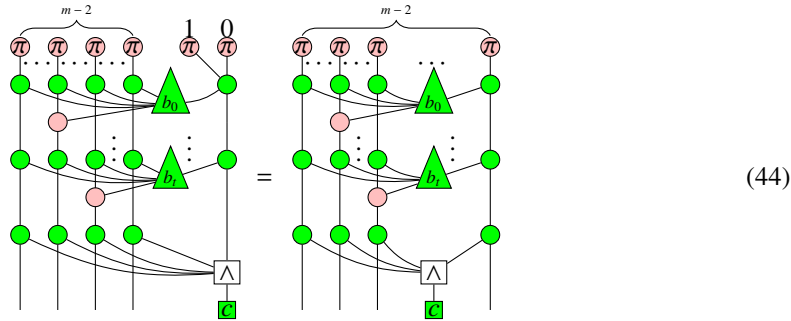
 \begin{equation}\label{normalformptrace7eq}
	\beginpgfgraphicnamed{TikZit//normalformptrace7dm}
	\InputIfFileExists{TikZit//normalformptrace7dm.tikz}{}{\input{./figures/TikZit//normalformptrace7dm.tikz}}%
	\endpgfgraphicnamed

       \end{equation}
where $c=a_{2^m-4}+a_{2^m-1}, b_k=a_{4k+3}+a_{4k},  0\leq k \leq 2^{m-2}-2, t=2^{m-2}-2$. Clearly, the RHS of (\ref{normalformptrace7eq}) is a normal form now.

\end{enumerate}

      \end{proof}

   \subsubsection{ Rewriting generators into normal form }
   In this section, we prove that all the generators bended in state diagrams or already being state diagrams can be rewritten into normal forms.
  
   \begin{enumerate}
   \item  For the generator  %
	\beginpgfgraphicnamed{TikZit//generalgreenspider}
	\InputIfFileExists{TikZit//generalgreenspider.tikz}{}{\input{./figures/TikZit//generalgreenspider.tikz}}%
	\endpgfgraphicnamed
, we only need consider the diagrams in (\ref{finitegenerators})  according to   Remark \ref{finitegeneratorsrk}: 
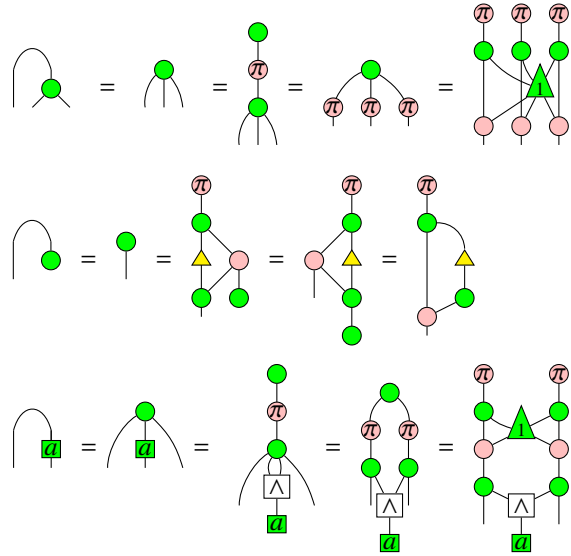
 \[   %
	\beginpgfgraphicnamed{TikZit//gcopytonormalfm}
	\InputIfFileExists{TikZit//gcopytonormalfm.tikz}{}{\input{./figures/TikZit//gcopytonormalfm.tikz}}%
	\endpgfgraphicnamed
\]
    \[   %
	\beginpgfgraphicnamed{TikZit//gcodottonormalfm}
	\InputIfFileExists{TikZit//gcodottonormalfm.tikz}{}{\input{./figures/TikZit//gcodottonormalfm.tikz}}%
	\endpgfgraphicnamed
\]
       \[   %
	\beginpgfgraphicnamed{TikZit//singleanodetonormalfm}
	\InputIfFileExists{TikZit//singleanodetonormalfm.tikz}{}{\input{./figures/TikZit//singleanodetonormalfm.tikz}}%
	\endpgfgraphicnamed
\]
     
    \item  For the generators %
	\beginpgfgraphicnamed{TikZit//Id}
	\begin{tikzpicture}
	\begin{pgfonlayer}{nodelayer}
		\node [style=none] (1) at (0.5, 0.3) {};
		\node [style=none] (2) at (0.5, -0.3) {};
		\node [style=none] (3) at (0.5, -0.5) {};
		\node [style=none] (4) at (0.5, 0.5) {};
	\end{pgfonlayer}
	\begin{pgfonlayer}{edgelayer}
		\draw (1.center) to (2.center);
	\end{pgfonlayer}
\end{tikzpicture}}%
	\endpgfgraphicnamed
,   %
	\beginpgfgraphicnamed{TikZit//cap}
	\begin{tikzpicture}
	\begin{pgfonlayer}{nodelayer}
		\node [style=none] (0) at (0, -0) {};
		\node [style=none] (1) at (1, -0) {};
	\end{pgfonlayer}
	\begin{pgfonlayer}{edgelayer}
		\draw [bend left=90, looseness=1.50] (0.center) to (1.center);
	\end{pgfonlayer}
\end{tikzpicture}}%
	\endpgfgraphicnamed
 and  %
	\beginpgfgraphicnamed{TikZit//cup}
	\begin{tikzpicture}
	\begin{pgfonlayer}{nodelayer}
		\node [style=none] (0) at (0, 0.5) {};
		\node [style=none] (1) at (1, 0.5) {};
	\end{pgfonlayer}
	\begin{pgfonlayer}{edgelayer}
		\draw [bend right=90, looseness=1.50] (0.center) to (1.center);
	\end{pgfonlayer}
\end{tikzpicture}}%
	\endpgfgraphicnamed
 , we have 
       \begin{equation}\label{idtonormalfmeq}
	\beginpgfgraphicnamed{TikZit//idtonormalfm}
	\InputIfFileExists{TikZit//idtonormalfm.tikz}{}{\input{./figures/TikZit//idtonormalfm.tikz}}%
	\endpgfgraphicnamed

     \end{equation}
 \item  For the generator %
	\beginpgfgraphicnamed{TikZit//newhadamard}
	\begin{tikzpicture}
	\begin{pgfonlayer}{nodelayer}
		\node [style=newh] (0) at (0, 0) {};
		\node [style=none] (1) at (0, 0.5) {};
		\node [style=none] (2) at (0, -0.5) {};
	\end{pgfonlayer}
	\begin{pgfonlayer}{edgelayer}
		\draw (1.center) to (2.center);
	\end{pgfonlayer}
\end{tikzpicture}}%
	\endpgfgraphicnamed
, we have 
   \[   %
	\beginpgfgraphicnamed{TikZit//hadtonormalfm}
	\InputIfFileExists{TikZit//hadtonormalfm.tikz}{}{\input{./figures/TikZit//hadtonormalfm.tikz}}%
	\endpgfgraphicnamed
\]
          
    \item Since %
	\beginpgfgraphicnamed{TikZit//triangleinvandtriangle}
	\InputIfFileExists{TikZit//triangleinvandtriangle.tikz}{}{\input{./figures/TikZit//triangleinvandtriangle.tikz}}%
	\endpgfgraphicnamed
 , we only need to consider   %
	\beginpgfgraphicnamed{TikZit//triangle}
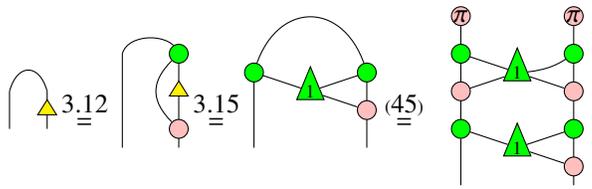
	\begin{tikzpicture}
	\begin{pgfonlayer}{nodelayer}
		\node [style=none] (0) at (0, 0.5) {};
		\node [style=triangle] (1) at (0, 0) {};
		\node [style=none] (2) at (0, -0.5) {};
	\end{pgfonlayer}
	\begin{pgfonlayer}{edgelayer}
		\draw (0.center) to (2.center);
	\end{pgfonlayer}
\end{tikzpicture}}%
	\endpgfgraphicnamed
:
  \[   %
	\beginpgfgraphicnamed{TikZit//triangletonormalfm}
	\InputIfFileExists{TikZit//triangletonormalfm.tikz}{}{\input{./figures/TikZit//triangletonormalfm.tikz}}%
	\endpgfgraphicnamed
\]
  
  \item  For the generator %
	\beginpgfgraphicnamed{TikZit//swap}
	\InputIfFileExists{TikZit//swap.tikz}{}{\input{./figures/TikZit//swap.tikz}}%
	\endpgfgraphicnamed
, first we have 
  \[   %
	\beginpgfgraphicnamed{TikZit//swaptonormalfm}
	\InputIfFileExists{TikZit//swaptonormalfm.tikz}{}{\input{./figures/TikZit//swaptonormalfm.tikz}}%
	\endpgfgraphicnamed
\]
   Then
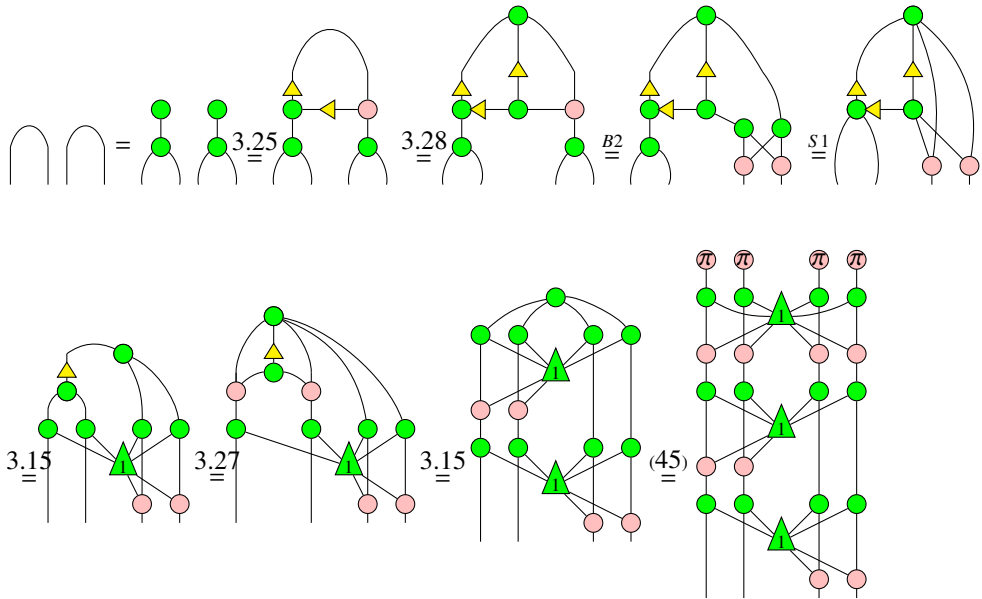
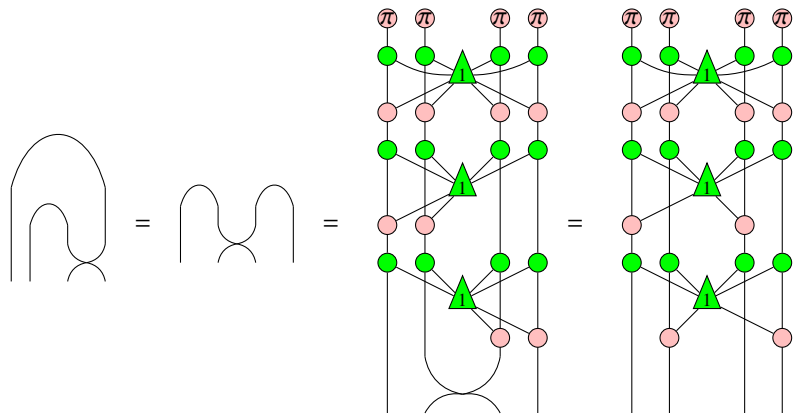
   \[   %
	\beginpgfgraphicnamed{TikZit//swaptonormalfm2}
	\InputIfFileExists{TikZit//swaptonormalfm2.tikz}{}{\input{./figures/TikZit//swaptonormalfm2.tikz}}%
	\endpgfgraphicnamed
\]  
   \end{enumerate}
   
\subsubsection{ Rewriting scalars into normal form }
Given an arbitrary scalar diagram, it can be seen as a state diagram on top plugged with cups or %
	\beginpgfgraphicnamed{TikZit//greencodot}
	\begin{tikzpicture}
	\begin{pgfonlayer}{nodelayer}
		\node [style=none] (0) at (0, 0.25) {};
		\node [style=gn] (1) at (0, 0) {};
	\end{pgfonlayer}
	\begin{pgfonlayer}{edgelayer}
		\draw (1) to (0.center);
	\end{pgfonlayer}
\end{tikzpicture}}%
	\endpgfgraphicnamed
 from the bottom. Such a state diagram has been shown that it can be rewritten into a normal form, and a normal form plugged with  %
	\beginpgfgraphicnamed{TikZit//greencodot}
	\begin{tikzpicture}
	\begin{pgfonlayer}{nodelayer}
		\node [style=none] (0) at (0, 0.25) {};
		\node [style=gn] (1) at (0, 0) {};
	\end{pgfonlayer}
	\begin{pgfonlayer}{edgelayer}
		\draw (1) to (0.center);
	\end{pgfonlayer}
\end{tikzpicture}}%
	\endpgfgraphicnamed
 can be bended up so that we have can still have a normal form with outputs connected by just cups. 
Since normal form with more than 2 outputs plugged with cups can be reduced to  a normal form with outputs as proved in Theorem \ref{selfplug},
 we just need to show that a normal form with exactly 2 outputs plugged with a cup can be written into a scalar normal form  %
	\beginpgfgraphicnamed{TikZit//scalarnorm}
	\begin{tikzpicture}
	\begin{pgfonlayer}{nodelayer}
		\node [style=rn] (0) at (0, 0.25) {$\pi$};
		\node [style=gbox, scale=1] (1) at (0, -0.25) {$a$};
	\end{pgfonlayer}
	\begin{pgfonlayer}{edgelayer}
		\draw (0) to (1);
	\end{pgfonlayer}
\end{tikzpicture}}%
	\endpgfgraphicnamed
, which has also been obtained in Theorem \ref{selfplug}.


\section{Conclusion and further work}    
In this paper, we give an algebraic axiomatisation of qubit ZX-calculus and prove its completeness by constructing a normal form of an arbitrary state diagram via elementary transformations in linear algebra. In comparing to previous complete axiomatisation of ZX-calculus, this axiomatisation  has no need to resort to either rules with trigonometry functions or other completeness results. Also the techniques used to achieve completeness here are very useful for application of this formalism. 

 There is lot of work following this direction. First, by the same method of this paper, we will fill all proof details in \cite{qwangrsmring}  for ZX-calculus over arbitrary commutative rings and semi-rings.  Second, we can generalise the completeness of qubit  ZX-calculus to qudit versions which have similar normal form as presented in  \cite{qwangslides}. Furthermore, since all different dimensional qudit ZX-calculus can be put together in a single framework called qufinite ZX-calculus as described in \cite{qwangslides},  the completeness could be obtained similarly as proved in this paper. As a consequence, one could reconstruct quantum theory using the qufinite ZX-calculus.
 
  On the application side,  we can utilise the techniques invented in this paper to simplify quantum circuits. Also with the diagrammatic representation of elementary matrices, we can solve linear systems in ZX-calculus  \cite{wanglinearstm}, which would be helpful for  tackling problems in quantum machine learning.

 \section*{Acknowledgements} 

This work is supported by AFOSR grant FA2386-18-1-4028. The author would like to thank Bob Coecke, KangFeng  Ng, and Aleks Kissinger for useful discussions. 

\bibliographystyle{eptcs}
\bibliography{generic}

\begin{thebibliography}{10}
\providecommand{\bibitemdeclare}[2]{}
\providecommand{\surnamestart}{}
\providecommand{\surnameend}{}
\providecommand{\urlprefix}{Available at }
\providecommand{\url}[1]{\texttt{#1}}
\providecommand{\href}[2]{\texttt{#2}}
\providecommand{\urlalt}[2]{\href{#1}{#2}}
\providecommand{\doi}[1]{doi:\urlalt{http://dx.doi.org/#1}{#1}}
\providecommand{\bibinfo}[2]{#2}

\bibitemdeclare{article}{msw2017}
\bibitem{msw2017}
\bibinfo{author}{Miriam \surnamestart Backens\surnameend},
  \bibinfo{author}{Simon \surnamestart Perdrix\surnameend} \&
  \bibinfo{author}{Quanlong \surnamestart Wang\surnameend}
  (\bibinfo{year}{2017}): \emph{\bibinfo{title}{A Simplified Stabilizer
  ZX-calculus}}.
\newblock {\sl \bibinfo{journal}{Electronic Proceedings in Theoretical Computer
  Science}} \bibinfo{volume}{236}, pp. \bibinfo{pages}{1--20},
  \doi{10.4204/eptcs.236.1}.
\newblock \urlprefix\url{http://dx.doi.org/10.4204/EPTCS.236.1}.

\bibitemdeclare{inproceedings}{debeaudrapbianwang}
\bibitem{debeaudrapbianwang}
\bibinfo{author}{Niel \surnamestart de~Beaudrap\surnameend},
  \bibinfo{author}{Xiaoning \surnamestart Bian\surnameend} \&
  \bibinfo{author}{Quanlong \surnamestart Wang\surnameend}
  (\bibinfo{year}{2020}): \emph{\bibinfo{title}{{Fast and Effective Techniques
  for T-Count Reduction via Spider Nest Identities}}}.
\newblock In \bibinfo{editor}{Steven~T. \surnamestart Flammia\surnameend},
  editor: {\sl \bibinfo{booktitle}{15th Conference on the Theory of Quantum
  Computation, Communication and Cryptography (TQC 2020)}}, {\sl
  \bibinfo{series}{Leibniz International Proceedings in Informatics (LIPIcs)}}
  \bibinfo{volume}{158}, \bibinfo{publisher}{Schloss Dagstuhl--Leibniz-Zentrum
  f{\"u}r Informatik}, \bibinfo{address}{Dagstuhl, Germany}, pp.
  \bibinfo{pages}{11:1--11:23}, \doi{10.4230/LIPIcs.TQC.2020.11}.

\bibitemdeclare{inproceedings}{tdsscalarble}
\bibitem{tdsscalarble}
\bibinfo{author}{Titouan \surnamestart Carette\surnameend},
  \bibinfo{author}{Dominic \surnamestart Horsman\surnameend} \&
  \bibinfo{author}{Simon \surnamestart Perdrix\surnameend}
  (\bibinfo{year}{2019}): \emph{\bibinfo{title}{SZX-Calculus: Scalable
  Graphical Quantum Reasoning}}.
\newblock In: {\sl \bibinfo{booktitle}{44th International Symposium on
  Mathematical Foundations of Computer Science, {MFCS} 2019, August 26-30,
  2019, Aachen, Germany}}, pp. \bibinfo{pages}{55:1--55:15},
  \doi{10.4230/LIPIcs.MFCS.2019.55}.

\bibitemdeclare{article}{CoeckeDuncan}
\bibitem{CoeckeDuncan}
\bibinfo{author}{Bob \surnamestart Coecke\surnameend} \& \bibinfo{author}{Ross
  \surnamestart Duncan\surnameend} (\bibinfo{year}{2011}):
  \emph{\bibinfo{title}{Interacting quantum observables: categorical algebra
  and diagrammatics}}.
\newblock {\sl \bibinfo{journal}{New Journal of Physics}}
  \bibinfo{volume}{13}(\bibinfo{number}{4}), p. \bibinfo{pages}{043016}.
\newblock \urlprefix\url{http://stacks.iop.org/1367-2630/13/i=4/a=043016}.
\newblock \bibinfo{note}{\doi {10.1088/1367-2630/13/4/043016}}.

\bibitemdeclare{article}{Cowtan_2020}
\bibitem{Cowtan_2020}
\bibinfo{author}{Alexander \surnamestart Cowtan\surnameend},
  \bibinfo{author}{Silas \surnamestart Dilkes\surnameend},
  \bibinfo{author}{Ross \surnamestart Duncan\surnameend}, \bibinfo{author}{Will
  \surnamestart Simmons\surnameend} \& \bibinfo{author}{Seyon \surnamestart
  Sivarajah\surnameend} (\bibinfo{year}{2020}): \emph{\bibinfo{title}{Phase
  Gadget Synthesis for Shallow Circuits}}.
\newblock {\sl \bibinfo{journal}{Electronic Proceedings in Theoretical Computer
  Science}} \bibinfo{volume}{318}, p. \bibinfo{pages}{213?228},
  \doi{10.4204/eptcs.318.13}.

\bibitemdeclare{inproceedings}{amarngwang}
\bibitem{amarngwang}
\bibinfo{author}{Amar \surnamestart Hadzihasanovic\surnameend},
  \bibinfo{author}{Kang~Feng \surnamestart Ng\surnameend} \&
  \bibinfo{author}{Quanlong \surnamestart Wang\surnameend}
  (\bibinfo{year}{2018}): \emph{\bibinfo{title}{Two Complete Axiomatisations of
  Pure-State Qubit Quantum Computing}}.
\newblock In: {\sl \bibinfo{booktitle}{Proceedings of the 33rd Annual ACM/IEEE
  Symposium on Logic in Computer Science}}, \bibinfo{series}{LICS ?18}, p.
  \bibinfo{pages}{502?511}, \doi{10.1145/3209108.3209128}.

\bibitemdeclare{inproceedings}{jpvbeyondlics}
\bibitem{jpvbeyondlics}
\bibinfo{author}{Emmanuel \surnamestart Jeandel\surnameend},
  \bibinfo{author}{Simon \surnamestart Perdrix\surnameend} \&
  \bibinfo{author}{Renaud \surnamestart Vilmart\surnameend}
  (\bibinfo{year}{2018}): \emph{\bibinfo{title}{Diagrammatic Reasoning Beyond
  Clifford+T Quantum Mechanics}}.
\newblock In: {\sl \bibinfo{booktitle}{Proceedings of the 33rd Annual ACM/IEEE
  Symposium on Logic in Computer Science}}, \bibinfo{series}{LICS '18},
  \bibinfo{publisher}{ACM}, \bibinfo{address}{New York, NY, USA}, pp.
  \bibinfo{pages}{569--578}, \doi{10.1145/3209108.3209139}.
\newblock \urlprefix\url{http://doi.acm.org/10.1145/3209108.3209139}.

\bibitemdeclare{inproceedings}{jpvnormfmlics}
\bibitem{jpvnormfmlics}
\bibinfo{author}{Emmanuel \surnamestart Jeandel\surnameend},
  \bibinfo{author}{Simon \surnamestart Perdrix\surnameend} \&
  \bibinfo{author}{Renaud \surnamestart Vilmart\surnameend}
  (\bibinfo{year}{2019}): \emph{\bibinfo{title}{A Generic Normal Form for
  ZX-Diagrams and Application to the Rational Angle Completeness}}.
\newblock In: {\sl \bibinfo{booktitle}{34th Annual {ACM/IEEE} Symposium on
  Logic in Computer Science, {LICS} 2019, Vancouver, BC, Canada, June 24-27,
  2019}}, pp. \bibinfo{pages}{1--10}, \doi{10.1109/LICS.2019.8785754}.
\newblock \urlprefix\url{https://doi.org/10.1109/LICS.2019.8785754}.

\bibitemdeclare{article}{KissingerG20}
\bibitem{KissingerG20}
\bibinfo{author}{Aleks \surnamestart Kissinger\surnameend} \&
  \bibinfo{author}{Arianne~Meijer{-}van \surnamestart de~Griend\surnameend}
  (\bibinfo{year}{2020}): \emph{\bibinfo{title}{{CNOT} circuit extraction for
  topologically-constrained quantum memories}}.
\newblock {\sl \bibinfo{journal}{Quantum Inf. Comput.}}
  \bibinfo{volume}{20}(\bibinfo{number}{7{\&}8}), pp.
  \bibinfo{pages}{581--596}.

\bibitemdeclare{book}{Mermin}
\bibitem{Mermin}
\bibinfo{author}{N.~David \surnamestart Mermin\surnameend}
  (\bibinfo{year}{2007}): \emph{\bibinfo{title}{Quantum Computer Science: An
  Introduction}}.
\newblock \bibinfo{publisher}{Cambridge University Press}.

\bibitemdeclare{article}{bobanthonywang}
\bibitem{bobanthonywang}
\bibinfo{author}{Anthony \surnamestart Munson\surnameend}, \bibinfo{author}{Bob
  \surnamestart Coecke\surnameend} \& \bibinfo{author}{Quanlong \surnamestart
  Wang\surnameend} (\bibinfo{year}{2020}): \emph{\bibinfo{title}{AND-gates in
  ZX-calculus: spider nest identities and QBC-completeness}}.
\newblock {\sl \bibinfo{journal}{accepted to QPL 2020}}.
\newblock \bibinfo{note}{ArXiv:1910.06818}.

\bibitemdeclare{article}{ngwang}
\bibitem{ngwang}
\bibinfo{author}{Kang~Feng \surnamestart Ng\surnameend} \&
  \bibinfo{author}{Quanlong \surnamestart Wang\surnameend}
  (\bibinfo{year}{2017}): \emph{\bibinfo{title}{A universal completion of the
  {ZX}-calculus}}.
\newblock \bibinfo{note}{ArXiv:1706.09877}.

\bibitemdeclare{article}{ngwang2}
\bibitem{ngwang2}
\bibinfo{author}{Kang~Feng \surnamestart Ng\surnameend} \&
  \bibinfo{author}{Quanlong \surnamestart Wang\surnameend}
  (\bibinfo{year}{2018}): \emph{\bibinfo{title}{Completeness of the
  {ZX}-calculus for Pure Qubit Clifford+T Quantum Mechanics}}.
\newblock
  \bibinfo{note}{\href{http://arxiv.org/abs/1801.07993}{arXiv:1801.07993}}.

\bibitemdeclare{book}{Nielsen}
\bibitem{Nielsen}
\bibinfo{author}{Michael~A. \surnamestart Nielsen\surnameend} \&
  \bibinfo{author}{Isaac~L. \surnamestart Chuang\surnameend}
  (\bibinfo{year}{2010}): \emph{\bibinfo{title}{{Quantum Computation and
  Quantum Information}}}.
\newblock \bibinfo{publisher}{Cambridge University Press},
  \bibinfo{address}{Cambridge}, \doi{10.1017/CBO9780511976667}.
\newblock \bibinfo{note}{\doi{ 10.1017/CBO9780511976667}}.

\bibitemdeclare{inproceedings}{Renaudprulelics}
\bibitem{Renaudprulelics}
\bibinfo{author}{Renaud \surnamestart Vilmart\surnameend}
  (\bibinfo{year}{2019}): \emph{\bibinfo{title}{A Near-Minimal Axiomatisation
  of ZX-Calculus for Pure Qubit Quantum Mechanics}}.
\newblock In: {\sl \bibinfo{booktitle}{34th Annual {ACM/IEEE} Symposium on
  Logic in Computer Science, {LICS} 2019, Vancouver, BC, Canada, June 24-27,
  2019}}, pp. \bibinfo{pages}{1--10}, \doi{10.1109/LICS.2019.8785765}.
\newblock \urlprefix\url{https://doi.org/10.1109/LICS.2019.8785765}.

\bibitemdeclare{article}{wanglinearstm}
\bibitem{wanglinearstm}
\bibinfo{author}{Quanlong \surnamestart Wang\surnameend}:
  \emph{\bibinfo{title}{Solving linear systems in ZX-calculus}}.
\newblock {\sl \bibinfo{journal}{In preparation}}.

\bibitemdeclare{article}{qwangrsmring}
\bibitem{qwangrsmring}
\bibinfo{author}{Quanlong \surnamestart Wang\surnameend}
  (\bibinfo{year}{2019}): \emph{\bibinfo{title}{ZX-calculus over arbitrary
  commutative rings and semirings (extended abstract)}}.
\newblock \bibinfo{note}{ArXiv:1912.01003}.

\bibitemdeclare{article}{wangalg2020}
\bibitem{wangalg2020}
\bibinfo{author}{Quanlong \surnamestart Wang\surnameend}
  (\bibinfo{year}{2020}): \emph{\bibinfo{title}{An algebraic axiomatisation of
  ZX-calculus}}.
\newblock {\sl \bibinfo{journal}{Proceedings of the 17th International
  Conference on Quantum Physics and Logic (QPL) 2020}}.
\newblock \bibinfo{note}{ArXiv:1911.06752}.

\bibitemdeclare{book}{qwangslides}
\bibitem{qwangslides}
\bibinfo{author}{Quanlong \surnamestart Wang\surnameend}
  (\bibinfo{year}{2020}): \emph{\bibinfo{title}{Enter a visual era: process
  theory embodied in ZX-calculus}}.
\newblock \bibinfo{publisher}{Presentation at 17th International Conference on
  Quantum Physics and Logic (QPL)}, \doi{10.13140/RG.2.2.17289.67682}.

\end{thebibliography}





\end{document}